\documentclass[a4paper,english,11pt]{scrartcl}
\usepackage{authblk}
\usepackage{hyperref}
\usepackage{amsthm}
\usepackage{graphicx}
\usepackage{url}
\usepackage{relsize}
\usepackage{amsfonts}
\usepackage{booktabs} 
\usepackage{comment}
\usepackage[normalem]{ulem}
\usepackage{empheq}
\usepackage{mathtools,adjustbox}				
\usepackage{amssymb,subfigure}	
\usepackage{mathrsfs, stmaryrd}		
\usepackage{mathtools}
\usepackage{multirow}
\usepackage[margin=1in]{geometry}
\usepackage{enumitem,framed}
\usepackage{moresize}
\usepackage{tikz}
\usetikzlibrary{graphs}
\usepackage{pgfplots}
\usepackage{graphicx}
\usepgfplotslibrary{statistics} 
\usetikzlibrary{pgfplots.statistics}
\pgfplotsset{compat=1.10}
\usepgfplotslibrary{fillbetween}
\usetikzlibrary{arrows.meta}
\tikzstyle{res}=[draw=black, fill=white, minimum size=1.3em, font={\tiny}]
\newcommand{\marksize}{1pt}
\newcommand{\mysize}{0.7cm}
\newcommand{\myscale}{0.78}
\newcommand{\myoscale}{1.2}
\newcommand{\heightbp}{6.5cm}
\newcommand{\scalebp}{1}
\usepackage{pgfplots}
\usepackage{graphicx}

\tikzstyle{res}=[draw=black, fill=white, minimum size=1.3em, font={\tiny}]
\tikzstyle{node}=[circle,draw=black, fill=black, minimum size=0.3em, inner sep=0pt, font={\tiny}]
\usepackage{sgame}
\usepackage{xcolor}

\usetikzlibrary{chains,shapes.multipart}
\usetikzlibrary{shapes,calc}
\usetikzlibrary{automata,positioning}
\usetikzlibrary{arrows, decorations.markings}
\definecolor{myred}{RGB}{220,43,25}
\definecolor{mygreen}{RGB}{0,146,64}
\definecolor{myblue}{RGB}{0,143,224}
\definecolor{darkgreen}{rgb}{0.0, 0.5, 0.0}
\newcommand{\Cross}{$\mathbin{\tikz [x=1.4ex,y=1.4ex,line width=.2ex, red] \draw (0,0) -- (1,1) (0,1) -- (1,0);}$}%
\def\mybig#1{{\hbox{$\left#1\vbox to23\p@{}\right.\n@space$}}}

\newcommand{\Z}{\mathbb{Z}}

\newcommand{\convk}{\mathrm{conv}^k}
\newcommand{\res}[2]{\mathrm{res}\big(#1,#2\big)}
\newcommand{\dom}{\mathrm{rdom}}
\newcommand{\pen}{\mathrm{pen}}

\newcommand{\ext}{\mathrm{conv}}
\newcommand{\extt}{\mathrm{ext}}
\newcommand{\Pre}[1]{\mathrm{Pre}_{#1}}
\newcommand{\Str}[1]{\mathcal{S}_{#1}}
\newcommand{\Stra}{\mathcal{S}}
\usepackage{array}				
\usepackage{framed,bm}

\newcommand{\conv}{\mathrm{conv}}
\newcommand{\N}{\mathbb{N}}	

\newcommand{\R}{\mathbb{R}}	

\usepackage{etoolbox}
\patchcmd{\SetTagPlusEndMark}{}{}{}{}
\usepackage{etex}

\usepackage[textsize=tiny,textwidth=2cm,shadow ]{todonotes}

\newcommand{\epi}{\textrm{epi}}

\title{Generalized Nash Equilibrium Problems\\ with Mixed-Integer Variables\thanks{A one-page abstract  appeared in the  \emph{Proceedings of the 17th Conference on Web and Internet Economics, 2021}~\cite{WINEGNEP}}}

\newcommand{\myparagraph}[1]{\paragraph{#1}}
\newcommand{\mycup}{\mathop{{\textstyle\bigcup}}}
\newcommand{\mtilde}[1]{\tilde{#1}}
\newcommand{\myhspace}{\hspace{0cm}}

\newtheorem{theorem}{Theorem}
\newtheorem{proposition}{Proposition}
\newtheorem{definition}{Definition}
\newtheorem{lemma}{Lemma}
 \newtheorem{corollary}{Corollary}
\newtheorem{example}{Example}

\author{Tobias Harks and Julian Schwarz}

\affil{\small University of Passau, Faculty of Computer Science and Mathematics, 94032 Passau\\
\href{mailto:julian.schwarz@uni-passau.de}{\{\texttt{tobias.harks,julian.schwarz\}@uni-passau.de}}}

\begin{document}

\maketitle

\begin{abstract}
    We consider generalized Nash equilibrium problems (GNEPs)
    with non-convex strategy spaces and non-convex cost functions.
    This general class of games includes  the important case of games with mixed-integer variables for which only a few results are known in the literature. We present a new approach to characterize equilibria
    via a convexification technique using the Nikaido-Isoda function. To any given instance  of the GNEP, we construct a set of  convexified instances and show that  a feasible strategy profile is an equilibrium for the original instance if and only if it is an equilibrium for any convexified instance and the convexified cost functions coincide with the initial ones. 
    We develop this {convexification approach} along three  dimensions:
   We first show that for \emph{quasi-linear} models, where a convexified instance exists in which  for fixed strategies of the opponent players, the cost function of every player is linear and the respective strategy space is polyhedral, 
  the  convexification reduces the GNEP to a standard (non-linear) optimization problem.
  Secondly, we derive two complete characterizations 
  of those  GNEPs for which the convexification leads to a
  jointly constrained or a jointly convex GNEP, respectively. 
  These characterizations require new concepts related to the interplay of the convex hull operator applied to restricted
  subsets of feasible strategies and may be interesting on their own. {Note that this characterization is also computationally relevant as jointly convex GNEPs
  have been extensively studied in the literature.}
Finally, we demonstrate the applicability of our  results by presenting a numerical study {regarding the computation of equilibria for three classes of GNEPs related to integral network flows and discrete market equilibria.}

\end{abstract} 


\section{Introduction}
The generalized Nash equilibrium problem  constitutes a fundamental class of noncooperative
games with applications in economics~\cite{Debreu54}, transport systems~\cite{Beckmann56} and electricity markets~\cite{Anderson13}.
The differentiating feature of GNEPs compared to classical games
is the flexibility to model dependencies among the strategy spaces of players, that is, 
the individual strategy space of every player depends on the strategies chosen by the rival players.
Examples in which this aspect is crucial appear  for instance in market games where
discrete goods are traded and the buyers have hard spending budgets: effectively, the strategy space of a buyer depends on the market price (set by the seller) as only those bundles of goods remain affordable that fit into the budget. Other examples appear in transportation systems, where joint capacities (e.g. road-, production- or storage capacity)  constrain the strategy space of a player. 
For further applications of the GNEP and an overview of the general theory, we refer to the excellent survey articles of Facchinei and Kanzow~\cite{FacchineiK10} and Fischer et al.~\cite{FISCHER2014}.

While the GNEP is a research topic with constantly increasing interest, the majority of work is concerned with the continuous and convex GNEP, i.e., instances of the GNEP where the strategy sets of players are convex and the cost functions are at least continuous with regard to all players' strategy choices and often convex in the own strategy choice. 
Our focus in this paper is to derive insights into non-convex or discrete GNEPs including GNEPs with mixed-integer variables. As main approach, we  will reformulate a non-convex GNEP via a convexification technique and then identify expressive subclasses of GNEPs which can then be solved by standard optimization problems or be reformulated via more structured convexified GNEPs.
  
Let us introduce the model formally and first recap the standard pure Nash equilibrium problem (NEP).
For an integer $k \in \N$, let $[k] := \{1,\dots,k\}$. Let $N = [n]$ be a finite set of players.
 Each player $i \in N $ controls the variables $x_i\in X_i\subseteq\R^{k_i}$. We call $x = (x_1,\dots, x_n)$ with $x_i \in X_i$ for all $i \in N$ a strategy profile and $X = X_1 \times \dots \times X_n\subseteq \R^k$ the strategy space, where 
 $k := (k_1,\ldots,k_n)$ and $\R^{(l_1,\ldots,l_s)} := \R^{\sum_{i=1}^s l_i}$ for any vector $l \in \N^s$ and $s \in \N$. We use standard game theory notation; for a strategy profile $x \in X$, we write $x = (x_i, x_{-i})$ meaning that $ x_i$ is the strategy that player~$i$ plays in $x$ and $x_{-i}$ is the partial strategy profile of all players except $i$.
The private cost of player~$i\in N$ in strategy profile $x \in X$ is defined by a function
$\pi_i:\R^k\rightarrow\R, x\mapsto \pi_i(x)$.
A (pure) Nash equilibrium is a strategy profile $x^*\in X$ with
\[ \pi_i(x^*)\leq \pi_i(y_i, x_{-i}^*) \text{ for all }y_i\in X_i,\; i\in N.\]
The GNEP generalizes the model by allowing that the strategy sets of every player
may depend on the rival players’ strategies. More precisely, for any $x_{-i}\in \R^{k_{-i}}$
(using the notation $k_{-i}:= (k_j)_{j \neq i}$), 
there is a feasible strategy set $X_i(x_{-i})\subseteq\R^{k_{i}}$.
In this regard, one can think of the strategy space of player $i\in N$ represented by a set-valued mapping $X_i:\R^{k_{-i}}\rightrightarrows \R^{k_{i}}$.
This leads to the notation of the combined strategy space represented by a mapping $X:\R^{k}\rightrightarrows \R^{k}, x \mapsto \prod_{i \in N} X_i(x_{-i})$ with $y\in X(x)\Leftrightarrow y_i\in X_i(x_{-i})$ for all $i\in N$.
The private  cost function  is given by $\pi_i: \mathbb{R}^k \to \mathbb{R}$ for every player $i\in N$. The problem of player $i\in N$ -- given the rivals' strategies $x_{-i}$ -- is to solve the following minimization problem:
\begin{equation}\label{eq: PlayerOpt}
\textstyle {\inf}_{y_i} \pi_i(y_i, x_{-i}) \text{ s.t.: } y_i\in X_i(x_{-i}).
\end{equation} 
A generalized Nash equilibrium (GNE) is a feasible strategy profile $x^*\in X(x^*)$ with
\[ \pi_i(x^*)\leq \pi_i(y_i, x^*_{-i}) \text{ for all }y_i\in X_i(x^*_{-i}),\; i\in N.\]
We can compactly represent a GNEP  by the tuple $I = (N,(X_i(\cdot))_{i\in N},(\pi_i)_{i \in N})$. 
In the sequel of this paper,  we will heavily use the Nikaido-Isoda function (short: NI-function), see \cite{nikaido1955}. 
\begin{definition}[NI-function]
Let an instance $I=(N,(X_i(\cdot))_{i\in N},(\pi_i)_{i \in N})$  of the GNEP be given.
For any two vectors $x,y \in \mathbb{R}^k$, the NI-function is defined as:
\begin{align*}
\textstyle\Psi(x,y) := \sum_{i \in N} \left[ \pi_i(x)-\pi_i(y_i,x_{-i}) \right].
\end{align*}
\end{definition}
By defining $\hat{V}(x) := \sup_{y \in X(x)} \Psi(x,y)$ we can recap the following well-known characterization of a generalized Nash equilibrium, see for instance Facchinei and Kanzow~\cite{FacchineiK10}. 
\clearpage

\begin{theorem}\label{CharGNE}
For an instance $I$ of the GNEP the following statements are equivalent.
\begin{itemize}
\item[1.] $x$ is a generalized Nash equilibrium for $I$.
\item[2.] $x \in X(x)$ and $\hat{V}(x) = 0$.
\item[3.] {$x$ is an optimal solution of $\inf_{x \in X(x)}\hat{V}(x)$ with value zero.}   
\end{itemize}
\end{theorem}
This characterization does not rely on any convexity assumptions 
on the strategy spaces nor on the private cost functions of the players.
Yet, the characterization seems computationally of limited interest as neither the Nikaido-Isoda function itself 
nor the fixed-point condition $x \in X(x)$  seems computationally tractable.

\subsection{Our Results and Organization of the Paper}
Our approach relies on a convexification technique applied to the original non-convex game
leading to a new characterization of the existence of Nash equilibria for GNEPs. 
In particular, we derive for any instance $I$ of the GNEP a set of convexified instances $\mathcal{I}^\conv$. Roughly speaking, the latter set consists of all those instances $I^\conv =(N,(X_i^\conv(\cdot))_{i \in N}, (\phi_i)_{i \in N})$, where for all players $i \in N$ and rivals strategies $x_{-i}$ contained in a certain subset of $\R^{k_{-i}}$, 
 the convexified strategy space $X_i^\conv(x_{-i})$ is given by the convex hull $\conv(X_i(x_{-i}))$ of the original strategy space and  the convexified private cost function $x_i \mapsto \phi_i(x_i,x_{-i})$ is the convex envelope of $x_i \mapsto \pi_i(x_i,x_{-i})$.
Our main result (Theorem~\ref{thm:main}) states that for any $I^\conv \in \mathcal{I}^\conv$, a strategy profile 
$x\in X(x)$ is a GNE for $I$ if and only if it is a GNE for $I^{\conv}$ and the convexified cost functions coincide with the original ones. The proof is based on using the Nikaido-Isoda functions for both games $I^{\conv}$ and $I$.
While the convexified {instances} may admit an equilibrium under certain circumstances, this equilibrium might still not be feasible for the original non-convex game. The advantage of our convex reformulation, however, lies in
the possibility that for some problems, it is computationally tractable to solve a convexified instance while preserving feasibility with respect
to the original game. 
In this regard, we study several
subclasses of GNEPs for which this methodology applies.

  In Section~\ref{sec:darstellung}, we consider
\emph{quasi-linear} GNEPs in which the cost functions of players are quasi-linear and the players' strategy spaces are quasi-polyhedral sets, that is, 
they admit a structure which allows  the convexified private cost functions to be chosen linearly 
for fixed
strategies of the other players. Similarly, the convexified strategy sets can be described by polyhedra whenever the rivals' strategies are fixed.
{By reformulating the $\hat{V}$ function of an associated convexified instance, 
we show in Theorem~\ref{theorem:darstellung} that a quasi-linear GNEP can be modeled as a standard (non-linear) optimization problem. The reformulation uses linear duality of the players' optimization problems and we note that this approach has been used before by Stein and Sudermann-Merx~\cite{stein2016cone} for the special case of linear GNEPs in which the cost functions and strategy sets can be described by linear functions in $x$. }

In Section~\ref{sec: JoinCons}, we study
\emph{jointly constrained} GNEPs which are also called GNEPs with \emph{shared constraints}.  These games have the differentiating feature that the players' strategy sets are  restricted via a shared feasible set $X\subseteq\R^k$. If $X$ is convex, one speaks of a jointly convex GNEP but we do not impose this on $X$ a priori. 
GNEPs with shared constraints have been extensively studied in the literature 
and in this regard we analyze the structure of
original non-convex GNEPs $I$ (not even jointly constrained) for which the set of convexified instances $\mathcal{I}^\conv$ contains a \emph{jointly constrained} or even a \emph{jointly convex} instance.
To this end, we introduce the new classes of \emph{$k$-restrictive-closed} and \emph{restrictive-closed} GNEPs for which we show that they completely characterize whether or not $\mathcal{I}^\conv$ contains a jointly constrained or a jointly convex instance, respectively.  
The property of ($k$-)restrictive-closedness is for example fulfilled for all $\{0,1\}^k$ jointly constrained instances $I$
and thus admits interesting applications.

In Section \ref{sec: Comp}, we present numerical results on the computation of equilibria for 
{three classes of GNEPs related to integral network flows and discrete market equilibria }
which are shown to belong to both classes of restrictive-closed and quasi-linear GNEPs. To find equilibria of an instance $I$,
{ we propose two different different methods based on our convexification result. 
Firstly, we present an approach where our quasi-linear reformulation is plugged into a standard non-convex solver (\textsc{BARON}). 
Secondly, we try to compute an integral GNE of a specific convexified instance $I^\conv \in \mathcal{I}$.  
by implementing different procedures
from the literature for solving a convex GNEP, 
enhanced by a simple rounding procedure in order to obtain an integral equilibrium.
}
Perhaps surprisingly, it turned out that our quasi-linear approach was not only faster (on average) in finding specifically integral GNE
for the original non-convex GNEP but also for computing (not necessarily integral) GNE for the convexified instances.

\subsection{Related Work} 
\myparagraph{Continuous and Convex GNEPs.} GNEPs have been studied intensively in terms of equilibrium existence and numerical
algorithms. It is fair to say, that the majority of works focus on the continuous and convex
case, that is, the cost functions of players are convex (or at least continuous) 
and the strategy spaces are convex.  One major reason for these restrictive assumptions
lies in the lack of tools to prove existence of equilibria. Indeed, most existence
results rely on an application of Kakutani's fixed point theorem which in turn
requires those convexity assumptions (e.g.~Rosen~\cite{Rosen65}).
We refer to the survey articles of Facchinei and Kanzow~\cite{FacchineiK10} and Fischer et al.~\cite{FISCHER2014} for an overview of the general theory. 

{We discuss in the following various approaches for computing GNE for convex and continuous GNEPs. }
Based on reducing the GNEP to the standard NEP, Facchinei and Sagratella~\cite{FacchineiS11} described an algorithm to compute all solutions of a jointly convex GNEP, where the joint restrictions are given by linear equality constraints. 
However, this algorithm does not terminate in finite time whenever there are infinitely many equilibria. 
Dreves~\cite{Dreves14,Dreves17} tackled this problem via an algorithm which computes in finite time the whole solution set for two different types of GNEPs. In \cite{Dreves14}, he investigated affine GNEPs with one-dimensional strategy sets in which the players' optimization problems are convex quadratic problems with a common linear constraint in $x$. 
{
The other type of GNEPs considered by Dreves~\cite{Dreves17} are the linear (not necessarily jointly convex) GNEPs introduced by Stein and Sudermann-Merx~\cite{stein2016cone}. While Dreves investigated the computation of all solutions, Stein and Sudermann-Merx studied the smoothness of a 
certain gap function that arises via a suitable extension of the $\hat{V}$ function. The latter extension is based on a dualization approach regarding
the second part of the NI-function, allowing for a reformulation of   $\hat{V}(x) = \sup_{y \in X(x)}\Psi(x,y)$  as a minimization problem. 
Note that this dualization step will also play a key role in our analysis of quasi-linear GNEPs in Section~\ref{sec:darstellung}. 
The applicability of the findings of Stein and Sudermann-Merx was demonstrated in~\cite{dreves2016solving}  by Dreves and Sudermann-Merx where they
investigated various numerical methods to compute equilibria of linear GNEPs, cf.~also~\cite{stein2018noncooperative}.} 
Returning to the jointly convex GNEP,  Heusinger and Kanzow~\cite{Heusinger2009-2} presented an optimization reformulation using the Nikaido-Isoda function, assuming that the cost functions $\pi_i(x_i,x_{-i})$ of the players are (at least) continuous in $x$ and convex in $x_i$. 
Under the same assumptions concerning the cost functions, Dreves, Kanzow and Stein \cite{Dreves_Kanzow_Stein} generalized this approach to \emph{player-convex} GNEPs, where additionally to the assumptions on the cost functions, the strategy sets are assumed to be described by $X_i(x_{-i}) = \{x_i\mid g_i(x_i,x_{-i})\leq 0\}$ for a restriction function $g_i$ which is (at least) continuous in $x$ and convex in $x_i$.
In comparison to this optimization reformulation, Dreves et al.~\cite{DrevesFKS11} took a different approach to finding equilibria via the KKT conditions of the GNEP. Under sufficient regularity, e.g.~$C^2$ cost- and restriction functions, they discuss how the KKT system of the GNEP may be solved in order to find generalized Nash equilibria. 

While the assumptions concerning the cost- and restriction functions in the above papers 
are  mild in the context of continuous GNEPs, it is a priori not clear, whether or not there exists a convexified instance in $\mathcal{I}^\conv$ which fulfills them,
and then allows for the application of algorithms from the domain of convex and continuous GNEPs. In this regard, we are concerned {in Section~\ref{sec: JoinCons}} with identifying subclasses of GNEPs
which guarantee the existence of such well-behaved convexified instances in  $\mathcal{I}^\conv$. 
\myparagraph{Non-Convex and Discrete GNEPs.}
For non-convex and discrete GNEPs, much less is known regarding the existence and computability of equilibria. 
{In fact, the only computational approach for finding pure GNE we are aware of {are} that of }
Sagratella~\cite{Sagratella2017AlgorithmsFG,sagratella2019generalized}. {In the former,} two different techniques for the subclass of so called \emph{generalized potential games} with mixed-integer variables {are presented}.
Similar to the jointly convex GNEPs, in these potential games, the players are restricted through a common convex set $X$ with the further restriction that some strategy components need to be integral and there is a potential function over the set $X$.
On the one hand, Sagratella introduced certain optimization problems with mixed-integer variables based on the fact that minimizers of the potential function correspond to a subset of generalized Nash equilibria. On the other hand, he showed that a Gauss-Seidel best-response (BR) algorithm may approximate equilibria arbitrary well within a finite amount of steps in this setting. 
{We remark, however, that a BR-algorithm is not a correct approach for GNEPs not admitting a potential function, because there are trivial examples (even for standard NEPs) in which a BR-algorithm may cycle forever and not terminate although equilibria exist.
In particular,  for  GNEPs that are not jointly constrained, the best response of a player may lead to an infeasible overall strategy profile resulting in  an empty best-response correspondence for some player. 
In this regard, Sagratella's results as well as several interesting models that have emerged based upon his results in the domain of automated driving \cite{FabianiAD20}, traffic control \cite{cenedese2021highway} or transportation problems \cite{SagratellaNFCTP}
are not directly extendable to the general mixed-integer GNEP setting that we consider in this paper.}
{
In~\cite{sagratella2019generalized}, Sagratella generalized
his approach in~\cite{Sagratella17} for standard NEPs (c.f.~below) and considered mixed-integer GNEPs in which each player's strategy set has partial integer constraints and depends on the other players' strategies via a linear constraint in her own strategy and the other players' integrally constrained strategies. Under further convexity and continuity assumptions regarding the cost functions, he then proposed a branch and bound method based on merit functions as well as a branch and bound method exploiting dominance of strategies for suitable cuts. }

Besides the work of Sagratella, we are only aware of the paper~\cite{AnandutaGrammatico2021} by Ananduta and Grammatico which deals with mixed-integer GNEPs. {They considered a model in which the dependency of costs and constraints can be additively separated between the continuous and integer variables. However, the authors did not  directly deal with the mixed-integer GNEP but rather with a mixed-strategy extension of the game.}
That is, the players are assumed to choose probability distributions over their integral strategies and the game is solved using mixed-strategy GNEP without integrality constraints. Based upon the specific structure of the latter, the authors then introduced a Bregman forward-reflected-backward splitting and design distributed algorithm to compute equilibria. 
{In this paper, we focus on pure GNE, partly because mixed or correlated strategies have no
meaningful physical interpretation in some games; see also the
discussion in Osborne and Rubinstein~\cite[\S~3.2]{Osborne94} about
critics of mixed Nash equilibria.
In particular, the definition of  a meaningful randomization concept for general GNEPs is non-trivial. This is, for instance, illustrated by 
 the special class of separable GNEPs considered in~\cite{AnandutaGrammatico2021}, where the proposed concept of mixed equilibria may associate a non-zero probability to strategy profiles which are not even feasible. }

While mixed-integer GNEPs are fairly unexplored,
there is a growing research  regarding general formulations of  non-convex (standard) NEPs. 
For instance, the so-called Integer Programming Games (IPGs) have recently gained interest in which the optimization problem of each player consists of minimizing a continuous function over a (fixed) polyhedron with partial integrality constraints. 
IPGs in a less general setting were first introduced by Köppe et al.~\cite{Koeppe11}, where they investigated the computational complexity of
computing pure NE.  
Along these lines, Carvalho et al.~\cite{carvalho2018existence,carvalho2022computing} investigated the existence~\cite{carvalho2018existence} and computation~\cite{carvalho2022computing} of NE {as well as recently published a tutorial on the computation of NE in IPGs in~\cite{carvalho2023integer}. Extending the Sample generation
method  presented in~\cite{carvalho2022computing} for IPGs,
Crönert and Minner~\cite{cronert2022equilibrium} proposed an algorithm for the complete enumeration of all mixed equilibria in finite games, i.e.~games belonging to the standard NEP with strategy sets being finite and succinctly described by a finite amount of inequalities.   }  
Regarding more specialized subclasses of IPGs, Del Pia et al.~\cite{PiaFM17} introduced a strongly polynomial-time algorithm to compute NE and  derived several related complexity results for IPGs in which each player has a totally unimodular strategy set. On the basis of aggregating the players strategy spaces, this approach has been extended by Kleer and Schäfer~\cite{KleerS17}.  
Another subclass of IPGs was explored in~\cite{guo2021copositive} where the authors focus on games in which each player solves a mixed-binary quadratic program where the cost function is quadratic in the whole strategy profile. Leveraging a completely positive program reformulation established by Burer~\cite{burer2009copositive}, the authors tackled the computation of NE via the associated KKT conditions.
For IPGs in which the strategy spaces are boxes, Kirst, Schwarz and Stein~\cite{kirst2022branch} recently proposed a branch-and-bound algorithm that computes the set of all approximate equilibria for a given approximation error based on discarding rules that determine sets which do not contain equilibria.  
Finally regarding IPGs, Dragotto and Scatamacchia~\cite{dragotto2023zero} presented an algorithm based on cutting planes to tackle the computation, enumeration and selection of pure NE in IPGs with purely integral strategies.

Departing from IPGs, Carvalho et al.~\cite{carvalho2021cut} introduced the class of Reciprocally-Bilinear Games (RBGs), where for each player, the cost function is bilinear in her own and her rivals' strategies while her strategy set is only required to have a convex hull whose closure is polyhedral.
A  Cut-and-Play algorithm is then presented which computes a mixed Nash equilibrium of an RBG based upon a scheme consisting of mainly two components: the computation of mixed Nash equilibria for approximated games and the  eventual refinement of the approximations.

Finally, another class of non-convex NEPs was investigated by Sagratella~\cite{Sagratella16} 
where a  branching method is presented to compute all NE for NEPs where the cost function of a player is convex in her own strategy and continuous in the complete strategy profile while her strategy set is given by a convex restriction function combined with integrality constraints for all strategy components. 
Recently, Schwarze and Stein~\cite{schwarze2022branch} were able to drop the latter convexity assumption of players' cost functions by  extending Sagratella's   approach via a branch-and-prune algorithm.

\section{Convexification}
In this section we will introduce for any instance $I=(N,(X_i(\cdot))_{i\in N},(\pi_i)_{i \in N})$ of the GNEP a corresponding set of \emph{convexified instances} $\mathcal{I}^\conv$ of the GNEP. 
In order to describe this convexification method, we introduce the following concept of a \emph{refined domain}: \begin{definition}
 Let $I=(N,(X_i(\cdot))_{i\in N},(\pi_i)_{i \in N})$ be an instance of the GNEP. For any $i \in N$ we define the \emph{refined domain} of the set-valued mapping $X_i: \R^{k_{-i}} \rightrightarrows \R^{k_i}$ by
 {
 \begin{align*}
     \mathrm{rdom}X_i  := \left\{x_{-i} \in \R^{k_{-i}} \mid \exists\, \mtilde{x}_i \in \R^{k_i}: (\mtilde{x}_i,x_{-i}) \in X\big((\mtilde{x}_i,x_{-i})\big) \right\}. 
 \end{align*}}
\end{definition} 
This concept is a refinement of the standard domain $\mathrm{dom}X_i := \{x_{-i} \in \R^{k_{-i}} \mid X_i(x_{-i}) \neq \emptyset\}$. Clearly, the standard domain of $X_i$ contains the refined domain. However, the refined domain is in general a proper subset as it also takes into account whether the rivals' strategies $x_{-i}$ are feasible or not. That is, only the rivals' strategies $x_{-i}$ are in the refinement where at least one feasible strategy $\mtilde{x}_i \in X_i(x_{-i})$ for player $i$ exists such that for any of her rivals $j\neq i$, the strategy $x_j \in X_j\big((\mtilde{x}_i,x_{-ij})\big)$ is feasible. Here, $(\mtilde{x}_i,x_{-ij})\in\R^{k_{-j}}$ is the partial strategy profile of all players except $j$ in which player $l$ plays strategy $x_l$ for $l \notin \{i,j\}$ and $\mtilde{x}_i$ for $l = i$. 
{Alternatively, the refined domain of player $i$ can be seen as the projection of the set of feasible strategy profiles $x \in X(x)$ to the rivals' strategy space $\R^{k_{-i}}$.}
This idea of relevant strategies leads to the following definition of what we call \emph{quasi-isomorphic} instances of the GNEP. 
\begin{definition}
Two instances $I=(N,(X_i(\cdot))_{i\in N},(\pi_i)_{i \in N})$ and $I'=(N,(X_i'(\cdot))_{i\in N},(\pi'_i)_{i \in N})$ are called \emph{quasi-isomorphic}, if for all $i \in N$ the refined domains coincide $\dom X_i = \dom X_i'$ and for all $x_{-i} \in \dom X_i$ the strategy sets and cost functions coincide, i.e.~$X_i(x_{-i}) = X_i'(x_{-i})$ and $\pi_i(x_i,x_{-i}) = \pi'_i(x_i,x_{-i})$ for all $x_i \in X_i(x_{-i})$.
\end{definition}
The above concept of quasi-isomorphic GNEPs $I,I'$  requires that the set of feasible strategy profiles of both games coincide and for each feasible strategy profile $x\in X(x)$ of $I$ the conditions for it to be a GNE in $I$ are precisely the same conditions as in $I'$ and vice versa.

{
Our convexification approach relies heavily on the concept of convex envelopes which we introduce next, cf. Rockafellar~\cite{Rockafellar70}. 
\begin{definition}
    Let $f:M\to \R\cup\{-\infty\}$ with $M\subseteq \R^l$ for some $l \in \N$ be a function into the extended reals.
    We denote by $\epi(f):=\{(x,y)\in M \times \R\mid y\geq f(x)\}$ its epigraph and call $\phi:\conv(M) \to \R\cup\{-\infty\}$ with 
    $\phi(x) := \inf\{ y\in \R \mid (x,y) \in \conv(\epi(f))\}\in \R\cup\{-\infty\}$ the convex envelope of $f$.  
\end{definition}
Note that Rockafellar~\cite[Theorem 5.3]{Rockafellar70} showed that the convex envelope
is a well-defined convex function (which may attain the value $-\infty$).  
Based on this definition, we can now describe our convexification method.}
\begin{definition}\label{def: conv}
 Let $I=(N,(X_i(\cdot))_{i\in N},(\pi_i)_{i \in N})$ and $I^\conv=(N,(X_i^\conv(\cdot))_{i\in N},(\phi_i)_{i \in N})$   be two instances of the GNEP. We call  $I^\conv$ a \emph{convexified instance} for $I$, if it fulfills for all $i \in N$ and $x_{-i}\in \dom X_i$ the following two criteria: 
 \begin{enumerate}
     \item $X_i^\conv(x_{-i}) = \conv(X_i(x_{-i}))$, i.e.~the convexified strategy space is given by  the convex hull of the original strategy space.\label{enu: c1}
     \item $\phi_i(x_i,x_{-i})= \phi_i^{x_{-i}}(x_i)$ for all $x_i \in \conv(X_i(x_{-i}))$ where $\phi_i^{x_{-i}}(\cdot):\conv(X_i(x_{-i}))\to\mathbb{R}$ denotes the convex envelope  of  $\pi_i(\cdot,x_{-i}):X_i(x_{-i})\to\mathbb{R}$, $ x_i \mapsto \pi_i(x_i,x_{-i})$.\label{enu: c2}
 \end{enumerate}
Note that \ref{enu: c1}.~is only required to hold for $x_{-i} \in \dom X_i$ and not the whole $\R^{k_{-i}}$. Similarly, \ref{enu: c2}.~restricts the cost function of a convexified instance only for $x$  with $x_i \in \conv(X_i(x_{-i})), x_{-i} \in \dom X_i$ and not the whole $\R^k$.
This degree of freedom leads to a whole set of convexified instances for $I$ which we denote by $\mathcal{I}^\conv(I)$. If the instance $I$ is clear from context, we will just speak of convexified instances and $\mathcal{I}^\conv$. 
\end{definition} 
{By definition,
the cost function of an instance of the GNEP need to be real-valued. In this regard,
the set $\mathcal{I}^\conv(I) \neq \emptyset$ is nonempty if and only if  
for all $i \in N$ and $x_{-i}\in \dom X_i$, 
the function $\pi_i(\cdot, x_{-i}):X_i(x_{-i}) \to \R$ admits a \emph{real-valued} convex envelope ${\phi_i^{x_{-i}}:\conv(X_i(x_{-i}))\to \R}$.
It thus follows immediately that the boundedness of~\eqref{eq: PlayerOpt} is sufficient for $\mathcal{I}^\conv(I)\neq \emptyset$. However, we remark that it is not necessary. In Section~\ref{sec:darstellung} it is shown 
that all player-linear mixed-integer GNEPs (cf.~Definition~\ref{def:lmi}) fulfill $\mathcal{I}^\conv(I)\neq \emptyset$, even though there may exist players with unbounded~\eqref{eq: PlayerOpt}. 

With our convexification method at hand, we can now describe our first main theorem.
}

\begin{theorem}\label{thm:main} 
Let $I=(N,(X_i(\cdot))_{i\in N},(\pi_i)_{i \in N})$ be an instance of the GNEP and $I^\conv \in \mathcal{I}^\conv$ any convexified instance.
For any $x \in X(x)$,
the following assertions are equivalent.
\begin{enumerate}
\item $x$ is a  generalized Nash equilibrium for $I$. \label{enu: mthm1}
\item $x$ is a generalized Nash equilibrium for $I^\conv$ and $\phi_i(x) = \pi_i(x)$ for all $i \in N$.\label{enu: mthm2}
\end{enumerate}
\end{theorem}
\begin{proof}
Let $I^\conv \in \mathcal{I}^\conv$ be an arbitrary convexified instance. 
We first show that for every $x \in X(x)$, the inequality $\hat{V}^{\conv}(x) \leq \hat{V}(x)$ holds, where $\hat{V}^{\conv}$ is the $\hat{V}$ function  for  $I^{\conv}$. For an arbitrary $x \in X(x)$, we have:
\begin{align}
\hat{V}^{\conv}(x) &= \sup_{y\in X^{\conv}(x)} \sum_{i \in N} \left[ \phi_i(x) - \phi_i(y_i,x_{-i}) \right]
=\sum_{i \in N}  \left[ \phi_i(x) \right] - \inf_{y\in X^{\conv}(x)} \sum_{i \in N}  \left[ \phi_i(y_i,x_{-i}) \right]. \label{Vconv} 
\end{align}
As $x \in X(x)$ and subsequently $x_{-i} \in \dom X_i$ for all $i \in N$, we have by Definition~\ref{def: conv}.\ref{enu: c1}.~that $y \in X^\conv(x)$ iff $y_i \in \conv(X_i(x_{-i}))$ for all $i\in N$. Furthermore the objective function $\sum_{i \in N}  \left[ \phi_i(y_i,x_{-i}) \right]$ is obviously separable in $y$. Therefore, the following is true:
\begin{align}\label{hilfeq1}
   \inf_{y\in X^{\conv}(x)} \sum_{i \in N}  \left[ \phi_i(y_i,x_{-i}) \right] = \sum_{i \in N}  \inf_{y_i \in \conv(X_i(x_{-i}))} \left[ \phi_i(y_i,x_{-i}) \right].
\end{align}
As for all $i \in N$ the function $\conv(X_i(x_{-i}))\to \R, \, y_i \mapsto \phi_i(y_i,x_{-i})$ is the convex envelope of the function $X_i(x_{-i})\to \R,\, y_i \mapsto \pi_i(y_i,x_{-i})$, the following equality holds:
\begin{align}\label{hilfeq2}
    \inf_{y_i \in \conv(X_i(x_{-i}))} \phi_i(y_i,x_{-i}) = \inf_{y_i \in X_i(x_{-i})} \pi_i(y_i,x_{-i}). 
\end{align}
Thus, we arrive at:
\begin{alignat}{2}
 \hat{V}^\conv(x)&= 
  \sum_{i \in N}  \left[ \phi_i(x) \right] - \inf_{y\in X^{\conv}(x)} \sum_{i \in N}  \left[ \phi_i(y_i,x_{-i}) \right] && \quad \text{ by } \eqref{Vconv}\nonumber\\ 
 &=\sum_{i \in N}  \left[ \phi_i(x) \right] - \sum_{i \in N}  \inf_{y_i \in X_i(x_{-i})} \left[ \pi_i(y_i,x_{-i}) \right]&& \quad\text{ by } \eqref{hilfeq1} \text{ and } \eqref{hilfeq2} \nonumber \\ 
 &\leq  \sum_{i \in N}  \left[ \pi_i(x) \right] - \sum_{i \in N} \inf_{y_i \in X_i(x_{-i})} \left[ \pi_i(y_i,x_{-i}) \right] && \quad\text{ by } \phi_i(x) \leq \pi_i(x), i\in N \text{ as } x\in X(x)\label{hilfeq3}\\
 &= \sum_{i \in N}  \left[ \pi_i(x) \right] - \inf_{y \in X(x)}\sum_{i \in N}   \left[ \pi_i(y_i,x_{-i}) \right]  && \quad\text{ by the same argumentation as for } \eqref{hilfeq1} \nonumber \\
& =\hat{V}(x) &&\nonumber 
\end{alignat}
Therefore, we have the inequality 
\begin{align}\label{ineq}
    \hat{V}^{\conv}(x) \leq \hat{V}(x) \quad \text{ for all } x\in X(x)
\end{align}
 which allows us to prove the equivalence $\ref{enu: mthm1}.\Leftrightarrow\ref{enu: mthm2}.$:
 
$\ref{enu: mthm1}.\Rightarrow\ref{enu: mthm2}.$:
Let $x\in X(x)$ be a generalized Nash equilibrium of $I$. Theorem~\ref{CharGNE} and inequality~\eqref{ineq} imply that $\hat{V}(x)=0\geq\hat{V}^{\conv}(x)$. Furthermore $x \in X^\conv(x)$ as $x \in X(x)$  and by observing that $\hat{V}^{\conv}(x) \geq 0$ for all $x \in X^\conv(x)$ we conclude that $\hat{V}^{\conv}(x) = 0$. Summarizing we have $x \in X^\conv(x)$ with  $\hat{V}^{\conv}(x) = 0$ which is equivalent to $x$ being a generalized Nash equilibrium for $I^\conv$ by Theorem \ref{CharGNE}. Furthermore $\hat{V}^{\conv}(x) =0=\hat{V}(x)$ implies that the inequality in \eqref{hilfeq3} must be tight, i.e.~$\sum_{i \in N} \pi_i(x)=\sum_{i \in N} \phi_i(x)$ holds. Together with $\phi_i(x) \leq \pi_i(x), i\in N$ we thus get $\phi_i(x) = \pi_i(x), i \in N$. 

$\ref{enu: mthm1}.\Leftarrow\ref{enu: mthm2}.$:
Let $x\in X(x)$ be a generalized Nash equilibrium of $I^\conv$ and $\phi_i(x) = \pi_i(x), i \in N$. Theorem \ref{CharGNE} implies that $\hat{V}^{\conv}(x) = 0$ while the equality $\phi_i(x) = \pi_i(x), i\in N$ implies that the inequality in (\ref{hilfeq3}) is tight and therefore $\hat{V}(x)=\hat{V}^{\conv}(x)=0$ holds. Again, Theorem \ref{CharGNE} yields that $x$ is a generalized Nash equilibrium for $I$ which finishes our proof. 
\end{proof}

Theorem \ref{thm:main} allows us to formulate the following characterization of a generalized Nash equilibrium: 
\begin{corollary}\label{korollar:main}
For an instance $I$ of the GNEP and any $I^\conv \in \mathcal{I}^\conv$, 
the following statements are equivalent.
\begin{itemize}
    \item[1.] $x$ is a generalized Nash equilibrium for $I$.
    \item[2.] $x \in X(x)$, $\hat{V}^\conv(x) = 0$ and $\phi_i(x) = \pi_i(x),i \in N$.
    \item[3.] $x$ is an optimal solution of \eqref{Opt} with value zero and $\phi_i(x) = \pi_i(x),i \in N$. \begin{align}\label{Opt}
    \inf_{x\in X(x)}\hat{V}^\conv(x) 
\end{align}
\end{itemize}
\end{corollary}

Before we illustrate the above concepts in the following Example~\ref{exa: concept}, let us briefly discuss the relevance of Theorem~\ref{thm:main} and Corollary~\ref{korollar:main}, in particular in regard of practical applications. In the latter, one may not have the precise description of the refined domains $\dom X_i$ nor convex envelopes $\phi_i^{x_{-i}}$ at hand as they may for example be just too difficult/expensive to compute. However, there might be types of instances $I$ for which one can a priori identify computationally tractable convexified instances without ever having to compute the refined domains or the convex envelopes in practice, cf.~the instance $I^\conv$ in Example~\ref{exa: concept}. Thus,  in the remainder of the paper, we will identify several types of original instances for which we can describe computationally tractable convexified instances.

We remark that a similar observation can be made from a theoretical standpoint. Identifying suitable, well-understood  convexified instances may yield structural insights for the original ones. For example, there exist various conditions for a convex GNEP in the literature guaranteeing the existence of an equilibrium which may be used together with Theorem~\ref{thm:main} to guarantee the existence of equilibria in non-convex games, cf.~Corollary~\ref{cor: convopt1}.

\begin{example}\label{exa: concept}
Let $I=(N,(X_i(\cdot))_{i\in N},(\pi_i)_{i \in N})$ be a 2-player GNEP, i.e.~$N = [2]$, where the strategy sets are 1-dimensional sets given by 
\begin{align*}
    X_1(x_2) &:= X_1^{\text{rel}}(x_2) \cap \Z := \{x_1 \in \R\mid (x_1-2)^2 + (x_2-2)^2 \leq 1  \} \cap \Z \text{ and }\\ 
    {X_2(x_1)} &{:= X_2^{\text{rel}}(x_1) \cap \Z :=  \{x_2 \in \R_{\geq 0} \mid x_2 + x_1 \geq 2,\;2x_2-x_1 \leq 5,\; x_2 +3x_1 \leq 9\}\cap \Z}
\end{align*}
for all $x_2,x_1 \in \R$ respectively.
\begin{figure}[!h]
\begin{center}
\scalebox{\myscale}{
\begin{tikzpicture}
\begin{axis}[clip mode=individual, 
  axis lines=middle,
  x = 1cm, y = 1cm,
  xmin=0,xmax=3.99,ymin=0,ymax=3.99,
  xtick distance=1,
  ytick distance=1,
  xlabel=$x_{1}$,
  ylabel=$x_{2}$,
  grid=major,
  grid style={thin,densely dotted,black!20}]
  \draw (axis cs:2,2) circle [radius=1cm];
\node at (axis cs:2,3) [circle, scale=0.3, draw=darkgreen,fill=darkgreen] {};
\node at (axis cs:2,1) [circle, scale=0.3, draw=darkgreen,fill=darkgreen] {};
\node at (axis cs:1,2) [circle, scale=0.3, draw=darkgreen,fill=darkgreen] {};
\node at (axis cs:3,2) [circle, scale=0.3, draw=darkgreen,fill=darkgreen] {};
\draw[thick,darkgreen] (axis cs:2,1) -- (axis cs:2,3);

\node[darkgreen] at (axis cs:2,3.5) {\scalebox{\myoscale}{$X_1(x_2)\times x_2$}};
\end{axis}
\end{tikzpicture} }\hspace{2cm}
\scalebox{\myscale}{
\begin{tikzpicture}
\begin{axis}[clip mode=individual, 
  axis lines=middle,
  x = 1cm, y = 1cm,
  xmin=0,xmax=3.99,ymin=0,ymax=3.99,
  xtick distance=1,
  ytick distance=1,
  xlabel=$x_{1}$,
  ylabel=$x_{2}$,
  grid=major,
  grid style={thin,densely dotted,black!20}]

\draw[thick] (axis cs:0,2) -- (axis cs:2,0);
\draw[thick] (axis cs:0,2.5) -- (axis cs:3,4);
\draw[thick] (axis cs:1.666,4) -- (axis cs:3,0);

\node at (axis cs:0,2) [circle, scale=0.3, draw=orange,fill=orange] {};
\node at (axis cs:1,1) [circle, scale=0.3, draw=orange,fill=orange] {};
\node at (axis cs:1,3) [circle, scale=0.3, draw=orange,fill=orange] {};
\node at (axis cs:2.666,1) [circle, scale=0.3, draw=orange,fill=orange] {};
\node at (axis cs:2.333,2) [circle, scale=0.3, draw=orange,fill=orange] {};
\node at (axis cs:2,0) [circle, scale=0.3, draw=orange,fill=orange] {};
\node at (axis cs:3,0) [circle, scale=0.3, draw=orange,fill=orange] {};
\node at (axis cs:2,3) [circle, scale=0.3, draw=orange,fill=orange] {};

\draw[thick,orange] (axis cs:0,2) -- (axis cs:2.333,2);
\draw[thick,orange] (axis cs:1,1) -- (axis cs:2.666,1);
\draw[thick,orange] (axis cs:1,3) -- (axis cs:2,3);
\draw[thick,orange] (axis cs:2,0) -- (axis cs:3,0);

\node[orange] at (axis cs:3.3,3) {\scalebox{\myoscale}{$X_2(x_1)\times x_1$}};
\end{axis}
\end{tikzpicture}} \hspace{2cm} \scalebox{\myscale}{
\begin{tikzpicture}
\begin{axis}[clip mode=individual, 
  axis lines=middle,
  x = 1cm, y = 1cm,
  xmin=0,xmax=3.99,ymin=0,ymax=3.99,
  xtick distance=1,
  ytick distance=1,
  xlabel=$x_{1}$,
  ylabel=$x_{2}$,
  grid=major,
  grid style={thin,densely dotted,black!20}]
  
\draw[thick] (axis cs:0,2) -- (axis cs:2,0);
\draw[thick] (axis cs:0,2.5) -- (axis cs:3,4);
\draw[thick] (axis cs:1.666,4) -- (axis cs:3,0);

\node at (axis cs:0,2) [circle, scale=0.3, draw=orange,fill=orange] {};
\node at (axis cs:1,1) [circle, scale=0.3, draw=orange,fill=orange] {};
\node at (axis cs:1,3) [circle, scale=0.3, draw=orange,fill=orange] {};
\node at (axis cs:2.666,1) [circle, scale=0.3, draw=orange,fill=orange] {};
\node at (axis cs:2.333,2) [circle, scale=0.3, draw=orange,fill=orange] {};
\node at (axis cs:2,0) [circle, scale=0.3, draw=orange,fill=orange] {};
\node at (axis cs:3,0) [circle, scale=0.3, draw=orange,fill=orange] {};
\node at (axis cs:2,3) [circle, scale=0.3, draw=orange,fill=orange] {};

\draw[thick,orange] (axis cs:0,2) -- (axis cs:2.333,2);
\draw[thick,orange] (axis cs:1,1) -- (axis cs:2.666,1);
\draw[thick,orange] (axis cs:1,3) -- (axis cs:2,3);
\draw[thick,orange] (axis cs:2,0) -- (axis cs:3,0);

\node at (axis cs:2,1) [circle, scale=0.8, draw=red, very thick] {};
\node at (axis cs:1,2) [circle, scale=0.8, draw=red, very thick] {};
\node at (axis cs:2,2) [circle, scale=0.8, draw=red, very thick] {};
\node at (axis cs:2,3) [circle, scale=0.8, draw=red, very thick] {};

  \draw (axis cs:2,2) circle [radius=1cm];
\node at (axis cs:2,3) [circle, scale=0.3, draw=darkgreen,fill=darkgreen] {};
\node at (axis cs:2,1) [circle, scale=0.3, draw=darkgreen,fill=darkgreen] {};
\node at (axis cs:1,2) [circle, scale=0.3, draw=darkgreen,fill=darkgreen] {};
\node at (axis cs:3,2) [circle, scale=0.3, draw=darkgreen,fill=darkgreen] {};
\draw[thick,darkgreen] (axis cs:2,1) -- (axis cs:2,3);

\end{axis}
\end{tikzpicture}}
\end{center}
    \caption{Representation of the strategy sets {and the resulting set of feasible strategy profiles $x \in X(x)$ marked via red circles}.}
    \label{fig: startexa}
\end{figure}
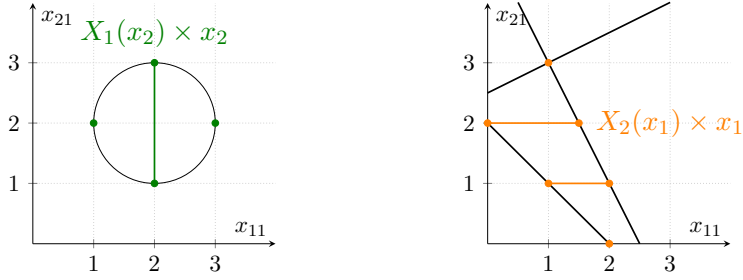

For the cost functions we set $\pi_1(x) := x_1\cdot x_2$ and $\pi_2(x) := (1 + |x_2|)^{x_1 - 1}$ for all $x \in \R^2$. 

The refined domains are described by 
\begin{align*}
    \dom X_1&= \{x_2 \in \R \mid \exists\,\mtilde{x}_1 \in \R :& &\myhspace\mtilde{x}_1 \in X_1(x_2),\; x_2 \in X_2(\mtilde{x}_1) \}\\
    &= \{x_2 \in \R \mid \exists\,\mtilde{x}_1 \in \R:\; & &\myhspace\mtilde{x}_1 \in \Z,\; (\mtilde{x}_1-2)^2 + (x_2-2)^2 \leq 1,\\
    && &\myhspace x_2 \in \Z_{\geq 0},\;  x_2 + \mtilde{x}_1 \geq 2,\;2x_2-\mtilde{x}_1 \leq 5,\; x_2 +3\mtilde{x}_1 \leq 9\}\\
    &=\{1,2,3\}\text{ and} \\
    \dom X_2&= \{x_1 \in \R \mid \exists\,\mtilde{x}_2 \in \R : && \myhspace x_1 \in X_1(\mtilde{x}_2),\; \mtilde{x}_2 \in X_2(x_1) \}\\ 
        &= \{x_1 \in \R \mid \exists\,\mtilde{x}_2 \in \R:\; & &\myhspace x_1 \in \Z,\, (x_1-2)^2 + (\mtilde{x}_2-2)^2 \leq 1,\\
    && &\myhspace\mtilde{x}_2 \in \Z_{\geq 0},\;  \mtilde{x}_2 + x_1 \geq 2,\;2\mtilde{x}_2-x_1 \leq 5,\; \mtilde{x}_2 +3x_1 \leq 9\}\\
    &=\{1,2\}. 
\end{align*}
{Alternatively, the set of feasible strategies is given by $\{(1,2),(2,1),(2,2),(2,3)\}$ with the projection $\{1,2,3\}$ to $R^{k_{-1}}$ and $\{1,2\}$ to $R^{k_{-2}}$. } 
In comparison, the standard domains are given by {$\mathrm{dom}X_1 = [1,3]$ and $\mathrm{dom}X_2 = [0,3]$ since $X_1(x_2)\neq \emptyset$ iff $x_2 \in [1,3]$ (cf.~the green points in the first picture in Figure~\ref{fig: startexa}, e.g.~$X_1(2.5)=\{2\}$) and analogously $X_2(x_1)\neq \emptyset$ iff $x_1 \in [0,3]$}.

By Definition~\ref{def: conv} an instance $I^\conv$ is a convexified instance for $I$ if and only if it fulfills: 
\begin{enumerate}
    \item $X_1^\conv(x_2) = \conv(X_1(x_2)) = \{2\}$  for $x_2 \in \{1,3\}$  and $X_1^\conv(x_2) = [1,3]$ for $x_2 = 2.$  \\$X_2^\conv(x_1) = \conv(X_2(x_1)) = [1,3]$ for  $x_1 =1$ and $X_2^\conv(x_1) = [0,3]$ for  $x_1 =2$.
    \item $\phi_1(x) = x_1 \cdot x_2$ for $x_1 \in \{2\}, x_2 \in \{1,3\}$ and $\phi_1(x) = 2x_1$ for $x_1 \in [1,3], x_2 \in\{2\}$. \\ 
    $\phi_2(x) = 1$ for $x_2 \in [1,3], x_1 \in \{1\}$ and  $\phi_2(x) = 1 + x_2$ for $x_2 \in [0,3], x_1 \in \{ 2\}$. 
\end{enumerate}
As described in Definition~\ref{def: conv}, the above setting allows for a whole set of convexified instances.
Some of these convexifications may have a natural
and computational efficient representation for which even the domains $\dom X_i,i\in N$ need not be known a priori.
To illustrate this aspect, let us
further specify two different
convexifications for our example that both belong to  $\mathcal{I}^\conv$:
$I^{\conv} = (N,(X_i^{\conv}(\cdot))_{i\in N},(\phi_i)_{i \in N})$ and $\bar{I}^{\conv} = (N,(\bar{X}_i^{\conv}(\cdot))_{i\in N},(\bar{\phi}_i)_{i \in N})$.

In the first game, we use the canonical relaxation of the original strategy sets for both players, i.e.~$X_i^{\conv}(x_{-i}) := X_i^{\text{rel}}(x_{-i})$ for $i \in [2]$ and all $x_{-i} \in \R^{k_{-i}}$.
Regarding the cost functions, we define $\phi_1(x):= \pi_1(x)$  and $\phi_2(x) := (x_1 -1) \cdot x_2 + 1$ for all $x\in \R^2$ which is possible as $\pi_1(x)$ is linear for fixed $x_2 \in \dom X_1$.
The possibility to use the canonical relaxations of the original strategy sets and the setting of $\phi_1 = \pi_1$ is of course not always possible, yet, in the following Section~\ref{sec:darstellung} we will identify certain types of original instances for which this is always possible.

In the second game,  we set 
\begin{align*}
    \bar{X}_1^{\conv}(x_2) :=\begin{cases}\{2\}, &\text{if } x_2 \in \{1,3\}\\
    [1,3], &\text{if } x_2 = 2 \\
    \emptyset, &\text{else}\end{cases} \text{ and }     \bar{X}_2^{\conv}(x_1) :=X_2^{\text{rel}}(x_1) \text{ for all } x_1 \in \R. 
\end{align*}
As cost functions we define $\bar{\phi}_1(x):=\pi_1(x)$ and $\bar{\phi}_2(x) :=\pi_2(x)$ for all $x \in \R^2$.

Both games in the above example fulfill \ref{enu: c1}.~and \ref{enu: c2}.~but their cost functions and strategy sets differ significantly whenever $x_{-i}\notin \dom X_i$. This results in different smoothness properties of the games. For instance, the convexified cost function $\phi_2$ of player 2 in the first game is convex whenever $x_1$ is fixed for all $x_1 \in \R$ whereas in the second game,  $\bar{\phi}_2$ is only convex for fixed $x_1$ when $x_1 \in [2,\infty)$. Similarly, the strategy sets of player 1 in the first game are described by a $C^\infty$ convex restriction function which contrasts the second game.  
If we now apply Corollary~\ref{korollar:main} to the two different convexifications $I^\conv$ and $\bar{I}^\conv$, we obtain quite different optimization problems~\eqref{Opt}: 
\begin{align*}
\inf_x \; \max_y \;  &x_1 \cdot x_2 + (x_1 -1) \cdot x_2 +1\\ &- \big(y_1\cdot x_2 + (x_1 - 1)\cdot y_2 + 1\big) \nonumber\\
\text{s.t.: }&(y_1 - 2)^2 + (x_2 -2)^2 -1 \leq 1 \nonumber\\
&y_2 + x_1 \geq 2,\;2y_2-x_1 \leq 5,\; y_2 +2x_1 \leq 5 \nonumber\\
&(x_1 - 2)^2 + (x_2 -2)^2 -1 \leq 1 \nonumber\\
&x_2 + x_1 \geq 2,\;2x_2-x_1 \leq 5,\; x_2 +2x_1 \leq 5 \nonumber\\
& y_1 \in \R,\;y_2\in\mathbb{R}_{\geq 0},\; x_1 \in \Z,\;x_2 \in \Z_{\geq 0} \nonumber
 \\
\inf_x \; \max_y \;  &x_1 \cdot x_2 + (1 + |x_2|)^{x_1 -1}\\ &- \big(y_1\cdot x_2 + (1 + |y_2|)^{x_1 -1}\big) \nonumber\\
\text{s.t.: }&y_1 \in \bar{X}_1^\conv(x_2) \nonumber\\
&y_2 + x_1 \geq 2,\;2y_2-x_1 \leq 5,\; y_2 +2x_1 \leq 5 \nonumber\\
&x_1 \in  \bar{X}_1^\conv(x_2) \nonumber\\
&x_2 + x_1 \geq 2,\;2x_2-x_1 \leq 5,\; x_2 +2x_1 \leq 5 \nonumber\\
& y_1 \in \R,\;y_2\in\mathbb{R}_{\geq 0},\; x_1 \in \Z,\;x_2 \in \Z_{\geq 0} \nonumber
\end{align*}
 
\end{example}
Note again the effect of the two different convexifications on the structure of the resulting feasible sets and their objective functions: in contrast to the second problem,  the first
one admits a quasi-linear objective
and a feasible set defined via smooth algebraic constraints.
In the remainder of the paper, we will further identify necessary and sufficient conditions of an instance $I$
so that it allows for ``well-behaved''
convexified instances and consequently leads to more tractable optimization problems~\eqref{Opt}.

\section{Quasi-Linear GNEPs}\label{sec:darstellung}
\subsection{MINLP-Reformulation}

In what follows we identify a subclass of the GNEP such that the optimization problem~\eqref{Opt} becomes more accessible. The main obstacle when solving~\eqref{Opt} in the general case is the need to separately evaluate the function $\hat{V}^\conv$ for any strategy profile $x$. This is due to the fact that the evaluation of $\hat{V}^\conv$ at any $x$ is itself a maximization problem which ultimately leads to a computationally intractable optimization problem~\eqref{Opt} as the latter constitutes of a maximization problem nested within a minimization problem. 
{For linear GNEPs, Stein and Sudermann-Merx~\cite{stein2016cone} showed that one can resolve this problem by dualizing the corresponding $\hat{V}$ function. 
We define in the following the class of \emph{quasi-linear} GNEPs $I$ for which there exists a convexified instance $I^\conv \in \mathcal{I}^\conv$ admitting the property that $\hat{V}^\conv(x)$ is a linear maximization problem, i.e.~for every $i\in N$ and fixed $x_{-i}\in \R^{k_{-i}}$, player $i\in N$ has a linear cost function as well as a polyhedral strategy set.
While $I^\conv$ does not need to belong to the linear GNEPs considered in~\cite{stein2016cone},  
their dualization idea is still applicable resulting in a reformulation of~\eqref{Opt} as an optimization problem~\eqref{lemma:dar:opt} in standard form.
}
\begin{definition}\label{GenSet}
An instance $I$ of the GNEP is called \emph{quasi-linear}, if there exists $I^\conv \in \mathcal{I}^\conv$ which fulfills for every $i\in N$ the following two statements:
\begin{enumerate}
    \item\label{enum:QL1}
    There exists a matrix-valued function $M_i: \dom X_i \to \mathbb{R}^{l_i \times k_i}$, $x_{-i} \mapsto M_i(x_{-i})$ and a vector-valued function $e_i:\dom X_i \to \mathbb{R}^{l_i}$, $x_{-i} \mapsto e_i(x_{-i}) $  for an integer $l_i \in \mathbb{N}$, such that
    \begin{align*}
        X_i^\conv(x_{-i}) = \left\lbrace x_i \in \mathbb{R}^{k_i}\mid M_i(x_{-i})x_i\geq e_i(x_{-i}) \right\rbrace \text{  for all }  x_{-i} \in \dom X_i. 
    \end{align*}
    \item\label{enum:QL2} There exists a vector-valued function $C_i:\dom X_i\to \mathbb{R}^{k_i}$, $x_{-i} \mapsto C_i(x_{-i})$ such that
    \begin{align*}
          \phi_i(x_i,x_{-i}) = C_i(x_{-i})^\top x_i\; \text{ for all } (x_i,x_{-i}) \in \R^{k_i}\times \dom X_i.
    \end{align*}
\end{enumerate}
\end{definition}

\begin{theorem}\label{theorem:darstellung}
Let $I$ be a quasi-linear GNEP with an instance $I^\conv \in \mathcal{I}^\conv$ as described in Definition~\ref{GenSet}. Every optimal solution 
$(x,\nu)$ of~\eqref{lemma:dar:opt} with value zero and $\phi_i(x) = \pi_i(x),i\in N$ corresponds to a GNE $x$ of $I$ and vice versa. 
\begin{alignat}{2}
\label{lemma:dar:opt}\tag{R}
\inf_{\nu,x}\; &\sum_{i \in N}  C_i\left(x_{-i}\right)^\top x_i - \sum_{i \in N} e_i(x_{-i})^\top \nu_i\nonumber \\
\text{s.t.: } &\nu_i^\top M_i(x_{-i}) =  C_i\left(x_{-i}\right)^\top &&\text { for all } i\in N,  \nonumber\\
&\nu_i \in \mathbb{R}^{l_i}_{\geq 0},\;x_i\in X_i(x_{-i}) &&\text { for all } i\in N. \nonumber
\end{alignat}
\end{theorem}

\begin{proof}
Consider for an arbitrary but fixed $x \in X(x)$ the function $\hat{V}^\conv(x)$ where the latter is again the $\hat{V}$ function corresponding to $I^\conv$. From the proof of Theorem~\ref{thm:main} 
we already know: 
\begin{align*}
    \hat{V}^\conv(x) = \sum_{i \in N} \phi_i(x) -  \sum_{i \in N}  {\inf}_{y_i\in X_i^\conv(x_{-i})}\phi_i(y_i,x_{-i})= \sum_{i \in N} \phi_i(x)  -  \sum_{i \in N} {\inf}_{y_i\in X_i^\conv(x_{-i})}C_i(x_{-i})^\top y_i. 
\end{align*}
By using Definition~\ref{GenSet}.\ref{enum:QL1}.,  we arrive at the following linear optimization problem~\eqref{LP_i} { with its corresponding dual~\eqref{DP_i}} for the optimization problem of player $i$ in the convexified game for the rivals' strategies $x_{-i}$: \vspace*{0.25cm}\\
\vspace*{0.25cm}
\begin{minipage}{0.45\textwidth}
\vspace{-0.4cm}
\begin{align}
{\inf}_{y_i} \; &C_i\left(x_{-i}\right)^\top y_i \nonumber\\
\text{s.t.: }& M_i(x_{-i})y_i\geq e_i(x_{-i}), \label{LP_i}\tag{$\mathrm{LP}_i(x_{-i})$}\\
& y_i\in\mathbb{R}^{k_i}, \nonumber
\end{align}
\end{minipage}
\hfill
\vline
\hfill
\begin{minipage}{0.45\textwidth}
\vspace{-0.4cm}
\begin{align}
{\sup}_{\nu_i} \; &e_i\left(x_{-i}\right)^\top \nu_i \nonumber\\
\text{s.t.: }&\nu_i^\top M_i(x_{-i})= C_i(x_{-i})^\top, \label{DP_i}\tag{$\mathrm{DP}_i(x_{-i})$}\\
& \nu_i\in\mathbb{R}^{l_i}_{\geq 0}.\nonumber
\end{align}
\end{minipage}
\;\vspace{0.25cm}\\
{Note that~\eqref{LP_i} attains its minimum iff the problem is bounded from below. In this case, we get by linear programming duality that the dual 
attains its maximum and their optimal objective values coincide. }
In the following let us denote by~\ref{DP_i} and~\ref{LP_i} also the corresponding optimal objective values { with the convention that~\ref{DP_i} $ = -\infty$ if \ref{DP_i} has no feasible solution. By this convention and the above argument, we have {\ref{DP_i} $ = $ \ref{LP_i}}   which allows us to reformulate $\hat{V}^\conv(x)$ as: }
\begin{align*}
    \hat{V}^\conv(x) = \sum_{i \in N} \phi_i(x) -  \sum_{i \in N}  \mathrm{LP}_i(x_{-i}) = \sum_{i \in N} \phi_i(x) - \sum_{i \in N}  \mathrm{DP}_i(x_{-i}). 
\end{align*}
Since the $n$ maximization problems~\eqref{DP_i} are completely separable we can combine them to one maximization problem. {Hence, by applying  the representation of the convex envelopes from Definition~\ref{GenSet}.\ref{enum:QL2}.,  we can describe $\hat{V}$ via the following optimization problem}
\begin{alignat}{2}
    {\inf}_\nu \; &\sum_{i \in N}C_i\left(x_{-i}\right)^\top x_i  - \sum_{i \in N}    e_i\left(x_{-i}\right)^\top \nu_i \label{eq: refV}\\
\text{s.t.: }& \nu_i^\top M_i(x_{-i})= C_i(x_{-i})^\top &&\quad \text{ for all } i \in N,\nonumber \\
& \nu_i\in\mathbb{R}^{l_i}_{\geq 0} &&\quad\text{ for all } i \in N\nonumber
\end{alignat}
{together with the property that $\Hat{V}(x) < \infty $ if and only if  \eqref{eq: refV}   attains its minimum at some $\nu$ with the objective value $\hat{V}(x)$. This allows us to relate to each $x\in X(x)$ with $\hat{V}(x)<\infty$ a feasible solution $(x,\nu)$ of~\eqref{lemma:dar:opt} with $\nu$ being optimal for~\eqref{eq: refV} and objective value equal to $\hat{V}(x)$ and vice versa. 
Hence,   any optimal  solution $(x,\nu)$ of~\eqref{lemma:dar:opt} corresponds to an optimal solution $x$ of~\eqref{Opt} with the same objective value. Conversely, an optimal solution $x$ of~\eqref{Opt} with $\hat{V}(x)<\infty$ can be identified with an optimal  solution $(x,\nu)$ of~\eqref{lemma:dar:opt}. Thus, the claim of the theorem follows by Corollary~\ref{korollar:main}. } 
 \end{proof}

Theorem~\ref{theorem:darstellung} is in particular interesting for quasi-linear GNEPs $I$ in which the conditions $x\in X(x)$ and $\phi_i(x) = \pi_i(x), i \in N$ are computationally tractable. Such a situation is present in various interesting applications where 
instances $I$ are used which belong to the class of what we call \emph{{player-linear mixed}-integer GNEPs}. 
\begin{definition}\label{def:lmi}
An instance $I$ belongs to the \emph{{player-linear mixed}-integer GNEPs}, if for every $i\in N$ the strategy space and cost functions are described by  
\begin{align}
    X_i(x_{-i}) &= \left\lbrace x_i \in \mathbb{Z}^{s_i}\times\mathbb{R}^{k_i-s_i}\mid \tilde{M}_i(x_{-i})x_i\geq \tilde{e}_i(x_{-i}) \right\rbrace \text{ for all } x_{-i} \in \dom X_i \label{eq:MixedInteger} \\
   \pi_i(x_i,x_{-i}) &= \tilde{C}_i(x_{-i})^\top x_i\; \text{ for all } (x_i,x_{-i}) \in \R^{k_i}\times \dom X_i \nonumber
\end{align}
for a matrix-valued function $\tilde{M}_i: \dom X_i \to \mathbb{R}^{l_i \times k_i}$ and  vector-valued functions $\tilde{e}_i:\dom X_i \to \mathbb{R}^{l_i}, x_{-i} \mapsto \tilde{e}_i(x_{-i})$  and $\tilde{C}_i:\dom X_i\to \mathbb{R}^{k_i}, x_{-i} \mapsto \tilde{C}_i(x_{-i})$ for some $s_i \in [k_i], l_i \in \N$. 
\end{definition}
Note that for a {player-linear mixed}-integer GNEP $I$, the convex hull of the set $X_i(x_{-i})$ is a polytope  (cf. Conforti et al.~\cite{Conforti2010}) for any $x_{-i} \in \dom X_i$. Thus, the existence of a convexified instance fulfilling Definition~\ref{GenSet}.\ref{enum:QL1}.~is guaranteed. The same holds true for Definition~\ref{GenSet}.\ref{enum:QL2}.~since for all $x_{-i}\in \dom X_i$ the convex envelope $\phi_i^{x_{-i}}(\cdot)$ of $\pi_i(\cdot,x_{-i}):\conv(X_i(x_{-i})) \to \R, x_i \mapsto \pi_i(x_i,x_{-i}) = \tilde{C}_i(x_{-i})^\top x_i$ is just $\pi_i(\cdot,x_{-i})$ due to the linearity of the latter. Therefore any {player-linear mixed}-integer GNEP is a quasi-linear GNEP {and thus automatically fulfills $\mathcal{I}^\conv \neq \emptyset$}.
Theorem~\ref{theorem:darstellung} yields the following corollary.
\begin{corollary}\label{cor:ref}
Let $I$ be a {player-linear mixed}-integer GNEP with an instance $I^\conv \in \mathcal{I}^\conv$ as described in Definition~\ref{GenSet}. Then, every optimal solution $(x,\nu)$ with value zero of the mixed-integer optimization problem~\eqref{opt:Ref} corresponds to a GNE $x$ of $I$ and vice versa.
\begin{alignat}{2}
\label{opt:Ref}\tag{$\tilde{\mathrm{R}}$}
\inf\; &\sum_{i \in N}  C_i\left(x_{-i}\right)^\top x_i - \sum_{i \in N}  e_i(x_{-i})^\top \nu_i\nonumber \\
\text{s.t.: } &{M}_i(x_{-i})x_i\geq e_i(x_{-i}) &&\text { for all } i\in N, \nonumber  \\
&\nu_i^\top M_i(x_{-i}) =  C_i\left(x_{-i}\right)^\top &&\text { for all } i\in N, \nonumber   \\
&\nu_i \in \mathbb{R}^{l_i}_{\geq 0},\;x_i\in\mathbb{Z}^{s_i}\times\mathbb{R}^{k_i-s_i} &&\text { for all } i\in N. \nonumber  
\end{alignat}
\end{corollary}
\begin{proof}
The argumentation previous to the corollary shows that $\phi_i(x) = \pi_i(x), i\in N$ holds for any $x \in X(x)$. Thus,  
the corollary follows immediately by Theorem~\ref{theorem:darstellung} and the fact that 
\begin{align*}
    X_i(x_{-i}) = \left\lbrace x_i \in \mathbb{Z}^{s_i}\times\mathbb{R}^{k_i-s_i}\mid {M}_i(x_{-i})x_i\geq {e}_i(x_{-i}) \right\rbrace
\end{align*}
holds for any $i \in N$ and $x_{-i} \in  \dom X_i$. To see this, we argue as follows. Firstly, we note that the 
{continuous} relaxation of the right set equals the convex hull of $X_i(x_{-i})$
since $I^\conv \in \mathcal{I}^\conv$ and $x_{-i} \in \dom X_i$.
Thus, the inclusion $\subseteq$ is verified by  
$X_i(x_{-i}) \subseteq \conv(X_i(x_{-i}))$ and the following:
\begin{align*}
X_i(x_{-i}) = X_i(x_{-i}) \cap\mathbb{Z}^{s_i}\times\mathbb{R}^{k_i-s_i} &\subseteq  \conv(X_i(x_{-i}))\cap\mathbb{Z}^{s_i}\times\mathbb{R}^{k_i-s_i}\\ 
&= \left\lbrace x_i \in \mathbb{R}^{k_i}\mid {M}_i(x_{-i})x_i\geq {e}_i(x_{-i}) \right\rbrace\cap\mathbb{Z}^{s_i}\times\mathbb{R}^{k_i-s_i}\\
&= \left\lbrace x_i \in \mathbb{Z}^{s_i}\times\mathbb{R}^{k_i-s_i}\mid {M}_i(x_{-i})x_i\geq {e}_i(x_{-i}) \right\rbrace.
\end{align*}
The inclusion $\supseteq$ is valid, as the {linear} relaxation of $X_i(x_{-i})$ is a convex set and thus contains the convex hull of $X_i(x_{-i})$. Therefore we get:
\begin{align*}
    \conv(X_i(x_{-i})) \cap\mathbb{Z}^{s_i}\times\mathbb{R}^{k_i-s_i} \subseteq  \left\lbrace x_i \in \mathbb{R}^{k_i}\mid \tilde{M}_i(x_{-i})x_i\geq \tilde{e}_i(x_{-i}) \right\rbrace \cap \mathbb{Z}^{s_i}\times\mathbb{R}^{k_i-s_i} = X_i(x_{-i}).
\end{align*}
\end{proof}
Even for {player-linear mixed}-integer GNEPs,  the computation of the matrix- and vector-valued functions $M_i,e_i,i \in N$ for a convexified instance $I^\conv$ as described in Definition~\ref{GenSet} may in general be  quite complex. An exception to that is present in what we call \emph{hole-free} GNEPs, cf.~\cite{murota2019survey} for a similar concept in the realm of discrete convexity. Here we can use for all $i \in N$ the $I$-defining matrix- and vector-valued functions $\tilde{M}_i$ and $\tilde{e}_i$ instead of $M_i$ and $e_i$ in the above optimization problem~\eqref{opt:Ref}. 
{The key point of these hole-free represented GNEPs is the property that the strategy set $X_i(x_{-i})$ of a player $i$ is perfectly described for relevant strategies $x_{-i}\in\dom X_i$ in the sense that their convex hull coincides with their relaxation. Hence, the continuous relaxation of a hole-free instance $I$ is a convexified instance in $\mathcal{I}^{\conv}(I)$.}
\begin{definition}
 We call a {player-linear mixed}-integer GNEP \emph{hole-free-represented}, if for all $i \in N$ and $x_{-i}\in \dom X_i$,  the following equality holds:
 \begin{align}\label{eq: IGfree}
     \conv\left(\left\lbrace x_i \in \mathbb{Z}^{s_i}\times\mathbb{R}^{k_i-s_i}\mid \tilde{M}_i(x_{-i})x_i\geq \tilde{e}_i(x_{-i}) \right\rbrace\right) = \left\lbrace x_i \in \mathbb{R}^{k_i}\mid \tilde{M}_i(x_{-i})x_i\geq \tilde{e}_i(x_{-i}) \right\rbrace.
 \end{align}
\end{definition}
Note that the equality~\eqref{eq: IGfree} does only need to hold for $x_{-i}\in \dom X_i$ but not necessarily for all $x_{-i}$ in the (potentially substantial) bigger set $\mathrm{dom}X_i$ as the following example illustrates.

\begin{example}[Capacitated Discrete Flow Games (CDFG)] \label{exa: CDFG}
We consider a directed  graph $G=(V,E)$ with nodes $V$ and edges $E$.  
There is a set of players $N= \{1, \dots,
n\}$ where each $i\in N$ is associated with an end-to-end pair $(s_i,t_i)\in V\times V$ {as well as an individual constraint-vector $c_i\in\Z^E_{\geq 0}$.} 
The strategy $x_i$ of a player $i \in N$ represents an integral $s_i$-$t_i$-flow with a flow value equal to her demand $d_i \in \N$.
Hereby, a player is   restricted in her strategy choice by her capacity constraints,  i.e.~for given rivals' strategies $x_{-i}$, her flow $x_i$ has to satisfy the restriction {$x_i \leq c_i - \sum_{s\neq i}x_s$}. 
Thus the strategy set of a player $i \in N$ is described by 
\begin{align}\label{eq: CDFGStrat}
X_i(x_{-i})=X_i' \cap  \Big\{x_i\in \Z_{\geq 0}^E \mid  x_i \leq c_i - \sum_{s\neq i}x_s \Big\} 
\text{ for all } x_{-i}\in\R^{k_{-i}},
\end{align}
where $X_i' = \{ x_i\in \Z_{\geq 0}^E \mid Ax_i = b_i\}$ is the flow polyhedron of player $i$ with $A$ the arc-incidence matrix of the graph $G$ and $b_i$ the vector with $(b_i)_{s_i} = d_i$,  $(b_i)_{t_i} = -d_i$, and zero, otherwise. 
{ Remark that the (standard) domain $\mathrm{dom} X_i$ is given by all (not necessarily integral) $x_{-i} \in  \R^{k_{-i}}$ such that the intersection in~\eqref{eq: CDFGStrat} is nonempty. In contrast, by $X_i(x_{-i}) \subseteq \Z^E$,  we may deduce that   $\dom X_i \subseteq \mathrm{dom} X_i \cap  \Z^{k_{-i}}$. } 

We define the cost functions by 
\begin{align*}
 \pi_i(x_i,x_{-i}):=  \big(\sum_{j\neq i}x_j\big)^\top \, C_i^1\, x_i + C_i^{2\top}\, x_i =  \Bigl(\big(\sum _{j\neq i}x_j\big)^\top \, C_i^1 + C_i^{2\top}\Bigr)\, x_i
\end{align*}
with $C_i^1 \in \R^{E\times E}$ and $C_i^2 \in \R^E$.
Here, the first term can be interpreted as costs that arise through congestion whereas the second term represents congestion independent costs for player~$i$. 

Clearly,  
the CDFG belongs to the  {player-linear mixed}-integer GNEPs.
Furthermore,  it is hole-free-represented.
To see this,  it is sufficient to verify that for all $x_{-i} \in \dom X_i$ the inclusion

\begin{align*}\textstyle
\Bigl\lbrace x_i\in \R_{\geq 0}^E \mid Ax_i = b_i,\; x_i \leq c_i - \sum_{s\neq i}x_s \Bigr \rbrace \subseteq \conv\Big(\Bigl\lbrace x_i\in \Z_{\geq 0}^E \mid Ax_i = b_i,\; x_i \leq c_i - \sum_{s\neq i}x_s \Bigr \rbrace\Big) 
\end{align*}
holds as $\supseteq$ is trivially fulfilled. Since $\dom X_i \subseteq \Z^{k_{-i}}$, the restriction $0 \leq x_i \leq c_i - \sum_{s\neq i}x_s $ is an integral box-constraint for any  $x_{-i} \in \dom X_i$. Thus the polytope on the left has integral vertices since the flow polyhedron is box-tdi, see Edmonds and Giles~\cite{Edmonds1977} and Schrijver~\cite{Schrijver03}  for a definition of box-tdi and the aforementioned property of the flow polyhedron. These integral vertices are clearly contained in the right set and therefore the inclusion follows. Notice that for non-integral $x_{-i}$, the inclusion $\subseteq$ is in general false. {Hence, the inclusion $\subseteq$ is in general not true for all $x_{-i} \in \mathrm{dom} X_i$.}
\end{example} 

{
Let us motivate the hole-free GNEPs with another example regarding discrete market equilibria.}

{
 \begin{example}[Equilibria in Transportation Markets]\label{exa: CompEq}
Using the same terminology as in the previous Example~\ref{exa: CDFG}, 
consider the situation in which the edges $E$ are up for sale  
and each player wants to buy a single $s_i,t_i$-path $x_i$ with the goal to maximize her linear utility $U_i(x_i) = C_i^\top x_i, C_i\in \Z^E$.
The market manager wants to determine
an integral price vector $p\in \Z_{\geq 0}^E$ for selling the edges such that every
player receives a $s_i,t_i$-path $x_i^* \in X_i'$ maximizing her quasi-linear utility
$x_i^*\in \arg\max_{x_i\in X_i'}\{ U_i(x_i)-p^\top x_i \}$ and unsold edges have prices equal zero.
The tuple $((x_i^*)_{i\in N},p)$ is  known as a \emph{competitive equilibrium}, cf.~e.g.~\cite{Bikchandani1997}. 

We can model this situation as a GNEP with $n+1$ players in which the first $n$ players correspond to the $n$ buyers 
and the $n+1$-th player to the market manager. We denote by $(x,p)$ a strategy  profile and set the costs
to the negated utility $\pi_i(x_i,x_{-i},p) = -(U_i(x_i)-p^\top x_i)$ for $i \in[n]$ and the costs of the market manager 
to  $\pi_{n+1}(p,x) = (\mathbf{1} - \sum_{i\in[n]} x_i)^\top p$ with $\mathbf{1} = (1,\ldots,1)^\top\in \R^E$.
For the strategy spaces we set $X_i(x_{-i},p) \equiv X_i'$ to the flow polyhedron 
and $X_{n+1}(x) :=  \Z_{\geq 0}^E$ if $\sum_{i\in[n]} x_i \leq \mathbf{1}$ and $X_{n+1}(x) := \emptyset$ else.
 It is straight forward to verify that any GNE of this GNEP corresponds to a competitive equilibrium and vice versa.
 Furthermore, it follows from the observations in Example~\ref{exa: CDFG} that the GNEP is a hole-free 
 \emph{linear} mixed-integer GNEP.
 \end{example} 
}

We conclude this section with another consequence of Theorem~\ref{theorem:darstellung}. Namely,
under certain assumptions, the restriction of $x \in X(x)$ in~\eqref{Opt} may be obsolete and can be relaxed to $x \in X^\conv(x)$ as the feasibility of an optimal solution for the original game is a priori ensured. Such a case is described in the following corollary where additionally, the existence of generalized Nash equilibria of the instance $I$ can be determined by solving a convex optimization problem. For the promised corollary, we need the following definition:
\begin{definition}
 Let $l\in \N$ and $M \subseteq \R^l$ be arbitrary. We denote by $E(M)$ the set of all extreme points of $M$:
 \begin{align*}
    E(M):=  \big\{ x \in M \mid x \notin \conv(M \setminus \{x\}) \big\}.
 \end{align*}
\end{definition}
\begin{corollary}\label{cor: convopt1}
Let $I$ be a quasi-linear GNEP with an instance $I^\conv \in \mathcal{I}^\conv$ as described in Definition~\ref{GenSet}.,
where the functions $C_i(x_{-i}) \equiv C_i$ and $e_i(x_{-i}) \equiv e_i$ are both constant in $x_{-i}$ for all $i \in N$. Furthermore, assume that
  \begin{align*}
    F:= \big\lbrace (x,\nu)\in \R^{\sum_i k_i+ l_i} \mid M_i(x_{-i})x_i\geq e_i,\; \nu^\top_i M_i(x_{-i}) =  C_i^\top \nonumber,  \;i=1,\ldots,n  \big\rbrace 
 \end{align*}
 is convex and any $(x,\nu) \in E(F)$ satisfies $x \in X(x)$. Then, $I$ has a generalized Nash equilibrium if and only if the following convex optimization problem has the optimal value 0. 
  \begin{align}\label{cor: convopt}\textstyle
     \inf \; \sum_{i \in N}  C_i^\top x_i- e_i^\top \nu_i \quad \text{s.t.: } (x,\nu) \in F 
 \end{align}
 
\end{corollary}
\begin{proof}
Since $I^\conv \in \mathcal{I}^\conv(I^\conv)$, the latter is itself a quasi-linear GNEP. Thus, the optimization problem in~\eqref{lemma:dar:opt} for $I^\conv$ and $I^\conv \in \mathcal{I}^\conv(I^\conv)$ instead of $I$ and $I^\conv$ equals~\eqref{cor: convopt}. The result then follows by Theorem~\ref{thm:main} and~\ref{theorem:darstellung}, the equality $\phi_i(x) = \pi_i(x),i \in N$ for all $x \in X(x)$ as well as the fact that the optimization problem in~\eqref{cor: convopt} attains its minimum (if it exists) at an extreme point of $F$ as the set $F$ is convex and the objective function is linear.  
\end{proof}

{
\subsection{Hole-free Linear Mixed-Integer GNEPs}
Besides the approach described in Corollary~\ref{cor:ref} to compute equilibria of a player-linear mixed-integer GNEP, let us mention in the following another possibility for the special case of hole-free \emph{linear} mixed-integer GNEPs, that is, hole-free  player-linear mixed-integer GNEPs where $\tilde{M}_i,\tilde{C}_i, i \in N$ are constant and  $\tilde{e}_i$ is linear. 
For this special class, the linear relaxation is not only a convexified instance but also belongs to the class of linear (continuous and convex) GNEPs. 
As mentioned in the introduction,  Dreves~\cite{Dreves17} introduced an algorithm for linear GNEPs 
which computes the whole solution set  in a finite amount of time. 
Hence, our convexification result in Theorem~\ref{thm:main} shows that one can   determine the whole solution set of the original instance $I$ by
applying Dreves' algorithm to the linear relaxation, computing the whole solution set of the latter and determining all originally feasible solutions by re-administering the mixed-integer conditions. 
Note that the CDFG described above for the case of $C_i^1 = 0 \in\R^{m\times m}$ belongs for example to the class of hole-free \emph{linear} mixed-integer GNEPs. 
}

\section{Jointly Constrained GNEPs} \label{sec: JoinCons}
In several interesting applications, the players' strategy sets are restricted by coupled constraints. 
\begin{definition}
 We call an instance $I$ \emph{jointly constrained} w.r.t.~$X\subseteq \R^k$ if for all $i \in N$ and $x_{-i} \in \R^{k_{-i}}$, the strategy set $X_i(x_{-i})$ has the following description:
 \begin{align*}
    X_i(x_{-i}) = \left\lbrace x_i \in \mathbb{R}^{k_i} \mid (x_i,x_{-i}) \in X\right \rbrace.
\end{align*}
\end{definition}
Notice that the joint restriction set $X \subseteq \mathbb{R}^k$ doesn't need to be convex and may be discrete.
 This type of GNEP occurs for example in the domain of automated driving \cite{FabianiAD20}, traffic control \cite{cenedese2021highway} or transportation problems \cite{SagratellaNFCTP}.
Before we investigate the structure of any convexified game $I^\conv\in\mathcal{I}^\conv$ of a jointly constrained GNEP, let us motivate this special type of GNEP further by the following example.  

\begin{example}[Jointly Constrained Atomic Congestion Games]\label{exa: Atom}
We first describe the atomic (resource-weighted) congestion game, which is a generalization of  the model of Rosenthal~\cite{Rosenthal73}, without joint restrictions.
The set of strategies available to player~$i\in N =\{1,\ldots,n\}$ is given by  $X_i \subseteq\times_{j\in E}\{0,d_{ij}\}$
for weights $d_{ij}>0$ and resources $j\in E =\{1,\ldots,m\}$.
Note that by assuming 
$ x_i\in \{0,1\}^m$ for all $i\in N$, that is, $d_{ij}=1$, we obtain the
standard congestion game model of Rosenthal.
The cost functions on resources are given by \emph{player-specific  functions} $c_{ij}(\ell_j( x))$, $j\in E$, $i\in N$, where $\ell(x):=\sum_{i\in N} x_i$.
The private cost of a player $i\in N$ 
for strategy profile $x\in \prod_{i \in N} X_i$ is defined by 
$\pi_i(x_i,x_{-i}):= \sum_{j\in E} c_{ij}(\ell_j( x))x_{ij}.$
This model can be generalized by allowing joint restrictions in the players' strategy sets, that is, extending the above model to a jointly constrained GNEP with respect to a set  $X \subseteq \prod_{i \in N} X_i$, e.g., if the usage of resources is bounded by hard capacities. 
\end{example}
GNEPs with joint constraints were first studied in detail by Rosen in 1965~\cite{Rosen65}. Since then, these GNEPs have been the object of a fairly intense study in the literature and became one of the best understood subclasses of the GNEP, see~\cite{FacchineiK10}  for more details. Our goal is to identify properties of an instance $I$ so that a convexified instance $I^\conv \in \mathcal{I}^\conv$ exists that belongs to these well-understood jointly constrained/convex GNEPs.
It seems quite natural to hope for a given jointly constrained instance $I$ that its convexification $\mathcal{I}^\conv$ contains a jointly constrained GNEP. 
However, this is in general not the case as the example in
Figure~\ref{fig: ExaSJC} illustrates.

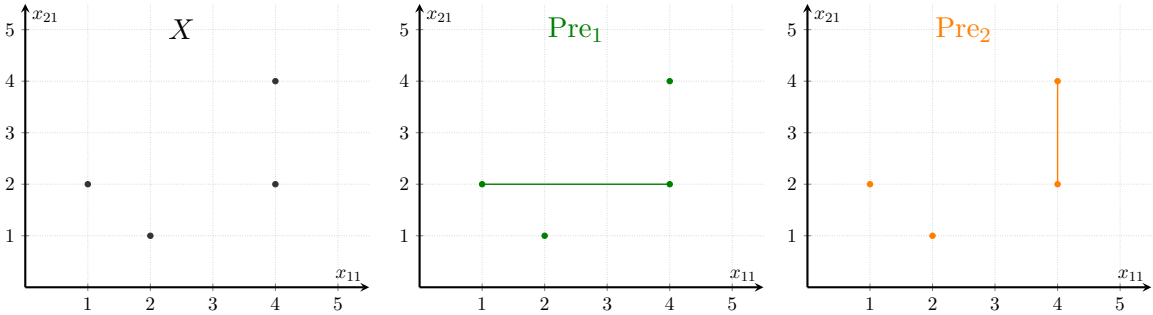
\begin{figure}[h!]
\begin{center}
\scalebox{\myscale}{
\begin{tikzpicture}
\begin{axis}[clip mode=individual, 
  axis lines=middle,
  x = \mysize, y = \mysize,
  xmin=0,xmax=5.5,ymin=0,ymax=5.5,
  xtick distance=1,
  ytick distance=1,
  xlabel=$x_{1}$,
  ylabel=$x_{2}$,
  grid=major,
  grid style={thin,densely dotted,black!20}]
\node at (axis cs:2,1) [circle, scale=0.3, draw=black!80,fill=black!80] {};

\node at (axis cs:1,2) [circle, scale=0.3, draw=black!80,fill=black!80] {};

\node at (axis cs:4,2) [circle, scale=0.3, draw=black!80,fill=black!80] {};
\node at (axis cs:4,4) [circle, scale=0.3, draw=black!80,fill=black!80] {};

\node at (axis cs:2.5,5) {\scalebox{\myoscale}{$X$}};

\end{axis}
\end{tikzpicture} }\hspace{2cm}
\scalebox{\myscale}{
\begin{tikzpicture}
\begin{axis}[clip mode=individual, 
  axis lines=middle,
    x = \mysize, y = \mysize,
  xmin=0,xmax=5.5,ymin=0,ymax=5.5,
  xtick distance=1,
  ytick distance=1,
  xlabel=$x_{1}$,
  ylabel=$x_{2}$,
  grid=major,
  grid style={thin,densely dotted,black!20}]
\node (t2) at (axis cs:2,1) [circle, scale=0.3, draw=darkgreen,fill=darkgreen] {};

\node (t1) at (axis cs:1,2) [circle, scale=0.3, draw=darkgreen,fill=darkgreen] {};

\node (t3) at (axis cs:4,2) [circle, scale=0.3, draw=darkgreen,fill=darkgreen] {};
\node (t4) at (axis cs:4,4) [circle, scale=0.3, draw=darkgreen,fill=darkgreen] {};
\draw[darkgreen,thick] (axis cs:1,2) -- (axis cs:4,2);
\node[darkgreen] at (axis cs:5,1.5) {};

\node[darkgreen] at (axis cs:2.5,5) {\scalebox{\myoscale}{$\Pre{1}$}};

\end{axis}
\end{tikzpicture}} \hspace{2cm}
\scalebox{\myscale}{
\begin{tikzpicture}
\begin{axis}[clip mode=individual, 
  axis lines=middle,
    x = \mysize, y = \mysize,
  xmin=0,xmax=5.5,ymin=0,ymax=5.5,
  xtick distance=1,
  ytick distance=1,
  xlabel=$x_{1}$,
  ylabel=$x_{2}$,
  grid=major,
  grid style={thin,densely dotted,black!20}]
\node (t2) at (axis cs:2,1) [circle, scale=0.3, draw=orange,fill=orange] {};

\node (t1) at (axis cs:1,2) [circle, scale=0.3, draw=orange,fill=orange] {};

\node (t3) at (axis cs:4,2) [circle, scale=0.3, draw=orange,fill=orange] {};
\node (t4) at (axis cs:4,4) [circle, scale=0.3, draw=orange,fill=orange] {};

\draw[orange,thick] (axis cs:4,2) -- (axis cs:4,4);
\node[orange] at (axis cs:4,4.5) {};
\node[orange] at (axis cs:2.5,5) {\scalebox{\myoscale}{$\Pre{2}$}};

\end{axis}
\end{tikzpicture}}
\end{center}
    \caption{Example for a 2-player jointly constrained GNEP $I$ w.r.t.~$X\subseteq \R^{(1,1)}$ represented by the four black dots in the first picture. The prescribed strategy sets $X_i^\conv(x_{-i}) = \conv(X_i(x_{-i})), x_{-i} \in \dom X_i$ which rule out the possibility for $I^\conv$ to be jointly constrained are represented in picture 2 and 3.}
    \label{fig: ExaSJC}
\end{figure}

In the above example, the instance $I$ is jointly constrained w.r.t.~$X$, where any $I^\conv \in \mathcal{I}^\conv$ can not be jointly constrained w.r.t.~some set $X^\conv$. This becomes evident when one assumes that $I^\conv$ would be jointly convex w.r.t.~some set $X^\conv$. Then $2 \in \dom X_1,\dom X_2$ and thus 
$(2,2) \in \conv(X_1(2)) \times 2 = X_1^\conv(2) \times 2$ would imply $(2,2) \in X^\conv$ which contradicts $(2,2) \notin \conv(X_2(2)) \times 2 = X_2^\conv(2) \times 2$.

As $X_i^\conv(x_{-i})$ for $x_{-i} \notin \dom X_i$ is not a priori determined, only  
the \emph{prescribed strategy sets}  of player $i \in N$,
$\Pre{i} := \mycup_{x_{-i} \in \dom X_i} \conv(X_i(x_{-i})) \times x_{-i}$,
 may prohibit the possibility for $\mathcal{I}^\conv$ to contain jointly constrained instances as the example in Figure~\ref{fig: ExaSJC} illustrates.  
This leads to the question whether or not we can adapt our convexification method
in order to obtain a jointly constrained convexification which still falls under our main Theorem~\ref{thm:main}. 
One naive approach would be to simply extend the strategy spaces of a convexified game $I^\conv \in \mathcal{I}^\conv$ leading to a an instance $I^\extt$ with 
$ X_i^\conv(x_{-i}) \subseteq X_i^\extt(x_{-i}) := \left\lbrace x_i \mid (x_i,x_{-i}) \in X^\extt\right \rbrace$
such that $I^\extt$ is jointly constrained w.r.t.~some  set $X^\extt$ and adjust the cost functions to $+\infty$ on the new strategies, that is
$\pi_i^\extt(x_i,x_{-i}):=  \phi_i(x_i,x_{-i})$ whenever $x_i \in X_i^\conv(x_{-i})$ and  $\pi_i^\extt(x_i,x_{-i}):= +\infty$ otherwise.
It is not hard to see that the equilibria of $I$ can be characterized by $I^\extt$ in the same fashion as in Theorem \ref{thm:main} with $I^\conv$. Yet this approach of extending $I^\conv$ seems computationally of limited interest as the extended cost functions do not have any regularity. One may try to extend the cost functions in a original-equilibria-preserving and smooth manner instead of just setting them to $+\infty$ outside of $X_i^\conv(x_{-i})$. Yet, it is not clear how to extend these functions reasonably in a computational regard as one wants as much regularity of the cost functions as possible while putting as little effort as possible  in the computation of the cost-functions themselves. We remark here that the cost functions $\phi_i(x_i,x_{-i})$ of any convexified instance are by Definition~\ref{def: conv}.\ref{enu: c2}.~only a priori determined on $\conv(X_i(x_{-i}))$ for $x_{-i} \in \dom X_i$ and thus a similar problem as described above occurs w.r.t.~the convexified cost functions $\phi_i$. But it is substantially easier to find any arbitrary smooth extension compared to  finding a smooth extension which preserves original GNE. On top of that, the functions $\phi_i(\cdot,x_{-i}):\conv(X_i(x_{-i})) \to \R, x_{-i} \in \dom X_i$ may have a natural and smooth extension to the whole domain, as it is the case for most quasi-linear GNEPs for example.
This gives rise to the question whether or not one can modify $I^\conv$ by  only extending the strategy spaces and, thus, without specifically tailoring the cost-functions to conserve original equilibria. However, one can quickly verify that this will in general lead to a loss of original GNE.  To see this, let's take a look back at the example in Figure~\ref{fig: ExaSJC}. Assume that the cost function  of player 2 is represented by $\phi_2(x_1,x_2) = -x_2$ on the whole $\R^2$. Let $I^\extt$ be a jointly constrained extension of $I^\conv$ as described above, but without changing the cost functions. Then $\{1,2\} \subseteq X_2^\extt(2)$ as $(1,2),(2,2) \in X_1^\conv(2)\times 2 \subseteq X^\extt$ and therefore the generalized Nash equilibrium $(x_1^*,x_2^*) = (2,1)$ for $I$  would not remain a GNE for the extension $I^\extt$. 

As the example in Figure~\ref{fig: ExaSJC} shows, for general jointly constrained GNEPs $I$, 
there may exist $x_{-i}\in \dom X_i$ with a subsequently prescribed convexified strategy set $X_i^\conv(x_{-i}) = \conv(X_i(x_{-i}))$ which prohibits the possibility for $\mathcal{I}^\conv$ to contain a  jointly constrained instance, showing that jointly constrainedness of $I$ is not sufficient for $\mathcal{I}^\conv$ to contain a jointly constrained instance.  
Perhaps surprising,
 the example in Figure~\ref{fig: ExaNotJC} illustrates that it is also not a necessary condition. 
\begin{figure}[h!]
\begin{center}
\scalebox{\myscale}{
\begin{tikzpicture}
\begin{axis}[clip mode=individual, 
  axis lines=middle,
  x = \mysize, y = \mysize,
  xmin=0,xmax=5.5,ymin=0,ymax=5.5,
  xtick distance=1,
  ytick distance=1,
  xlabel=$x_{1}$,
  ylabel=$x_{2}$,
  grid=major,
  grid style={thin,densely dotted,black!20}]
\node (t2) at (axis cs:2.5,1) [circle, scale=0.3, draw=darkgreen,fill=darkgreen] {};

\node (t3) at (axis cs:1,2.5) [circle, scale=0.3, draw=darkgreen,fill=darkgreen] {};
\node (t4) at (axis cs:4,2.5) [circle, scale=0.3, draw=darkgreen,fill=darkgreen] {};
\draw[darkgreen,thick] (axis cs:1,2.5) -- (axis cs:4,2.5);
\node[darkgreen] at (axis cs:5,1.5) {};

\node[darkgreen] at (axis cs:2.5,5) {\scalebox{\myoscale}{$\Pre{1}$}};

\end{axis}
\end{tikzpicture}}\hspace{2cm}
\scalebox{\myscale}{
\begin{tikzpicture}
\begin{axis}[clip mode=individual, 
  axis lines=middle,
  x = \mysize, y = \mysize,
  xmin=0,xmax=5.5,ymin=0,ymax=5.5,
  xtick distance=1,
  ytick distance=1,
  xlabel=$x_{1}$,
  ylabel=$x_{2}$,
  grid=major,
  grid style={thin,densely dotted,black!20}]
\node (t2) at (axis cs:2.5,1) [circle, scale=0.3, draw=orange,fill=orange] {};

\node (t1) at (axis cs:2.5,4) [circle, scale=0.3, draw=orange,fill=orange] {};

\node (t3) at (axis cs:1,2.5) [circle, scale=0.3, draw=orange,fill=orange] {};
\node (t4) at (axis cs:2.5,2.5) [circle, scale=0.3, draw=orange,fill=orange] {};

\draw[orange,thick] (axis cs:2.5,1) -- (axis cs:2.5,4);
\node[orange] at (axis cs:4,4.5) {};
\node[orange] at (axis cs:2.5,5) {\scalebox{\myoscale}{$\Pre{2}$}};

\end{axis}
\end{tikzpicture}}\hspace{2cm}
\scalebox{\myscale}{
\begin{tikzpicture}
\begin{axis}[clip mode=individual, 
  axis lines=middle,
 x = \mysize, y = \mysize,
  xmin=0,xmax=5.5,ymin=0,ymax=5.5,
  xtick distance=1,
  ytick distance=1,
  xlabel=$x_{1}$,
  ylabel=$x_{2}$,
  grid=major,
  grid style={thin,densely dotted,black!20}]
\node at (axis cs:2.5,1) [circle, scale=0.3, draw=black!80,fill=black!80] {};

\node at (axis cs:2.5,4) [circle, scale=0.3, draw=black!80,fill=black!80] {};

\node at (axis cs:1,2.5) [circle, scale=0.3, draw=black!80,fill=black!80] {};
\node at (axis cs:4,2.5) [circle, scale=0.3, draw=black!80,fill=black!80] {};

\node at (axis cs:2.5,5) {\scalebox{\myoscale}{$X^\conv$}};
\draw[thick] (axis cs:2.5,1) -- (axis cs:2.5,4);
\draw[thick] (axis cs:1,2.5) -- (axis cs:4,2.5);

\end{axis}
\end{tikzpicture} }
\end{center}
    \caption{Example for a 2-player GNEP $I$ which is not jointly constrained but admits a jointly constrained convexified instance $I^\conv \in \mathcal{I}^\conv$. The prescribed strategy sets are represented in picture 1 and 2 where the dots correspond to the original strategies. In the third picture is an example for a possible joint restriction set $X^\conv$.}
    \label{fig: ExaNotJC}
\end{figure}
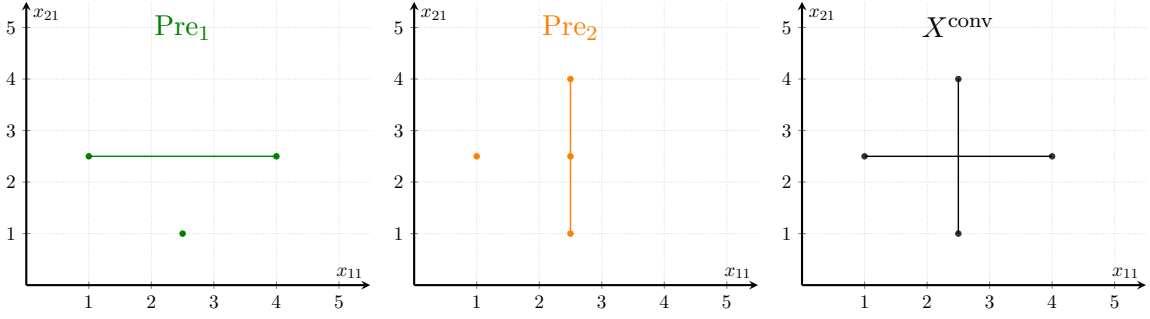

\subsection{$k$-restrictive-closed GNEPs} 
The insights of the previous subsection raise the question which instances $I$ admit jointly constrained instances in $ \mathcal{I}^\conv$ and which do not. 
In order to answer this question we define some necessary concepts in the following.
\begin{definition}
 For a vector $k=(k_1,\ldots,k_n)\in \N^n$ and  a set $X \subseteq \R^k$, we 
 \begin{itemize}
     \item define for $i\in [n],x_{-i} \in \R^{k_{-i}}$ the restriction $\res{X}{x_{-i}} := \left\{ \tilde{x} \in X \mid  \tilde{x} = (\tilde{x}_i,x_{-i}) \right\}$ of $X$ w.r.t.~$x_{-i}$.
     \item say that $X$ is \emph{$k$-convex}, if for all $i \in [n]$ and $x_{-i} \in \R^{k_{-i}}$, the restriction $\res{X}{x_{-i}}$ is convex. Note that any convex set $X$ is also $k$-convex as $\res{X}{x_{-i}} = X \cap (\R^{k_i}\times x_{-i})$ is the intersection of two convex sets in this case.
     \item define the $k$-convex hull of $X$ as the smallest $k$-convex set that contains $X$, that is, we define $\convk(X) := \bigcap\{Z \subseteq \R^k \mid Z \text{ is $k$-convex}, X\subseteq Z \}.$ 
\end{itemize}
\end{definition}
The concept of $k$-convex sets is a special case of so-called $\mathcal{O}$-convex sets (see e.g.~\cite{fink2012restricted}). A set in  $\R^d$ for some $d\in \N$ is $\mathcal{O}$-convex, if its intersection with every $\mathcal{O}$-line is connected, where an $\mathcal{O}$-line is a 1-dimensional intersection of several $\mathcal{O}$-hyperplanes. The latter are in turn hyperplanes that are parallel to some hyperplane contained in the orientation set $\mathcal{O}\subseteq \mathcal{P}(\R^d)$ which is a subset of the power set of $\R^d$ containing hyperplanes. 

The following lemma shows that $k$-convex sets are $\mathcal{O}$-convex sets for a certain orientation set $\mathcal{O}$. 
In particular, for the special case of $k=(1,1) \in \N^2$,  $k$-convexity reduces to orthogonal convexity (in 2 dimensions), see e.g.~\cite{OTTMANN1984157}.
\begin{lemma}\label{lem: oriconv}
For a vector $k=(k_1,\ldots,k_n)\in \N^n$, a set $X \subseteq \R^k$ is $k$-convex if and only if it is $\mathcal{O}$-convex for the orientation set containing the hyperplanes of the form $H_a = \{x \in \R^k \mid a^\top x = 0 \} \subseteq \R^k$ 
with $ a = (a_{ij})_{i \in N, j \in [k_{i}]}\in \R^k$ having zero entries for all $i \in N$ and corresponding $j \in [k_i]$ 
except for one $i^* \in N$, i.e.
\begin{align*}
    \mathcal{O}:= \Big\{H_a\subseteq \R^k \Big \vert \; \exists i^* \in N: 
     a_{ij}= 0, i\neq i^*, j\in [k_i]  \Big\}.
\end{align*}
\end{lemma}
\begin{proof}
We first show that the set of $\mathcal{O}$-lines is given by  
\begin{align}\label{eq: Olines}
\mathcal{O}\text{-lines} = \Big \{\{\lambda\cdot( \tilde{x}_i,0_{-i}) + x \in \R^{k} \mid \lambda \in \R\} \mid i \in N, \tilde{x}_i \in \R^{k_i}, x \in \R^k\Big\}
\end{align}
where $0_{-i}\in \R^{k_{-i}}$ denotes the vector only consisting of zeros. For the inclusion $\supseteq$, let $i \in N$ and $\tilde{x}_i \in \R^{k_i}$ be arbitrary. Let $A\in \R^{k_i - 1\times k_i}$ be a matrix with $\ker(A) = \{\lambda \tilde{x}_i \mid \lambda \in \R\}$ and denote by $A_j$ the $j$-th row (interpreted as a column vector). Then the hyperplanes of the form $\{x \in \R^k \mid (A_j,0_{-i})^\top x = 0 \} \in \mathcal{O}$ for all $j \in [k_i - 1]$ are contained in $\mathcal{O}$. Similarly, we have $\{x \in \R^k \mid e_{lj}^\top x = 0 \} \in \mathcal{O}$ for all $l\in N, l \neq i$ and $j \in [k_l]$ where $e_{lj}$ denotes the standard basis vector with a $1$ at the $lj$-th position. Intersecting all these above mentioned hyperplanes results in  
\begin{align*}
    \Big(\quad\bigcap_{\mathclap{j \in [k_i -1]}} \{x \in \R^k \mid (A_j,0_{-i})^\top x = 0 \} \Big) \cap \Big(\quad\bigcap_{\mathclap{\substack{l \in N,l\neq i\\j \in [k_l]}}} \{x \in \R^k \mid e_{lj}^\top x = 0 \}  \Big) = 
    \{\lambda( \tilde{x}_i,0_{-i}) \in \R^{k} \mid \lambda \in \R\}
\end{align*}
which shows $\supseteq$ in \eqref{eq: Olines} as $\mathcal{O}$-lines are lines that are parallel to some line constructible as the intersection of hyperplanes in $\mathcal{O}$.

In order to show $\subseteq$ in \eqref{eq: Olines}, let $\bigcap_{s \in [L]} \mathcal{H}_s = \{\lambda \tilde{x} \in \R^k \mid \lambda \in \R \}$ for some hyperplanes $\mathcal{H}_s \in \mathcal{O}, s\in[L]$ for a $L \in \N$. By the definition of $\mathcal{O}$, we can represent the intersection as the linear equation system $\mathrm{diag}(A_1,\ldots,A_n)x = 0$ for a block diagonal matrix consisting of some matrices $A_i \in \R^{L_i \times k_i}, i \in N$ with $\sum_{i \in N}L_i = L$. Thus, 
$\ker(A_i) = \{\lambda \tilde{x}_i \mid \lambda \in \R\}$ needs to hold for any $i \in N$ which shows that only one $i \in N$ may exist with $\tilde{x}_i \neq 0_{k_i}$ as otherwise the intersection of the hyperplanes $\bigcap_{s \in [L]} \mathcal{H}_s $ would not be 1-dimensional. Thus, $\subseteq$ in \eqref{eq: Olines} holds.

Now we are ready to prove the equivalence of $k$-convexity and $\mathcal{O}$-convexity. We start with the only if direction. Let $X\subseteq \R^k$ be $k$-convex and let $\mathcal{L}:=\{\lambda\cdot( \tilde{x}_i,0_{-i}) + x \in \R^{k} \mid \lambda \in \R\}$ be an arbitrary $\mathcal{O}$-line. Then $\mathcal{L}\cap X = \mathcal{L}\cap \res{X}{x_{-i}}$ which shows by $X$ being $k$-convex that $\mathcal{L}\cap X$ is a convex set as it is the intersection of two convex sets and thus is in particular connected.

For the if direction, let $X$ be an $\mathcal{O}$-convex set. Now assume for contradiction that there exists $x_{-i} \in \R^{k_{-i}}$ such that $\res{X}{x_{-i}}$ is not convex, that is, there exist $x^1:=(x_i^1,x_{-i}),x^2:=(x_i^2,x_{-i}) \in \res{X}{x_{-i}}$ and $\alpha \in (0,1)$ with $x^{\alpha}:= \alpha x^1 + (1- \alpha)x^2 \notin \res{X}{x_{-i}}$, i.e. $x^\alpha \notin X$. This contradicts that $X$ is $\mathcal{O}$-convex as $\mathcal{L}:= \{\lambda\cdot(x_i^1 - x_i^2,0_{-i}) + (x_i^2,x_{-i}) \in \R^{k} \mid \lambda \in \R\}$ is an $\mathcal{O}$-line with $\mathcal{L}\cap X$ not being connected as $x^\alpha \in \mathcal{L}$ and $x^1,x^2 \in \mathcal{L}\cap X$ but $x^\alpha \notin X$.
\end{proof}

As already noted in the previous section, the jointly constrainedness of $I$ is not necessary in order for $\mathcal{I}^\conv$ to contain a jointly constrained instance which leads to the following definition of  what we call \emph{$k$-restrictive-closed} GNEPs.
\begin{definition}\label{def: kresclos}
Let $I$ be an instance of the GNEP. We define the \emph{complete} (relevant) \emph{strategy set} of player $i \in N$ by  $\Str{i} := \mycup_{x_{-i} \in \dom X_i} X_i(x_{-i}) \times x_{-i}$ 
 and denote their union over all players by $\Stra := \mycup_{i \in N}\Str{i}$.
An instance $I$ is called \emph{$k$-restrictive-closed} (w.r.t.~the $\convk$-operator), if for all $i \in N$ and $x_{-i} \in \dom X_i$, the following equality holds: 
     \begin{align}\label{eq: kresclo}
         \convk\big(\res{\Str{i}}{x_{-i}}\big) = \res{\convk(\Stra)}{x_{-i}}.
     \end{align}
\end{definition}

 The above concept of $k$-restrictive-closedeness requires that for fixed $x_{-i}\in\dom X_i$, the $k$-convex hull of the restriction
 of $\Str{i}$  w.r.t.~$x_{-i}$ is equal to the restriction of $\convk(\Stra)$ w.r.t.~$x_{-i}$.
  Remark that in the special case of a jointly constrained instance $I$ w.r.t.~a restriction set $X$, the  complete strategy set of each player $i \in N$ and subsequently their union equals the restriction set $X$, i.e.~$\Str{i} = \Stra = X$. Thus, in the case of jointly constrained instances, one can identify $k$-restrictive-closedness by only investigating the joint restriction set $X$ as~\eqref{eq: kresclo} becomes a condition solely for $X$.
  
 We note in the following proposition that the inclusion $\subseteq$ in~\eqref{eq: kresclo} always holds. 
 \begin{proposition}\label{prop: incl} 
 In Definition~\ref{def: kresclos}, the inclusion $\subseteq$ in \eqref{eq: kresclo} 
     holds for all $i \in N, x_{-i} \in \dom X_i$.
 \end{proposition}
 \begin{proof}
For  $i \in N$ and $x_{-i} \in \dom X_i$ arbitrary, we have by the definition of $\convk(\Stra)$ that the restriction $\res{\convk(\Stra)}{x_{-i}}$ is convex and therefore also $k$-convex. As $\res{\Str{i}}{x_{-i}}\subseteq \res{\convk(\Stra)}{x_{-i}}$, the results follows by the definition of the $\convk$-operator. 
 \end{proof}

 The following theorem gives various equivalent characterizations of $k$-restrictive-closedeness. 
 The first two equivalences $\ref{enu: kresclos}.\Leftrightarrow\ref{enu: Ex1}.$~and $\ref{enu: kresclos}.\Leftrightarrow\ref{enu: Ex2}.$ are interesting as they give geometric interpretations of $k$-restrictive-closed GNEPs. For instance, by $\ref{enu: kresclos}.\Leftrightarrow\ref{enu: Ex2}.$~one can easily verify that the example in Figure~\ref{fig: ExaSJC} is not $k$-restrictive-closed, as $x_{-2} = x_1 :=2 \in \dom X_2$, yet the point $(x_1,x_2) := (2,2)$ of the restriction $\res{\Pre{1}}{x_{-2}} = \{(2,1),(2,2)\}$ is not contained in $\res{\Pre{2}}{x_{-2}} = \res{\Pre{2}}{2} = \{(2,1)\}$. Similar, by  $\ref{enu: kresclos}.\Leftrightarrow\ref{enu: Ex1}.$~one can immediately verify that the example in Figure~\ref{fig:kovskoe} (see below) is not $k$-restrictive-closed as $E\left( \res{\convk(\Stra)}{x_{-i}} \right) = \{(3,1),(3,4)\} \not\subseteq \Str{i} = X$ for $i = 2$ and $x_{-2} = x_1 := 3$.
 
 The last equivalence  $\ref{enu: kresclos}.\Leftrightarrow\ref{enu: prop}.$ 
 will allows us to show in the subsequent Theorem~\ref{thm:oecon} that $k$-restrictive-closed GNEPs with nonempty $\mathcal{I}^\conv$ are exactly the GNEPs $I$ which admit a jointly constrained convexification $I^\conv \in \mathcal{I}^\conv$. 

 \clearpage
  \begin{theorem}\label{thm: equis}
 Let $I$ be an instance of the GNEP. Then, the following statements are equivalent: 
   \begin{enumerate}
     \item $I$ is $k$-restrictive-closed. \label{enu: kresclos}
     \item $E\left( \res{\convk(\Stra)}{x_{-i}} \right) \subseteq \Str{i}$ for all $i \in N$ and $x_{-i} \in \dom X_i$. \label{enu: Ex1}
     \item $\res{\Pre{i}}{x_{-j}}\subseteq \res{\Pre{j}}{x_{-j}}$ for all $i,j \in N$ and $x_{-j} \in \dom X_j$. \label{enu: Ex2}
     \item For all $x \in\R^k$ and $i \in N$ the following implication holds: \label{enu: prop}
     \begin{align*}
         x_i \in \conv(X_i(x_{-i})) \wedge x_{-i} \in \dom X_i&& \Rightarrow &&\forall j \in N: \, x_j \in \conv(X_j(x_{-j})) \lor   x_{-j} \notin \dom X_j.  
     \end{align*}

 \end{enumerate}
 \end{theorem}
 \begin{proof}
 We will prove the theorem by showing the equivalences $\ref{enu: kresclos}.\Leftrightarrow\ref{enu: Ex1}.\Leftrightarrow\ref{enu: Ex2}.\Leftrightarrow\ref{enu: prop}.$ in this order. Before doing so, we observe that the following inclusion is valid: 
 \begin{align}\label{eq: PreincS}
    \mycup_{i \in N}\Pre{i} \subseteq \convk(S).
 \end{align}
 To see this, we remark that we have the following equations for the prescribed strategy sets $\Pre{i}$ of player $i \in N$: 
 \begin{align*}
    \Pre{i} &:= \mycup_{x_{-i} \in \dom X_i} \conv(X_i(x_{-i})) \times x_{-i} = \mycup_{x_{-i} \in \dom X_i} \conv\left(\res{\Str{i}}{x_{-i}}\right)\\&= \mycup_{x_{-i} \in \dom X_i} \convk\big(\res{\Str{i}}{x_{-i}}\big)
 \end{align*}
where the last equality follows as for arbitrary
$i \in N$ and $x_{-i} \in \R^{k_{-i}}$, the equality
\begin{align}\label{eq: convkeqconv}
    \convk\big(\res{\Str{i}}{x_{-i}}\big) = \conv\big(\res{\Str{i}}{x_{-i}}\big)
\end{align}
holds. Clearly, $\subseteq$ holds in~\eqref{eq: convkeqconv} since any convex set is also $k$-convex. To see that $\supseteq$ holds, we argue as follows: For any $k$-convex set $Z \supseteq \res{\Str{i}}{x_{-i}}$, the set $\res{Z}{x_{-i}}$ needs to be convex by definition. Since the latter also contains $\res{\Str{i}}{x_{-i}}$, the inclusion $\conv\big(\res{\Str{i}}{x_{-i}}\big) \subseteq \res{Z}{x_{-i}}\subseteq Z$ holds which shows $\supseteq$ in~\eqref{eq: convkeqconv}. 
 
 Thus, the inclusion in~\eqref{eq: PreincS} follows by Proposition~\ref{prop: incl} and the following representation of $\convk(\Stra)$:
\begin{align*}
    \convk(\Stra)= \mycup_{i\in N} \mycup_{x_{-i} \in \R^{k_{\scalebox{0.33}{$-i$}}}} \res{\convk(\Stra)}{x_{-i}}. 
\end{align*}
Now we are ready to prove the equivalences:  
 
  $\ref{enu: kresclos}.\Leftrightarrow\ref{enu: Ex1}.$: Let $i \in N$ and $x_{-i}\in\dom X_i$ be arbitrary.
  
 $\ref{enu: kresclos}.\Rightarrow\ref{enu: Ex1}.$:  By \ref{enu: kresclos}.~together with the equality in~\eqref{eq: convkeqconv}, we get 
\begin{align*}
    E\left( \res{\convk(\Stra)}{x_{-i}} \right) = E\Big( \conv\big(\res{\Str{i}}{x_{-i}} \big)\Big) \subseteq \res{\Str{i}}{x_{-i}} \subseteq \Str{i}.
\end{align*}
 
   $\ref{enu: kresclos}.\Leftarrow\ref{enu: Ex1}.$: The following implications hold:
 \begin{alignat*}{2}
   \ref{enu: Ex1}. \quad \Rightarrow \quad &E\left( \res{\convk(\Stra)}{x_{-i}} \right) &&\subseteq \quad\res{\Str{i}}{x_{-i}} \\ 
   \Rightarrow\quad &\conv\Big(E\big( \res{\convk(\Stra)}{x_{-i}} \big)\Big) \; &&\subseteq \quad\conv\big(\res{\Str{i}}{x_{-i}}\big) \\ 
   \Rightarrow \quad&\res{\convk(\Stra)}{x_{-i}} &&\subseteq\quad \convk\big(\res{\Str{i}}{x_{-i}}\big) 
 \end{alignat*}
 where the last inclusion follows by~\eqref{eq: convkeqconv} and the convexity of $\res{\convk(\Stra)}{x_{-i}}$. Since the inclusion $\supseteq$ in the last line always holds by Proposition~\ref{prop: incl}, the claim follows. 
 
 
 $\ref{enu: Ex1}.\Rightarrow\ref{enu: Ex2}.$: Assume for contradiction that there exists $i,j \in N$, $x_{-j} \in\dom X_j$ and a strategy profile $x:=(x_j,x_{-j}) \in \res{\Pre{i}}{x_{-j}} \setminus \res{\Pre{j}}{x_{-j}}$. By~\eqref{eq: PreincS}, we know that $\res{\Pre{i}}{x_{-j}} \subseteq \res{\convk(\Stra)}{x_{-j}}$. As \ref{enu: Ex1}.~and $\Str{j}\subseteq \Pre{j}$ holds, we know $x \notin E\big( \res{\convk(\Stra)}{x_{-j}} \big)$. Thus, there exists a convex combination $x = \sum_{s=1}^L \lambda_s (x^s_j,x_{-j})$ with $(x^1_j,x_{-j}),\ldots,(x^L_j,x_{-j})\neq x$ \emph{extreme points} of $\res{\convk(\Stra)}{x_{-j}}$ and $\lambda \in \Lambda_L := \lbrace \alpha \in \mathbb{R}^L_{\geq 0}\mid \sum_{k=1}^{L}\alpha_k =1 \rbrace$ for some $L \in \N$. By \ref{enu: Ex1}.~we have for all $s \in [L]$ that $(x^s_j,x_{-j}) \in \Str{j}$, that is, $x^s_j \in X_j(x_{-j})$. This in turn implies that $x_j \in \conv(X_j(x_{-j}))$ which together with $x_{-j} \in\dom X_j$ shows that $x \in \Pre{j}$ which contradicts our assumption.
 
   $\ref{enu: Ex1}.\Leftarrow\ref{enu: Ex2}.$: 
 We first show that $\convk(\Stra) = \mycup_{j \in N}\Pre{j}$ holds. By our observation~\eqref{eq: PreincS}, we know that $\supseteq$ always holds. Thus it suffices to show that  $\mycup_{j \in N}\Pre{j}$ is $k$-convex since $\Stra \subseteq \mycup_{j\in N}\Pre{j}$. 
 In order to do so we have to show that $\res{\mycup_{j \in N}\Pre{j}}{x_{-i}}$ is convex for all $i \in N$ and $x_{-i}\in \R^{k_{-i}}$. For the latter set we have:
     \begin{align}\textstyle
       \res{\mycup_{j \in N}\Pre{j}}{x_{-i}} 
       =\mycup_{j \in N}\res{\Pre{j}}{x_{-i}}  
       = \mycup_{j \in N} \res{\Pre{i}}{x_{-i}} 
       =\res{\Pre{i}}{x_{-i}}  \label{eq: unPrePre}
   \end{align}
   where the penultimate equality is valid by \ref{enu: Ex2}.
   Since $\res{\Pre{i}}{x_{-i}}$ is clearly convex by definition, the set $\res{\mycup_{j \in N}\Pre{j}}{x_{-i}}$ is in fact convex. Therefore $\convk(\Stra) = \mycup_{j \in N}\Pre{j}$ holds which implies for all $i \in N$ and $x_{-i} \in \dom X_i$:
   \begin{align*} \textstyle
       E\big(\res{\convk(\Stra)}{x_{-i}} \big)  = E\big(\res{\mycup_{j \in N}\Pre{j}}{x_{-i}} \big) = E\big(\res{\Pre{i}}{x_{-i}}\big) \subseteq \res{\Str{i}}{x_{-i}}
   \end{align*}
   where the penultimate equality follows by~\eqref{eq: unPrePre} and the last equality by the definition of $\Pre{i},\Str{i}$.

  $\ref{enu: Ex2}.\Rightarrow\ref{enu: prop}.$: Let $x \in \R^k$ with $x_i \in \conv(X_i(x_{-i}))$ and $x_{-i} \in \dom X_i$. In what follows, let $j \in N$ be arbitrary and assume $x_{-j} \in \dom X_j$. Since \ref{enu: Ex2}.~and $x \in \res{\Pre{i}}{x_{-j}}$ hold, it follows that $x \in \res{\Pre{j}}{x_{-j}}$ which implies $x_j \in \conv(X_j(x_{-j}))$.
  
$\ref{enu: Ex2}.\Leftarrow\ref{enu: prop}.$: Assume for contradiction that there exists $i,j \in N$, $x_{-j} \in \dom X_j$ and a strategy profile $(\bar{x}_j,x_{-j}) \in \res{\Pre{i}}{x_{-j}}  \setminus \res{\Pre{j}}{x_{-j}}$. By $(\bar{x}_j,x_{-j}) \in \Pre{i}$ we may infer that $x_i \in \conv(X_i(\bar{x}_j,x_{-ij}))$ and $(\bar{x}_j,x_{-ij})\in \dom X_i$.
Subsequently by \ref{enu: prop}., $\bar{x}_j \in \conv(X_j(x_{-j}))$ since $x_{-j}\in \dom X_j$. This implies $(\bar{x}_j,x_{-j}) \in \Pre{j}$ which contradicts our assumption.
 \end{proof}
  
  From the proof follows directly the following necessary condition for $I$ to be $k$-restrictive-closed.
\begin{lemma}\label{lem: krcPre}
  If $I$ is $k$-restrictive-closed, then the union of the prescribed strategy sets over all players equals the  $\convk$-hull of $\Stra$, i.e. $\convk(\Stra)= \mycup_{i\in N} \Pre{i}$.
\end{lemma}

However, note that the property $\convk(\Stra) = \mycup_{i \in N}\Pre{i}$ is not sufficient for $I$ to be $k$-restrictive-closed as the example in Figure~\ref{fig:kovskoe} shows.  

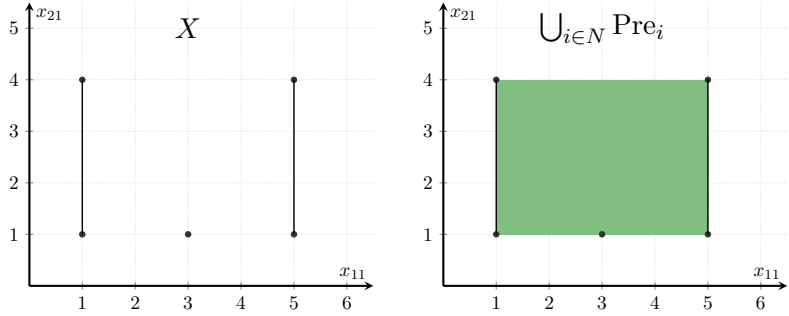
\begin{figure}[!h]
\begin{center}
\scalebox{\myscale}{
\begin{tikzpicture}
\begin{axis}[clip mode=individual, 
  axis lines=middle,
  x = \mysize, y = \mysize,
  xmin=0,xmax=6.5,ymin=0,ymax=5.5,
  xtick distance=1,
  ytick distance=1,
  xlabel=$x_{1}$,
  ylabel=$x_{2}$,
  grid=major,
  grid style={thin,densely dotted,black!20}]
\node at (axis cs:1,1) [circle, scale=0.3, draw=black!80,fill=black!80] {};
\node at (axis cs:3,1) [circle, scale=0.3, draw=black!80,fill=black!80] {};
\node at (axis cs:5,1) [circle, scale=0.3, draw=black!80,fill=black!80] {};
\node at (axis cs:1,4) [circle, scale=0.3, draw=black!80,fill=black!80] {};
\node at (axis cs:5,4) [circle, scale=0.3, draw=black!80,fill=black!80] {};
\draw[thick] (axis cs:1,1) -- (axis cs:1,4);
\draw[thick] (axis cs:5,1) -- (axis cs:5,4);
\node at (axis cs:3,5) {\scalebox{\myoscale}{$X$}};

\end{axis}
\end{tikzpicture} }\hspace{2cm}
\scalebox{\myscale}{
\begin{tikzpicture}
\begin{axis}[clip mode=individual, 
  axis lines=middle,
  x = \mysize, y = \mysize,
  xmin=0,xmax=6.5,ymin=0,ymax=5.5,
  xtick distance=1,
  ytick distance=1,
  xlabel=$x_{1}$,
  ylabel=$x_{2}$,
  grid=major,
  grid style={thin,densely dotted,black!20}]
\filldraw [fill = darkgreen!50, draw = none] (axis cs:1,1) rectangle (axis cs:5,4);
\node at (axis cs:1,1) [circle, scale=0.3, draw=black!80,fill=black!80] {};
\node at (axis cs:3,1) [circle, scale=0.3, draw=black!80,fill=black!80] {};
\node at (axis cs:5,1) [circle, scale=0.3, draw=black!80,fill=black!80] {};
\node at (axis cs:1,4) [circle, scale=0.3, draw=black!80,fill=black!80] {};
\node at (axis cs:5,4) [circle, scale=0.3, draw=black!80,fill=black!80] {};
\draw[thick] (axis cs:1,1) -- (axis cs:1,4);
\draw[thick] (axis cs:5,1) -- (axis cs:5,4);

\node at (axis cs:3,5) {\scalebox{\myoscale}{$\mycup_{i \in N}\Pre{i}$}};
\end{axis}
\end{tikzpicture}}
\end{center}
    \caption{A 2-player jointly constrained GNEP $I$ w.r.t.~$X\subseteq \R^{k}, k := (1,1)$ represented by the black set in the first picture. In picture 2 the union of the prescribed strategy sets is represented which equals the $k$-convex hull (and even the regular convex hull) of $X$. Yet, for $i = 2$ and $x_{-2} = x_1 := 3$ we have  $\convk\big(\res{X}{x_{-i}}\big) = \{(3,1)\} \subsetneq \{3\}\times [1,4] =\res{\convk(X)}{x_{-i}}$. Remember that $X = \Str{i}=\Stra$ holds for any $i \in N$ in a jointly constrained instance.}
    \label{fig:kovskoe}
\end{figure}
 
 \begin{theorem}\label{thm:oecon}
Let $I$ be an instance of the GNEP with $\mathcal{I}^\conv(I) \neq \emptyset$. Then $I$ is
$k$-restrictive-closed if and only if the convexification $\mathcal{I}^\conv$ contains a jointly constrained instance. 
In this case, $\mathcal{I}^\conv$ contains for any set $X^\conv$ that satisfies
\begin{align}\label{eq:strongJCthm}
    \res{X^\conv}{x_{-i}} = \res{\convk(\Stra)}{x_{-i}} && \text{ for all } i \in N \text{ and }  x_{-i}\in \dom X_i,
\end{align}
 a jointly constrained instance with $X^\conv$ as its joint restriction set.
\end{theorem}
\begin{proof}
We start with the only if direction. 
Let $I^\conv$ be a jointly constrained instance w.r.t.~any restriction set $X^\conv$ as described above. Furthermore, let the cost functions of $I^\conv$ fulfill the requirements (Definition~\ref{def: conv}.\ref{enu: c2}.) for $I^\conv$ to belong to $\mathcal{I}^\conv$. 
{ Such cost functions exist due to $\mathcal{I}^\conv(I) \neq \emptyset$. } 
We want to prove that $I^\conv \in \mathcal{I}^\conv$. Thus, we have to show that the condition of Definition~\ref{def: conv}.\ref{enu: c1}.~for arbitrary $i \in N$ and $ x_{-i} \in \dom X_i$ is fulfilled, that is:
\begin{align*}
  X_i^\conv(x_{-i}):=\left\lbrace x_i \in \R^{k_i} \mid (x_i,x_{-i}) \in X^\conv \right \rbrace \overset{!}{=} \conv(X_i(x_{-i})) 
\end{align*}
which is equivalent to 
$\res{X^\conv}{x_{-i}} = \conv\big(\res{\Str{i}}{x_{i}}\big).$
The latter equality is valid due to $I$ being $k$-restrictive-closed and $X^\conv$ fulfilling~\eqref{eq:strongJCthm}.

For the if direction, let $I^\conv \in \mathcal{I}^\conv$ be jointly constrained w.r.t.~$X^\conv$. Then for any $x \in \R^k$ with $x_i \in \conv(X_i(x_{-i}))$ and $x_{-i}\in \dom X_i$ for some $i \in N$ we have $x_i \in X_i^\conv(x_{-i})$ and by the jointly constrainedness of $I^\conv$ that $(x_i,x_{-i}) \in X^\conv$. Again by the jointly constrainedness of $I^\conv$, we get for all $j \in N$ that $x_j \in X_j^\conv(x_{-j})$ which implies that $x_j \in \conv(X_j(x_{-j}))$ if $x_{-j} \in \dom X_j$. Therefore Theorem~\ref{thm: equis}($\ref{enu: prop}. \Rightarrow \ref{enu: resclo}.$) shows that $I$ is $k$-restrictive-closed. 
\end{proof}

\begin{figure}[!ht]
\begin{center}
\scalebox{\myscale}{
\begin{tikzpicture}
\begin{axis}[clip mode=individual, 
  axis lines=middle,
  x= \mysize, y = \mysize,
  xmin=0,xmax=5.5,ymin=0,ymax=5.5,
  xtick distance=1,
  ytick distance=1,
  xlabel=$x_{1}$,
  ylabel=$x_{2}$,
  grid=major,
  grid style={thin,densely dotted,black!20}]
\node at (axis cs:2,2) [circle, scale=0.3, draw=black!80,fill=black!80] {};

\node at (axis cs:1,2) [circle, scale=0.3, draw=black!80,fill=black!80] {};

\node at (axis cs:4,2) [circle, scale=0.3, draw=black!80,fill=black!80] {};
\node at (axis cs:4,4) [circle, scale=0.3, draw=black!80,fill=black!80] {};

\node at (axis cs:2.5,5) {\scalebox{\myoscale}{$X$}};

\end{axis}
\end{tikzpicture} }\hspace{2cm}
\scalebox{\myscale}{
\begin{tikzpicture}
\begin{axis}[clip mode=individual, 
  axis lines=middle,
 x= \mysize, y = \mysize,
  xmin=0,xmax=5.5,ymin=0,ymax=5.5,
  xtick distance=1,
  ytick distance=1,
  xlabel=$x_{1}$,
  ylabel=$x_{2}$,
  grid=major,
  grid style={thin,densely dotted,black!20}]
\node at (axis cs:2,2) [circle, scale=0.3, draw=black!80,fill=black!80] {};

\node at (axis cs:1,2) [circle, scale=0.3, draw=black!80,fill=black!80] {};

\node at (axis cs:4,2) [circle, scale=0.3, draw=black!80,fill=black!80] {};
\node at (axis cs:4,4) [circle, scale=0.3, draw=black!80,fill=black!80] {};

\node[blue!50] at (axis cs:3,4.5) {\scalebox{1}{$X^\conv = \conv^{(1,1)}(X)$}};
\draw[blue!50,thick] (axis cs:1,2) -- (axis cs:4,2);
\draw[blue!50,thick] (axis cs:4,2) -- (axis cs:4,4);

\end{axis}
\end{tikzpicture}}\hspace{2cm}
\scalebox{\myscale}{
\begin{tikzpicture}
\begin{axis}[clip mode=individual, 
  axis lines=middle,
  x= \mysize, y = \mysize,
  xmin=0,xmax=5.5,ymin=0,ymax=5.5,
  xtick distance=1,
  ytick distance=1,
  xlabel=$x_{1}$,
  ylabel=$x_{2}$,
  grid=major,
  grid style={thin,densely dotted,black!20}]
\node at (axis cs:2,2) [circle, scale=0.3, draw=black!80,fill=black!80] {};

\node at (axis cs:1,2) [circle, scale=0.3, draw=black!80,fill=black!80] {};

\node at (axis cs:4,2) [circle, scale=0.3, draw=black!80,fill=black!80] {};
\node at (axis cs:4,4) [circle, scale=0.3, draw=black!80,fill=black!80] {};

\addplot[name path=g, domain=2:4, samples=2,draw=none] {x};
\addplot[name path=h, domain=2:4, samples=2,draw=none] {2 };

    \addplot [
        thick,
        color=blue,
        fill=blue, 
        fill opacity=0.25
    ]
    fill between[
        of=g and h,
        soft clip={domain=1:4},
    ];

\node[blue!50] at (axis cs:2.5,3.6) {\scalebox{\myoscale}{$X^\conv$}};
\draw[blue!50,thick] (axis cs:1,2) -- (axis cs:2,2);

\end{axis}
\end{tikzpicture}}
\end{center}
    \caption{Example for a 2-player jointly constrained GNEP $I$ w.r.t.~a $(1,1)$-restrictive-closed $X\subseteq \R^{(1,1)}$ represented by the four black dots in the first picture. In picture 2 and 3 are two suitable choices of $X^\conv$ which fulfill~\eqref{eq:strongJCthm}.}
\end{figure}
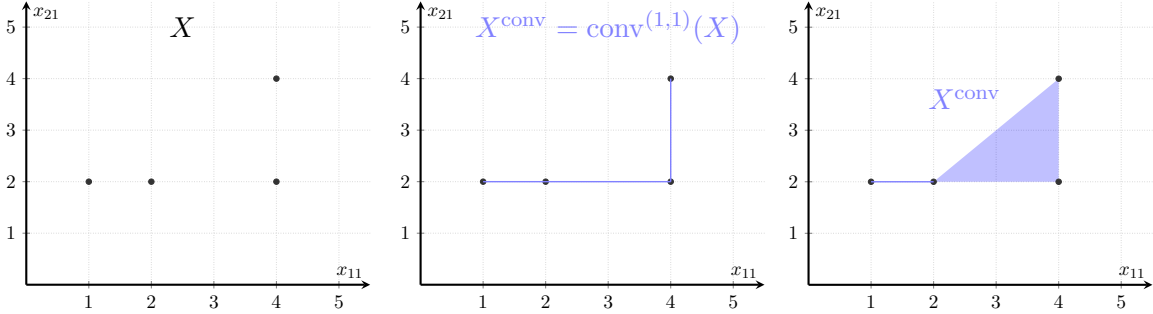

\subsection{Restrictive-closed GNEPs} 
Jointly constrained GNEPs w.r.t.~a convex restriction set are often referred to as \emph{jointly convex} in the literature and constitute one of the best understood subclasses of the GNEP. Thus, for $k$-restrictive-closed GNEPs, the question arises whether or not we can chose a \emph{convex} restriction set $X^\conv$.
However, for general $k$-restrictive-closed GNEPs $I$, 
there may exist $x_{-i}\in \dom X_i$ with a subsequently prescribed convexified strategy set $X_i^\conv(x_{-i}) = \conv(X_i(x_{-i}))$ which prohibit the possibility for $I^\conv$ to be jointly convex w.r.t.~some convex set $X^\conv$. 
An example for such a situation is given in Figure~\ref{fig: ExaJointConstr} by $x_{-2}:=2 \in \dom X_2$ and $X_2^\conv(2)$. To see this, assume $I^\conv\in \mathcal{I}^\conv$ was jointly convex w.r.t.~some convex set $X^\conv$. 
The union of the complete strategy sets of the players in the convexified instance $\Stra^\conv$ is then given by $\Stra^\conv= X^\conv$. 
Since $X^\conv$ is convex and $\Stra \subseteq \Stra^\conv$ clearly holds, we even have $\conv(\Stra) \subseteq X^\conv$.  Thus, we get: 
\begin{alignat*}{2}
X_2^\conv(2) &:= \conv(X_2(2))= \{2\} && \text{ by } I^\conv \in \mathcal{I}^\conv \\
& \subsetneq [2,8/3] = \left\lbrace x_2 \in \mathbb{R} \mid (2,x_2) \in \conv(\Stra)\right\rbrace \quad && \\
&\subseteq \left\lbrace x_2 \in \mathbb{R} \mid (2,x_2) \in X^\conv\right\rbrace && \text{ by } \conv(\Stra) \subseteq X^\conv 
\end{alignat*}
which contradicts that $I^\conv$ is jointly convex w.r.t.~$X^\conv$.

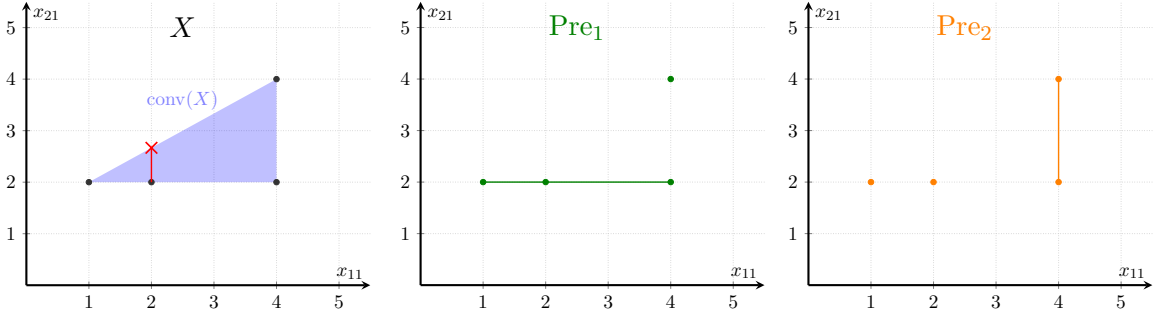
\begin{figure}[!h]
\begin{center}
\scalebox{\myscale}{
\begin{tikzpicture}
\begin{axis}[clip mode=individual, 
  axis lines=middle,
x= \mysize, y = \mysize,
  xmin=0,xmax=5.5,ymin=0,ymax=5.5,
  xtick distance=1,
  ytick distance=1,
  xlabel=$x_{1}$,
  ylabel=$x_{2}$,
  grid=major,
  grid style={thin,densely dotted,black!20}]
\node at (axis cs:2,2) [circle, scale=0.3, draw=black!80,fill=black!80] {};

\node at (axis cs:1,2) [circle, scale=0.3, draw=black!80,fill=black!80] {};

\node at (axis cs:4,2) [circle, scale=0.3, draw=black!80,fill=black!80] {};
\node at (axis cs:4,4) [circle, scale=0.3, draw=black!80,fill=black!80] {};

\addplot[name path=g, domain=1:4, samples=2,draw=none] {2*x/3 + 4/3};
\addplot[name path=h, domain=2:4, samples=2,draw=none] {2 };

    \addplot [
        thick,
        color=blue,
        fill=blue, 
        fill opacity=0.25
    ]
    fill between[
        of=g and h,
        soft clip={domain=1:4},
    ];
\draw[red,thick] (axis cs:2,2) -- (axis cs:2,8/3);
\node[red] at (axis cs:2,8/3) {\Cross};
\node[blue!50] at (axis cs:2,3.6) {\scalebox{1}{$\conv(X)$}};
\node at (axis cs:2.5,5) {\scalebox{\myoscale}{$X$}};

\end{axis}
\end{tikzpicture} }\hspace{2cm}
\scalebox{\myscale}{
\begin{tikzpicture}
\begin{axis}[clip mode=individual, 
  axis lines=middle,
x= \mysize, y = \mysize,
  xmin=0,xmax=5.5,ymin=0,ymax=5.5,
  xtick distance=1,
  ytick distance=1,
  xlabel=$x_{1}$,
  ylabel=$x_{2}$,
  grid=major,
  grid style={thin,densely dotted,black!20}]
\node (t2) at (axis cs:2,2) [circle, scale=0.3, draw=darkgreen,fill=darkgreen] {};

\node (t1) at (axis cs:1,2) [circle, scale=0.3, draw=darkgreen,fill=darkgreen] {};

\node (t3) at (axis cs:4,2) [circle, scale=0.3, draw=darkgreen,fill=darkgreen] {};
\node (t4) at (axis cs:4,4) [circle, scale=0.3, draw=darkgreen,fill=darkgreen] {};
\draw[darkgreen,thick] (axis cs:1,2) -- (axis cs:4,2);
\node[darkgreen] at (axis cs:5,1.5) {};

\node[darkgreen] at (axis cs:2.5,5) {\scalebox{\myoscale}{$\Pre{1}$}};
\end{axis}
\end{tikzpicture}}\hspace{2cm}
\scalebox{\myscale}{
\begin{tikzpicture}
\begin{axis}[clip mode=individual, 
  axis lines=middle,
x= \mysize, y = \mysize,
  xmin=0,xmax=5.5,ymin=0,ymax=5.5,
  xtick distance=1,
  ytick distance=1,
  xlabel=$x_{1}$,
  ylabel=$x_{2}$,
  grid=major,
  grid style={thin,densely dotted,black!20}]
\node (t2) at (axis cs:2,2) [circle, scale=0.3, draw=orange,fill=orange] {};

\node (t1) at (axis cs:1,2) [circle, scale=0.3, draw=orange,fill=orange] {};

\node (t3) at (axis cs:4,2) [circle, scale=0.3, draw=orange,fill=orange] {};
\node (t4) at (axis cs:4,4) [circle, scale=0.3, draw=orange,fill=orange] {};

\draw[orange,thick] (axis cs:4,2) -- (axis cs:4,4);
\node[orange] at (axis cs:4,4.5) {};
\node[orange] at (axis cs:2.5,5) {\scalebox{\myoscale}{$\Pre{2}$}};

\end{axis}
\end{tikzpicture}}
\end{center}
    \caption{Example for a 2-player $(1,1)$-restrictive-closed, jointly constrained GNEP $I$ w.r.t.~$X\subseteq \R^{(1,1)}$ represented by the four black dots in the first picture. The prescribed strategy sets $\Pre{1},\Pre{2}$ which rule out the possibility for $I^\conv$ to be jointly convex are represented in picture 2 and 3 respectively.}
    \label{fig: ExaJointConstr}
\end{figure}

In the previous section, we identified the restrictive-closedness for $I$ w.r.t.~the $k$-convex hull operator as the characterizing property for $\mathcal{I}^\conv$ to contain a jointly constrained instance. It turns out that  the restrictive-closedness of $I$ w.r.t.~the \emph{standard} convex hull operator is the characterizing property of $I$  for $\mathcal{I}^\conv$  to contain a jointly convex instance.

  \begin{definition}\label{def: projective-closed}
 We call an instance $I$ \emph{restrictive-closed} (w.r.t.~the convex hull operator),  if
 for all $i\in N$ and  $x_{-i} \in \dom X_i$ the equality
 \begin{align*}
     \conv\left(\res{\Str{i}}{x_{-i}} \right) = \res{\conv(\Stra)}{x_{-i}}
 \end{align*} 
 holds. Note that $\subseteq$ always holds, cf. Proposition~\ref{prop: incl}.
 \end{definition}
 The above concept of restrictive-closed sets requires that  for fixed $x_{-i}\in\dom X_i$, the convex hull of the restriction of $\Str{i}$ w.r.t.~$x_{-i}$ is equal to the restriction of $\conv(\Stra)$ w.r.t.~$x_{-i}$. 
 
 By observing that we have for any $i \in N$ and $x_{-i} \in \dom X_i$:
 \begin{align*}
     \conv\left(\res{\Str{i}}{x_{-i}} \right) = \convk\left(\res{\Str{i}}{x_{-i}} \right) \subseteq \res{\convk(\Stra)}{x_{-i}} \subseteq \res{\conv(\Stra)}{x_{-i}}
 \end{align*}
 due to \eqref{eq: convkeqconv}, Proposition~\ref{prop: incl} and $\convk(\Stra) \subseteq \conv(\Stra)$ respectively,  we can immediately state the following necessary conditions for restrictive-closedness of $I$.
 \begin{lemma}
 If $I$ is restrictive-closed, then 
 \begin{enumerate}
     \item $I$ is $k$-restrictive-closed.
     \item  for all $i \in N$ and $x_{-i}\in \dom X_i$ the equality 
      $\res{\convk(\Stra)}{x_{-i}} = \res{\conv(\Stra)}{x_{-i}}$ holds.
 \end{enumerate}
 \end{lemma}
 Note that the second necessary condition is not sufficient, not even for $k$-restrictive-closedness of $I$ as the example in Figure~\ref{fig:kovskoe} illustrates.

Similar to the previous section, we can give two equivalent characterizations of restrictive-closedness. 
The equivalence $\ref{enu: resclo}.\Leftrightarrow\ref{enu: exrc}.$ is again a geometric interpretation of restriction-closed GNEPs. With the help of it, one can easily verify that the example in Figure~\ref{fig: ExaJointConstr} is not restrictive-closed, as $x_{-2} = x_1: =2 \in \dom X_2$ but the restriction $\res{\conv(X)}{x_{-2}} = \{2\}\times [2,8/3]$ (marked in red in the first picture)  has the extreme point $(2,8/3)$ (marked as red cross) which is not contained in $X$.

The equivalence $\ref{enu: resclo}.\Leftrightarrow\ref{enu: rcprop}.$~will
allow us to show in the subsequent Theorem~\ref{thm: JointConst} that restrictive-closed GNEPs are exactly the GNEPs $I$ which admit a jointly convex convexification $I^\conv\in \mathcal{I}^\conv$. 
 \begin{theorem}\label{lem: geometricInt}
 Let $I$ be an instance of the GNEP. Then the following statements are equivalent: 
  \begin{enumerate}
     \item $I$ is restrictive-closed.\label{enu: resclo}
     \item  $E \left(\res{\conv(\Stra)}{x_{-i}} \right) \subseteq \Str{i}$ for all $i \in N$ and $x_{-i} \in \dom X_i$.\label{enu: exrc}
     \item For all $i \in N$ and  $x_{-i} \in \R^{k_{-i}}$ the following implication holds: \label{enu: rcprop}
     \begin{align*}
         x_{-i} \in \dom X_i \Rightarrow \conv(X_i(x_{-i})) = \{x_i \in \R^{k_i} \mid (x_i,x_{-i}) \in \conv(\Stra)\}. 
     \end{align*}
 \end{enumerate}
 \end{theorem}
 \begin{proof}
  $\ref{enu: resclo}\Leftrightarrow\ref{enu: exrc}$: Let $i \in N$ and $x_{i} \in \dom X_i$ be arbitrary.
  
 $\ref{enu: resclo}\Rightarrow\ref{enu: exrc}$: We get by the restrictive-closedness of $I$:
 \begin{align*}
     E \left(\res{\conv(\Stra)}{x_{-i}} \right) = E\big(\conv\big(\res{\Str{i}}{x_{-i}}\big)\big) \subseteq \res{\Str{i}}{x_{-i}} \subseteq \Str{i}
 \end{align*}
 
$\ref{enu: resclo}\Leftarrow\ref{enu: exrc}$: 
The following implications hold:
 \begin{alignat*}{2}
   \ref{enu: exrc}. \quad \Rightarrow \quad &E\left( \res{\conv(\Stra)}{x_{-i}} \right) &&\subseteq \quad\res{\Str{i}}{x_{-i}} \\ 
   \Rightarrow\quad &\conv\Big(E\big( \res{\conv(\Stra)}{x_{-i}} \big)\Big) \; &&\subseteq \quad\conv\big(\res{\Str{i}}{x_{-i}}\big) \\ 
   \Rightarrow \quad&\res{\conv(\Stra)}{x_{-i}} &&\subseteq\quad \conv\big(\res{\Str{i}}{x_{-i}}\big)
 \end{alignat*}
 where the last inclusion follows by the convexity of $\res{\conv(\Stra)}{x_{-i}}$.
Since the inclusion $\supseteq$ in the last line always holds, the claim follows.

$\ref{enu: resclo}.\Leftrightarrow\ref{enu: rcprop}.$: Let $i \in N$ and $x_{-i}\in \dom X_i$ be arbitrary.

$\ref{enu: resclo}.\Rightarrow\ref{enu: rcprop}.$: 
This follows immediately by the definition of restrictive-closedness:
\begin{align*}
     X_i^\conv(x_{-i}) \times x_{-i} &= \conv\big(X_i(x_{-i})\big) \times x_{-i} = \conv\big(\res{\Str{i}}{x_{-i}}\big)\\ &\overset{\ref{enu: resclo}.}{=} \res{\conv(\Stra)}{x_{-i}} = \{x_i \in \R^{k_i} \mid (x_i,x_{-i}) \in \conv(\Stra)\} \times x_{-i}.
\end{align*}

$\ref{enu: resclo}.\Leftarrow\ref{enu: rcprop}.$: We calculate:
\begin{align*}
    \conv\big(\res{\Str{i}}{x_{-i}}\big) &= \conv\big(X_i(x_{-i})\big) \times x_{-i}\\ &\overset{\ref{enu: rcprop}.}{=} \{x_i \in \R^{k_i} \mid (x_i,x_{-i}) \in \conv(\Stra)\} \times x_{-i} = \res{\conv(\Stra)}{x_{-i}}.
\end{align*}
 \end{proof}
 
\begin{theorem}\label{thm: JointConst}
Let $I$ be an instance of the GNEP with $\mathcal{I}^\conv(I) \neq \emptyset$. Then $I$ is restrictive-closed if and only if the convexification $\mathcal{I}^\conv$ contains a jointly convex instance. In this case, $\mathcal{I}^\conv$ contains for any convex set $X^\conv$ that satisfies
\begin{align}\label{eq: PCthm}
\res{X^\conv}{x_{-i}} = \res{\conv(\Stra)}{x_{-i}} && \text{ for all } i \in N \text{ and } x_{-i}\in \dom X_i,
\end{align}
 a jointly convex instance with $X^\conv$ as its restriction set.
\end{theorem}
\begin{proof}
We start with the only if direction.
Let $I^\conv$ be a jointly convex instance w.r.t.~any restriction set $X^\conv$ as described above. Furthermore, let the cost functions of $I^\conv$ fulfill the requirements for $I^\conv$ to belong to $\mathcal{I}^\conv$ {which exist due to $\mathcal{I}^\conv(I) \neq \emptyset$}. 
We want to prove that $I^\conv \in \mathcal{I}^\conv$.
Thus, we have to show that for arbitrary $i \in N$ and $ x_{-i} \in \dom X_i$, the strategy set $X_i^\conv(x_{-i})$ of  $I^\conv$ fulfills:
\begin{align*}
  X_i^\conv(x_{-i}):=\left\lbrace x_i \in \R^{k_i} \mid (x_i,x_{-i}) \in X^\conv \right \rbrace \overset{!}{=} \conv(X_i(x_{-i})) 
\end{align*}
which is equivalent to 
$\res{X^\conv}{x_{-i}} = \conv\big(\res{\Str{i}}{x_{i}}\big).$
The latter equality is valid due to $I$ being restrictive-closed and $X^\conv$ fulfilling~\eqref{eq: PCthm}. 

For the if direction, let $I^\conv \in \mathcal{I}^\conv$ be jointly convex w.r.t.~$X^\conv$. Then for any $i \in N$ and $x_{-i}\in \dom X_i$, we have 
\begin{align}\label{eq: thmjcv}
    \conv(X_i(x_{-i})) =: X_i^\conv(x_{-i}) =  \{x_i \in \R^{k_i} \mid (x_i,x_{-i}) \in X^\conv\}
\end{align}
by $I^\conv \in \mathcal{I}^\conv$ and $I^\conv$ being jointly convex w.r.t.~$X^\conv$. 
Furthermore the jointly constrainedness implies that $\Stra^\conv = X^\conv$. As $\Stra \subseteq \Stra^\conv$ we get $\conv(\Stra) \subseteq X^\conv$ due to $X^\conv$ being convex. Thus the equality in~\eqref{eq: thmjcv} implies
\[ \conv(X_i(x_{-i})) =  \{x_i \in \R^{k_i} \mid (x_i,x_{-i}) \in X^\conv\}\supseteq  \{x_i \in \R^{k_i} \mid (x_i,x_{-i}) \in \conv(\Stra)\}. \]
The above inclusion is in fact an equality as 
the proof of Theorem~\ref{lem: geometricInt} ($\ref{enu: resclo}.\Rightarrow\ref{enu: rcprop}.$) together with the fact that  $\conv\big(\res{\Str{i}}{x_{-i}}\big) \subseteq \res{\conv(\Stra)}{x_{-i}}$ always holds shows that the inclusion $\conv(X_i(x_{-i}))\subseteq \{x_i \in \R^{k_i} \mid (x_i,x_{-i}) \in \conv(\Stra)\}$ is also always fulfilled. Thus, Theorem~\ref{lem: geometricInt} ($\ref{enu: resclo}.\Leftarrow\ref{enu: rcprop}.$)  shows that $I$ is restrictive-closed. 
\end{proof}

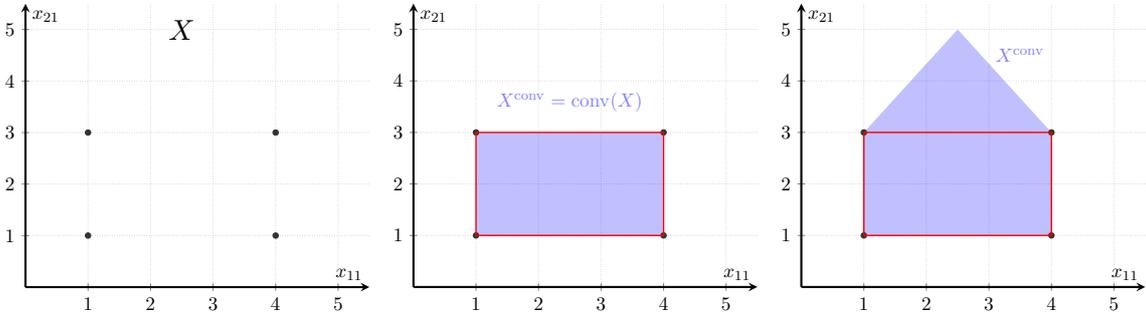
\begin{figure}[h!]
\begin{center}
\scalebox{\myscale}{
\begin{tikzpicture}
\begin{axis}[clip mode=individual, 
  axis lines=middle,
  x= \mysize, y = \mysize,
  xmin=0,xmax=5.5,ymin=0,ymax=5.5,
  xtick distance=1,
  ytick distance=1,
  xlabel=$x_{1}$,
  ylabel=$x_{2}$,
  grid=major,
  grid style={thin,densely dotted,black!20}]
\node at (axis cs:1,3) [circle, scale=0.3, draw=black!80,fill=black!80] {};

\node at (axis cs:1,1) [circle, scale=0.3, draw=black!80,fill=black!80] {};

\node at (axis cs:4,3) [circle, scale=0.3, draw=black!80,fill=black!80] {};
\node at (axis cs:4,1) [circle, scale=0.3, draw=black!80,fill=black!80] {};

\node at (axis cs:2.5,5) {\scalebox{\myoscale}{$X$}};
\end{axis}
\end{tikzpicture}}\hspace{2cm}
\scalebox{\myscale}{
\begin{tikzpicture}
\begin{axis}[clip mode=individual, 
  axis lines=middle,
x= \mysize, y = \mysize,
  xmin=0,xmax=5.5,ymin=0,ymax=5.5,
  xtick distance=1,
  ytick distance=1,
  xlabel=$x_{1}$,
  ylabel=$x_{2}$,
  grid=major,
  grid style={thin,densely dotted,black!20}]
\node at (axis cs:1,3) [circle, scale=0.3, draw=black!80,fill=black!80] {};

\node at (axis cs:1,1) [circle, scale=0.3, draw=black!80,fill=black!80] {};

\node at (axis cs:4,3) [circle, scale=0.3, draw=black!80,fill=black!80] {};
\node at (axis cs:4,1) [circle, scale=0.3, draw=black!80,fill=black!80] {};

\addplot[name path=g, domain=1:4, samples=2,red,thick] {3};
\addplot[name path=h, domain=1:4, samples=2,red,thick] {1};
\addplot [mark = none,thick,red] coordinates {(4, 1) (4, 3)};
\addplot [mark = none,thick,red] coordinates {(1, 1) (1, 3)};

    \addplot [
        thick,
        color=blue,
        fill=blue, 
        fill opacity=0.25
    ]
    fill between[
        of=g and h,
        soft clip={domain=1:4},
    ];

\node[blue!50] at (axis cs:2.5,3.6) {\scalebox{1}{$X^\conv = \conv(X)$}};
\end{axis}
\end{tikzpicture}}\hspace{2cm}
\scalebox{\myscale}{
\begin{tikzpicture}
\begin{axis}[clip mode=individual, 
  axis lines=middle,
  x= \mysize, y = \mysize,
  xmin=0,xmax=5.5,ymin=0,ymax=5.5,
  xtick distance=1,
  ytick distance=1,
  xlabel=$x_{1}$,
  ylabel=$x_{2}$,
  grid=major,
  grid style={thin,densely dotted,black!20}]
\node at (axis cs:1,3) [circle, scale=0.3, draw=black!80,fill=black!80] {};

\node at (axis cs:1,1) [circle, scale=0.3, draw=black!80,fill=black!80] {};

\node at (axis cs:4,3) [circle, scale=0.3, draw=black!80,fill=black!80] {};
\node at (axis cs:4,1) [circle, scale=0.3, draw=black!80,fill=black!80] {};

\addplot[name path=g, domain=0:2.5, samples=2,draw=none] {4/3 * x + 5/3};
\addplot[name path=f, domain=2.5:4, samples=2,draw = none] {-4/3 * x + 25/3};
\addplot[name path=h, domain=0:4, samples=2,draw=none] {1};
    \addplot [
        thick,
        color=blue,
        fill=blue, 
        fill opacity=0.25
    ]
    fill between[
        of=g and h,
        soft clip={domain=1:2.5},
    ];
    \addplot [
        thick,
        color=blue,
        fill=blue, 
        fill opacity=0.25
    ]
    fill between[
        of=f and h,
        soft clip={domain=2.5:4},
    ];
\node[blue!50] at (axis cs:4,4.5) {\scalebox{\myoscale}{$X^\conv$}};
\addplot[domain=1:4, samples=2,red,thick] {3};
\addplot[domain=1:4, samples=2,red,thick] {1};
\addplot [mark = none,thick,red] coordinates {(4, 1) (4, 3)};
\addplot [mark = none,thick,red] coordinates {(1, 1) (1, 3)};

\end{axis}
\end{tikzpicture}}
\end{center}
\caption{Example for a 2-player restrictive-closed, jointly constrained GNEP $I$ w.r.t.~$X\subseteq \R^{(1,1)}$ represented by the four black dots in the first picture. In picture 2 and 3 are two suitable choices of $X^\conv$. The sets that have to coincide after~\eqref{eq: PCthm} are represented by the 4 red lines.}
 \label{fig: Remark}
\end{figure}

\subsection{Applications}
In this subsection we show how several interesting game classes belong to the restrictive-closed GNEPs. In order to do so, we present in the following a sufficient condition for restrictive-closedness
which requires the definition of \emph{pseudo jointly constrained} GNEPs. 

\begin{definition}
 We call an instance $I$ \emph{pseudo jointly constrained}, if for all $i \in N$ and $x_{-i} \in \dom X_i$ the strategy set $X_i(x_{-i})$ can be described as:
 \begin{align*}
     X_i(x_{-i}) = \{x_i \in \R^{k_i} \mid (x_i,x_{-i}) \in \Stra\}.
 \end{align*}
\end{definition}
Remark that any jointly constrained instance w.r.t.~a restriction set $X$ is obviously pseudo jointly constrained as $\Str{i}= X = \Stra$ holds for all $i \in N$.
 Subsequently, Figure~\ref{fig: ExaSJC} shows that not every pseudo jointly constrained GNEP is $k$-restrictive-closed. Similarly, not every $k$-restrictive-closed GNEP is pseudo jointly constrained as the example in Figure~\ref{fig: ExaNotJC} illustrates. 
\begin{lemma}\label{lem: nobadsets}
Let $I$ be a  pseudo jointly constrained GNEP. Then $I$ is restrictive-closed, if
the projection $P_i(\Stra) := \left\lbrace x_i \in \R^{k_i}\mid \exists\; x_{-i}: (x_i,x_{-i}) \in \Stra \right \rbrace$ of $\Stra$  to the strategy space $\R^{k_i}$ of player $i \in N$
only consist of extreme points, i.e.  $E(P_i(\Stra)) = P_i(\Stra)$.
\end{lemma}
\begin{proof}
Let $i\in N$ and $x_{-i} \in \dom X_i$ be arbitrary. We have to show that the equality in Definition~\ref{def: projective-closed} holds. As mentioned before, $\subseteq$  always holds. 
For the other inclusion $\supseteq$ we argue that the following steps are valid:
$\res{\conv(\Stra)}{x_{-i}} \subseteq \conv\big(\res{\Stra}{x_{-i}}\big) \subseteq \conv\big(\res{\Str{i}}{x_{-i}}\big).$

For the first inclusion, let $(x_i,x_{-i}) \in  \res{\conv(\Stra)}{x_{-i}}$. Then there exists 
a convex combination $(x_i,x_{-i})=\sum_{s=1}^L\lambda_s x^s$ with $x^s \in \Stra, s\in [L]$ for some $L \in \N$ and $\lambda \in \Lambda_{L}$. Since $x_{-i} \in \dom X_i$ there exists a $x_i^* \in \R^{k_i}$ with $(x_i^*,x_{-i}) \in X\big((x_i^*,x_{-i})\big)$. Subsequently $(x_i^*,x_{-i}) \in \Stra$ and $x_j \in P_j(\Stra) = E(P_j(\Stra))$ for all $j \neq i$. Similarly, as $x^s \in \Stra$ we have $x_j^s \in P_j(\Stra)$ for all $j \in N, s\in[L]$. Therefore  $x_j = \sum_{s=1}^L\lambda_s x^s_j$ for $j \in N$ implies  $x_j^s = x_j$ for all $s\in [L], j\neq i$. Therefore $x^s=(x_i^s,x_{-i}) \in \res{\Stra}{x_{-i}}, s\in [L]$ which shows that $x \in \conv\big(\res{\Stra}{x_{-i}}\big)$.

The second equality is a direct consequence of $I$ being pseudo jointly constrained as pseudo jointly constrainedness is equivalent to  $
    \res{\Stra}{x_{-i}} = \res{\Str{i}}{x_{-i}} \text{ for all } i \in N, x_{-i} \in \dom X_i.$
\end{proof}
We get as a direct consequence of the above lemma that all $0,1$ pseudo jointly constrained games are restrictive-closed.
\begin{corollary}\label{cor: nobadsets}
If $I$ is a pseudo jointly constrained GNEP with $\Stra \subseteq \{0,1\}^k$, then $I$ is restrictive-closed.
\end{corollary}
\begin{proof}
The projection of $\Stra \subseteq \{0,1\}^k$ to the strategy space of any player $i \in N$ is a subset of the hypercube $\{0,1\}^{k_i}$ which only consists of extreme points. Thus $E(P_i(\Stra)) = P_i(\Stra)$ holds.
\end{proof}

Another interesting class of GNEPs which belongs to the restrictive-closed GNEPs are the jointly constrained discrete flow games described in the following lemma.   

\begin{lemma}\label{lem: JCDFG}
Let $I$ be an instance of the GNEP as described in Example~\ref{exa: CDFG} but instead of linear \emph{individual} capacity constraints, we consider a joint restriction imposed by a convex function $g:\R^k \to \R^s$ for some $s \in \N$. More precisely,  the strategy set of player $i \in N$ is 
described by 
\begin{align*}
X_i(x_{-i})=\Bigl\lbrace x_i\in \Z_{\geq 0}^m \mid Ax_i = b_i,\; g(x_i,x_{-i})\leq 0 \Bigr \rbrace && \text{ for all } x_{-i} \in \R^{k_{-i}}.
\end{align*}
Furthermore let
\begin{align}\label{eq: setX} 
    X:= \prod_{i\in N} \Bigl\lbrace x_i\in \Z_{\geq 0}^m \mid Ax_i = b_i\Bigr \rbrace\; \cap \; \Bigl\lbrace x\in \R_{\geq 0}^k \mid g(x) \leq 0 \Bigr \rbrace.
\end{align} 
If for all $i\in N$ and 
\begin{align}\label{eq: dom}
    x_{-i}\in  \{\tilde{x}_{-i} \in \R^{k_{-i}} \mid \exists\,\tilde{x}_i \in \R^{k_i}: (\tilde{x}_i,\tilde{x}_{-i}) \in X \},
\end{align}
 the restriction $g(x_i,x_{-i}) \leq 0$ for $x_i$ is equivalent to an integral box-constraint $a_i^{x_{-i}} \leq x_i \leq b_i^{x_{-i}}$ with $a_i^{x_{-i}},b_i^{x_{-i}} \in \Z^{k_{i}}$, then $I$ is restrictive-closed.  
\end{lemma}
\begin{proof}
It is not hard to see that $I$ is quasi-isomorphic to a  GNEP $I'=(N,(X_i'(\cdot))_{i \in N},(\pi_i)_{i \in N})$ which has the same cost functions as $I$ and is jointly constrained w.r.t.~$X$. As for two quasi-isomorphic instances the respective complete (relevant) strategy sets for a player $i \in N$ coincide $\Str{i} = \mathcal{S}_i'$, we get $\Stra = X$. Furthermore we have that  $\dom X_i = \dom X_i'$ where the latter is given by the set in~\eqref{eq: dom}.
To verify the restrictive-closedness, we show that $\Stra =\Stra'= X$ fulfills the condition stated in Definition~\ref{def: projective-closed}. Let $i \in N$ and $x_{-i} \in \dom X_i$. Since $\subseteq$ always holds we just have to show the inclusion $\supseteq$. To prove this, define the relaxation of $X$ by
\begin{align}\label{eq: defRelX}
    \hat{X}:= \prod_{i\in N} \Bigl\lbrace x_i\in \R_{\geq 0}^m \mid Ax_i = b_i\Bigr \rbrace\; \cap \; \Bigl\lbrace x\in \R_{\geq 0}^k \mid g(x) \leq 0 \Bigr \rbrace.
\end{align} 
We argue that the following two inclusions are valid. Clearly, they imply that $\supseteq$ in Definition~\ref{def: projective-closed} holds.
\begin{align}
          \left \lbrace x_i \in \R^{k_i} \mid (x_i,x_{-i}) \in \conv(X) \right \rbrace &\subseteq \left \lbrace x_i \in \R^{k_i} \mid (x_i,x_{-i}) \in \hat{X} \right \rbrace \label{eq: ExampleCDFG1}\\
          &\subseteq\conv\left(\left \lbrace x_i \in \R^{k_i} \mid (x_i,x_{-i}) \in X \right \rbrace \right). \label{eq: ExampleCDFG2}
\end{align}
By the definition of $X$ it follows immediately that $X \subseteq \hat{X}$. Since $\hat{X}$ is convex, the inclusion $\conv(X) \subseteq \hat{X}$ and thus also the inclusion \eqref{eq: ExampleCDFG1} holds. By rewriting the sets for the inclusion~\eqref{eq: ExampleCDFG2} via the definition of $X$ and $\hat{X}$, we get equivalently:
\begin{align*}
    \Bigl\lbrace x_i\in \R_{\geq 0}^m \mid Ax_i = b_i,\; g(x_i,x_{-i}) \leq 0 \Bigr \rbrace  \subseteq \conv\Bigl( \Bigl\lbrace x_i\in \Z_{\geq 0}^m \mid Ax_i = b_i,\; g(x_i,x_{-i}) \leq 0 \Bigr \rbrace \Bigr). 
\end{align*}
As $x_{-i} \in \dom X_i= \dom X_i'$ and thus fulfills~\eqref{eq: dom}, the restriction $g(x_i,x_{-i}) \leq 0$ is an integral box-constraint. Thus the polytope on the left has integral vertices since the flow polyhedron is box-tdi, cf.~Example~\ref{exa: CDFG}. These integral vertices are clearly contained in the right set and therefore the inclusion follows. Hence $I$ is restrictive-closed. 
\end{proof}

The above proof shows that the relaxed set $\hat{X}$ fulfills the equality stated in~\eqref{eq: PCthm} and thus a convexified instance from $\mathcal{I}^\conv$ which is jointly convex w.r.t.~$\hat{X}$ can be solved in order to derive insights into the original instance $I$ via our main Theorem~\ref{thm:main}. 
This is extremely convenient in a computational regard as $\conv(X) \neq \hat{X}$ in general as the following instance of the CDFG shows. Note that  the CDFG described in Example~\ref{exa: CDFG} belongs to the flow games introduced in Lemma~\ref{lem: JCDFG}.

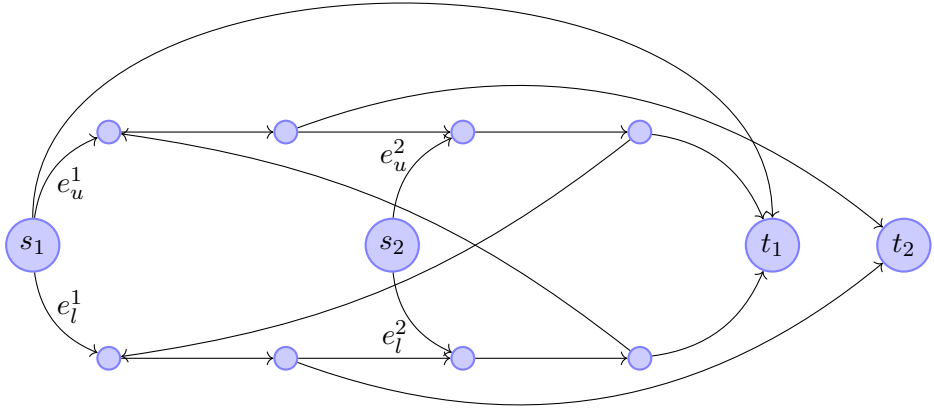
\begin{figure}[h!]
\begin{center}
\begin{tikzpicture} [ node distance = 2cm,
Small/.style={circle,draw=blue!50,fill=blue!20,thick, inner sep=0mm, minimum size=3mm},
Knoten/.style={circle,draw=blue!50,fill=blue!20,thick, inner sep=0mm, minimum size=5mm}]
\node[Knoten] (s1) at (0,2) {$s_1$};
\node[Knoten] (s2) [right = 2cm of s1] {$s_2$};
\node[Small] (1) at (1,1.25) {};
\node[Small] (2) [right = of 1] {};
\node[Small] (3) [right = of 2] {};
\node[Small] (4) [right = of 3] {};
\node[Small] (5) at (1,2.75) {};
\node[Small] (6) [right = of 5] {};
\node[Small] (7) [right = of 6] {};
\node[Small] (8) [right = of 7] {};
\node[Knoten] (t1) [ right = 9cm of s1] {$t_1$};
\node[Knoten] (t2) [right = 1cm of t1] {$t_2$};

\draw[->, bend right] (s1) to node[above right, xshift = -0.1cm, yshift = -0.15cm] {$e^1_l$} (1);
\draw[->, bend left ] (s1) to node[below right, xshift = -0.1cm, yshift = 0.15cm] {$e^1_u$} (5);
\draw[->] (5) to (6);
\draw[->] (6) to (7);
\draw[->] (7) to (8);
\draw[->] (1) to (2);
\draw[->] (2) to (3);
\draw[->] (3) to (4);
\draw[->, bend right = 35] (4) to (5);
\draw[->, bend left =35] (8) to (1);
\draw[->, out = 20, in = 220] (s2) to node[pos = 0.15, above left,yshift = -0.1cm] {$e^2_u$}  (7);
\draw[->,  out = 340, in = 140] (s2) to node[pos= 0.15,below left, yshift = 0.1cm] {$e^2_l$}  (3);
\draw[->, bend left=25] (6) to (t2);
\draw[->, bend right=25] (2) to (t2);
\draw[->,out = 0, in = 220] (4) to (t1);
\draw[->, out = 0, in = 140] (8) to (t1);
\draw[->] (s1) .. controls (0.5,4) and (8.5,4)  .. (t1);
\end{tikzpicture}
\caption{Example for a capacitated discrete flow game where $\conv(X) \subsetneq \hat{X}$ }
\label{fig: CDFG}
\end{center}
\end{figure} 
\begin{example}[$\conv(X) \neq \hat{X}$ in general] 
Let $I$ be an instance of the CDFG where $N=\{1,2\}$ and $G$ is given by the graph displayed in Figure~\ref{fig: CDFG}. {Both players have the same capacity $c_1 = c_2 = \mathbf{1}\in \R^E$ and want to send one unit of flow.} 
Then $X$ consists of only two elements, namely $X = \{(x_1^*,x_2^u),(x_1^*,x_2^l)\}$ where we denote by $x_1^*$ the flow sending one flow unit over the edge $(s_1,t_1)$ and by $x_2^u$ resp.~$x_2^l$ the unique path from $s_2$ to $t_2$ starting with the upper edge $e^2_u$ resp.~lower edge $e^2_l$.
The set $\hat{X}$ contains for example the point $\frac{1}{2}\cdot(x_1^u + x_1^l, x_2^u + x^l_2) \notin\conv(X)$ where we define $x_1^u$ and $x_1^l$ analogously to $x_2^u$ and $x_2^l$, thus showing that $\conv(X) \subsetneq \hat{X}$. 
\end{example}

As another example for restrictive-closed GNEPs we revisit Example~\ref{exa: Atom}.
\addtocounter{theorem}{-18}
\begin{example}[continued]
Assume that the weights $d_{ij} = 1$ are equal to one for all $i\in N$, $j \in E$. Then $X\subseteq \{0,1\}^{n\cdot m}$ and thus by Corollary~\ref{cor: nobadsets}, we're dealing with a restrictive-closed GNEP. Subsequently Theorem~\ref{thm: JointConst} shows that we can define $I^\conv \in \mathcal{I}^\conv$ as a jointly convex GNEP w.r.t.~any set $X^\conv$ fulfilling~\eqref{eq: PCthm}. Concerning the cost functions of $I^\conv$, we observe that
 the cost functions $\pi_i, i\in N$ are quasi-linear, i.e.~they allow for convexified cost functions as in Definition~\ref{GenSet}.\ref{enum:QL2}. This is due to the fact that $X_i(x_{-i}) \to \R, x_i \mapsto \pi_i(x_i,x_{-i})$ is linear for all $x_{-i} \in\dom X_i$ and $i \in N$ which can be verified by the following description: 
\begin{align*}\textstyle
 \pi_i(x_i,x_{-i}):= \sum_{j\in E} c_{ij}(\ell_j( x))x_{ij} = c_i(\mathbf{1} + \sum_{j\neq i}x_j  )^\top x_i \text{ for all } x_i \in X_i(x_{-i}) \subseteq \{0,1\}^m
\end{align*}  
Thus defining $\phi_i(x) := c_i(\mathbf{1} + \sum_{j\neq i}x_j)^\top x_i$ for all $x \in \R^k$ fulfills the restrictions for the cost functions of a convexified instance, i.e.~Definition~\ref{def: conv}.\ref{enu: c2}.
With this definition of $I^\conv$ and assuming that $c_i:\R^m \to \R^m$ is a smooth function, $I^\conv$ 
is a jointly convex GNEP w.r.t.~$X^\conv$ with smooth cost-functions for the players and thus various methods to solve $I^\conv$ are known in the literature.        
\end{example}

\section{Computational Study}\label{sec: Comp}
 {In this section, we present numerical results on the computation of generalized Nash equilibria for Examples \ref{exa: CDFG} and~\ref{exa: CompEq}, i.e., ~the capacitated discrete flow games and transportation markets. 

We will consider three different types of methods for the computation of equilibria. 
 The first type (see Section~\ref{subsec:ql}) exploits the fact that the instances are hole-free GNEPs allowing us to apply the reformulation of the GNEP via Corollary~\ref{cor:ref}.   
Note that this
approach is correct (assuming enough run-time) and has the striking advantage
that any positive lower bound of the resulting
global optimization problem serves
as a certificate for the non-existence of GNE.

 The second type (see Section~\ref{subsec:mv}) uses Theorem~\ref{thm:main} in the sense that we try to compute a GNE for the convexified GNEP (e.g. by finding local  minimizer of the $\hat V$ function)
 and then check feasibility for the original non-convex GNEP. 
This approach has the advantage that
it can use well-known numerical methods from the area of convex GNEPs and in addition
it is in principle  applicable to all GNEPs and not only quasi-linear GNEPs. On the down-side, this approach is only correct, if we were able
to compute \emph{all} GNE for the convexified instance and check them for original feasibility. 

The third type (see Section~\ref{subsec:br}) is a best response algorithm (which we term BR), where the players -- whenever they can strictly improve their costs--  update their strategy using a best response.
Remark, however, that a BR-algorithm is not correct in general as it may not terminate due to cycling or may stop at infeasible strategy profiles for which a player's 
optimization problem is infeasible. As a consequence, 
the BR heuristic is not applicable
to the instances of Example~\ref{exa: CompEq} (cf.~Section~\ref{sec: results}). 

Let us emphasize  that prior to our paper,
there were no existing methods available in the literature that can deal with general non-convex or even quasi-linear GNEPs. 
Hence, a comparison of our proposed methods to some benchmark methods from the literature is not possible.

In the following subsections, we first describe the set of test instances which we generated. 
Then, we examine the aforementioned methods in  more detail and conclude 
by presenting their numerical results in terms of the number of GNE found, certificates of non-existence and computation times on average. 
}

\subsection{Test Instances}

We generated 10 different graphs $G=(V,E)$ for  each  $|V| \in \{10,15,20\}$ and each of the three different player set sizes $N \in \{2,4,10\}$.  
The edges of the graph were assigned randomly with each pair of nodes $a\neq b \in V$ having a $\{20\%,15\%,10\%\}$ chance for $|V| \in \{10,15,20\}$ to be connected by the directed arc $(a,b)$. Concerning the source sink pair of each player, we generated two types. Namely on the one hand a single source single sink type in which every player gets the same randomly selected (connected) source sink pair. On the other hand a multi source multi sink type in which  each player has an individual randomly selected (connected) source sink pair.
Similar, the weight of each player, i.e.~the integral amount of flow each player wants to send, is either chosen uniformly at random from the range of 1 to 10 or set to 1 for each player.  
In conclusion, we generated 10 graphs for each  combination of $|V|$, $|N|$, the two types concerning the source/sink assignment and  the weight assignment, leading to a total of 360 different graph-player setups. 

{
\paragraph{JCDFG.}
For the above described graph-player setups, we considered 
two different types of the CDFG. 
The first one is a jointly capacitated version (JCDFG) in which every player has the same capacity vector $c_i = c, i\in N$. 
To generate capacities that have an impact on the strategy sets, we first chose the capacities uniformly at random from a relative small range of 1 to $\max(n,d_1,\ldots,d_n)$. If the resulting strategy space  is empty, the capacities are reassigned. This random reassignment is executed until either the strategy space is not empty anymore or a limit for the amount of reassignments is exceeded. In the latter case, the range of values in which the capacities are chosen is incremented by one and the procedure is repeated. 
Regarding the cost functions, we use $\pi_i(x_i,x_{-i}):=  (\sum_{j\neq i}x_j)^\top \, C_i^1\, x_i + C_i^{2\top}\, x_i$.
If the player's weights are arbitrary, $C_i^1 \in \Z_{\geq 0}^{m\times m}$ and $C_i^2 \in \Z_{\geq 0}^{m}$ are randomly generated  with values in the range of 0 to 20. Otherwise, we set $C_i^1=\mathrm{diag}(C_i^2)$ as the diagonal matrix having $C_i^2$ 
as its diagonal as the  JCDFG can then be interpreted as a jointly constrained atomic congestion game in this case, cf.~Example \ref{exa: Atom}. 

\paragraph{ICDFG.}
We also considered an individually capacitated version of the CDFG (ICDFG)
by changing the above described jointly constrained instances only w.r.t.~the capacities of the players. 
In contrast to above, players now have individual capacities, i.e.~$c_i\neq c_j$ in general, which  
are analogously generated to the jointly constrained case. 

\paragraph{Transportation Markets.}
For the transportation markets of Example~\ref{exa: CompEq}, we set the weights in all of the above graph-player
setups to 1 and then
 considered only those graph-player setups 
in which at least one edge-disjoint path-allocation exists. 
This resulted in 74 instances. 
The players' costs were then set to $\pi_i(x_i,x_{-i},p):=   \big(p-C_i^{2}\big)^\top \, x_i$ (with $C_i^2$ from the JCDFG). 
Hereby, $p$ is the strategy of the market manager
which we added on top of the previously existing players. 
The latter's optimization problem is defined as described in Example~\ref{exa: CompEq}
with the addition of an upper price bound $p_e \leq \mathrm{PB}, e \in E$
for which we considered three different cases $\mathrm{PB}\in \{20,35,50\}$. 
}

{
\subsection{Computing Generalized Equilibria}
In order to compute original GNE, we 
use the continuous relaxation of the original games as convexification $I^\conv$ which is possible since GNEPs corresponding to the JCDFG, ICDFG and transportation markets are hole-free represented GNEPs. 
Based on this convexified instance, we implemented the following methods in \textsc{MATLAB}\textsuperscript{\tiny\textregistered}
in order to find equilibria of $I$.

\subsubsection{Quasi-linear Reformulation}\label{subsec:ql}
The first method is relying on the fact that both types of the CDFG and the transportation markets are hole-free-represented {player-linear mixed}-integer GNEPs for which Corollary~\ref{cor:ref} is applicable and consequently  the problem of finding a GNE reduces to finding a global optimum of a MINLP.
A striking advantage of this reformulation is its computational tractability as it allows for the application of 
global (MINLP) solvers such as  \textsc{BARON}, cf.~\cite{Baron18}. 
Note that these solvers typically require that the objective and restriction functions have an \emph{algebraic} description, i.e.~only consist of solver-supported operations like $+,-,\cdot$ etc., which is not the case for the non-reformulated,
original problem~\eqref{Opt}. 
Furthermore, the possibility to use \textsc{BARON} comes with
the additional benefit that \textsc{BARON} generates lower bounds on the optimal objective value 
during the search for a global optimum. 
A lower bound larger than zero  is a  certificate
for non-existence of equilibria and
the computation can be exited as soon as such a bound is found, cf.~the computational results for the transportation markets in Section~\ref{sec: results}. 

}

{

\subsubsection{Minimizing Variants of the $\hat{V}$ Function}\label{subsec:mv}
The second type of methods is based on  minimizing different variants of the $\hat{V}$ function of $I^\conv$,
rounding the found minimum and then verifying whether the rounded 
solution is a GNE. 
In contrast to the first method, here we do not deal 
with an optimization problem with a complete \emph{algebraic} description, 
but each objective function ($\hat{V}$) call requires to solve $|N|$ linear minimization problems. 
Solving  optimization problems with an objective function of such inaccessible form 
is a challenging task and we're not aware of any solvers that 
may produce lower bounds on the optimal objective value and hence guarantee global minimality.
Instead, we utilize the \textsc{MATLAB}\textsuperscript{\tiny\textregistered} Optimization Toolbox and
implemented the $\hat{V}$ function using the LP solver \texttt{linprog} and calculated local minima
via the \texttt{fmincon} solver.

For the case of the jointly convex GNEPs and hence the convexification of the JCDFG, 
we are able to make use of the following regularization of the $\hat{V}$ function:
\begin{align*}
    \hat{V}_\alpha(x) := \max_{y \in \hat{X}} \sum_{i\in N} \left[\pi_i(x) - \pi_i(y_i,x_{-i}) - \frac{\alpha}{2}\,||\,x_i-y_i\,||^2\right]
\end{align*}
where $\hat{X}$ is the joint constraint set of $I^\conv$ (cf.~\eqref{eq: defRelX}), $||\cdot||$ is the Euclidean norm and $\alpha>0$  denotes a regularization parameter. 
Heusinger and Kanzow~\cite{Heusinger2009-2}  
showed that $\hat{V}_\alpha$
is bounded from below by zero, 
every feasible solution with value zero corresponds to a (normalized) GNE of $I^\conv$
and is continuously differentiable.
The latter fact is beneficial in a computational regard as it allows one to provide an analytic gradient, 
significantly speeding up the computation of a local minimum. 
In this regard, we also computed local minima
via the \texttt{fmincon} solver of $\hat{V}_\alpha$ with $\alpha = 0.02$ for $I^\conv$ of the JCDFG. 

Although the transportation markets do not correspond to a jointly convex GNEP, 
we are still able to define a similar regularization by
\begin{align*}
    \hat{V}_\alpha(x,p) :=   \pi_{n+1}(x,p) + \max_{y \in \hat{X}'} \sum_{i\in [n]} \left[\pi_i(x,p) - \pi_i(y_i,x_{-i},p) - \frac{\alpha}{2}\,||\,x_i-y_i\,||^2\right] 
\end{align*}
where the $n+1$-th player resembles the market manager and $\hat{X}'$ the continuous relaxation of the product of the flow polytopes $\times_{i \in [n]}X_i'$. 
By a similar argumentation as in~\cite{Heusinger2009-2}, 
it is easy to see that this function
is bounded from below by zero, 
every feasible solution with value zero corresponds to a  GNE of $I^\conv$
and is continuously differentiable.
Hence, we also tried to computed (local) minima
via the \texttt{fmincon} solver of $\hat{V}_\alpha$ with $\alpha = 0.02$. 

An interesting question is whether and how 
one can adjust these techniques to find GNE of the convexified game but at the same time 
preserve original feasibility. 
We performed a first step into this direction by also implementing a penalized version of the above two methods 
in which we augmented the $\hat{V}$/$\hat{V}_\alpha$ function by the additive term $ \frac{1}{m \cdot n}\sum_{i\in N}\sum_{j\in E} \sin (\pi \cdot x_{ij})^2$
which penalizes non-integrality, resulting in the functions $\hat{V}^{\mathrm{pen}}$/$\hat{V}_\alpha^{\mathrm{pen}}$.
The idea here is that local minima found by the solver should be more likely to be integral and hence originally feasible.
Yet, this penalty term must be viewed with caution as the computation of a single local minimum is likely to be more time consuming and 
new local minima with an objective value bigger than zero may be generated through this penalty term.

The \texttt{fmincon}  solver requests a starting point. Thus, we computed an ordered and common set of 2000 random starting points by projecting random vectors in $[0,\max(n,d_1,\ldots,d_n)]^k$ to the set of feasible strategy profiles of $I^\conv$. {This is done by solving for each random vector $r$ the quadratic program $\min_{x \in \hat{X}}{||r-x||^2}$ via the \texttt{quadprog} solver of \textsc{MATLAB}. Note that the corresponding computation time was negligible.} 
Beginning with the first starting point, a (local) minimum is then computed of the respective objective function. Each component of this local minimum is  then rounded to the nearest integer. The resulting integral vector is then checked for feasibility and whether or not it is a GNE of $I^\ext$ by evaluating the $\hat{V}$ function for $I^\ext$ at that point. If the rounded solution is not a GNE, the next (local) minimum is computed with the usage of the next starting point. This procedure is executed until either a GNE has been found, all  starting vectors were tried or a time limit of one minute is exceeded, in which case the current computation is exited and no further (local) minima are computed. 

\subsubsection{Best Response Algorithm}\label{subsec:br}
 The BR-algorithm is applied to the previously mentioned starting vectors until an original GNE was found or the time limit of one minute is exceeded.
Hereby, starting with the first player, in each iteration the current strategy profile is 
updated by the current player's best response which is computed by solving her corresponding integral linear program via \textsc{BARON}.    
If the current player's optimization problem is infeasible, 
the BR-algorithm stops and is applied to the next starting vector.  

}

{
\subsection{Results}\label{sec: results}
All methods have been implemented in \textsc{MATLAB}\textsuperscript{\tiny\textregistered} R2023a on Windows 10 Enterprise. The computations have been performed on a machine with Intel Core i5-12500  and 32 GB of memory.
An overview of the results can be seen in Table~\ref{table: overview}. 
The ``GNE'' column of a method displays how often an equilibrium was found while the ``Time'' column shows how long it took (in seconds) to compute the equilibrium on average.
The ``Non-Existence'' column shows how often \textsc{BARON} was able to give a lower bound larger zero on the objective, 
i.e.~giving a certificate for non existence of equilibria and the corresponding ``Time'' column shows how long it took (in seconds) to compute the lower bound.
In order to illustrate the 
behaviour of the methods with respect to different  player sets, 
we also present the results of the JCDFG and ICDFG subdivided into the three  possibilities $N \in \{2,4,10\}$. 
{To demonstrate the behaviour of the various methods for one instance-type, we also present in Figure~\ref{fig: WhiskerJCDFG} and Figure~\ref{fig: WhiskerICDFG} boxplots of the performance of all methods
 based on 100 randomly generated instances of the type (2,20,m,10) for the JCDFG and ICDFG.
 The diagrams show the distribution of the computation time (in seconds) of an integral GNE. The mark inside each box denotes the median,
boxes represent lower and upper quartiles, and the whisker ends show the minimum and maximum, respectively, apart from possible outliers
marked by a cycle. }

The results regarding the transportation markets demonstrate that 
a BR heuristic may fail completely for certain types of GNEPs as infeasible strategy profiles may lead to the non-existence of best responses for players. 
In contrast, the results of the quasi-linear approach for these transportation markets illustrate the advantage that comes with our quasi-linear reformulation. 
Namely the possibility to generate certificates for non-existence of equilibria. In this regard, \textsc{BARON} was able to find in all market instances either a GNE or such a certificate. 

Finally, remark that the overall weak performance of the $\hat{V}$ and   $\hat{V}^{\mathrm{pen}}$ approaches 
 can be mainly attributed to the numerical complexity of \texttt{fmincon} when using differentiation.  This task is  time-demanding, 
requiring numerous costly $\hat{V}$/$\hat{V}^{\mathrm{pen}}$ evaluations.
}

\begin{table}[h!]
    \centering
\resizebox{\columnwidth}{!}{%
\begin{tabular}{l*{14}{c}}\toprule
\multirow{2}{*}{Example} & \multicolumn{4}{c}{Quasi-Linear}&\multicolumn{2}{c}{$\hat V$} & \multicolumn{2}{c}{$\hat{V}^\pen$} & \multicolumn{2}{c}{$\hat V_\alpha$} & \multicolumn{2}{c}{$\hat{V}_\alpha^\pen$} &   \multicolumn{2}{c}{BR}  
\\\cmidrule(lr){2-5}  \cmidrule(lr){6-7} \cmidrule(lr){8-9}\cmidrule(lr){10-11}\cmidrule(lr){12-13} \cmidrule(lr){14-15} 
          & GNE & Time & Non-Existence & Time& GNE & Time & GNE & Time  & GNE & Time & GNE & Time  & GNE & Time   \\\midrule

        JCDFG, $|N|= 2$ & 120 & 0.25 & 0 & - & 113 & 18.59 & 114 & 17.26 & 120 & 0.65 & 120 & 0.54 & 113 & 0.64\\
        JCDFG, $|N|= 4$ & 118 & 0.97 & 0 & - & 58 & 23.00 & 62 & 23.16 & 119 & 4.31 & 119 & 3.63 & 103 & 2.40 \\ 
       JCDFG, $|N|=10$ & 103 & 11.65 & 0 &- & 15 & 29.21 & 15 & 27.23 & 80 & 12.84 & 80 & 11.23 & 96 & 3.57 \\
       \midrule
       JCDFG & 341 & 3.94 & 0 & - & 186 & 20.82 & 191 & 19.96 & 319 & 5.07 & 319 & 4.38 & 312 & 2.12 \\
       \midrule 
        ICDFG, $|N|= 2$ & 97 & 0.22 & 6 & 0.33 & 93 & 16.70 & 93 & 16.21 & \multicolumn{2}{c}{not applicable}  &  \multicolumn{2}{c}{not applicable}  & 97 & 1.88 \\
        ICDFG, $|N|= 4$  & 68 & 0.50 & 3 & 0.58 & 39 & 22.94 & 37 & 23.51 &  \multicolumn{2}{c}{not applicable}  &  \multicolumn{2}{c}{not applicable}  & 71 & 0.97 \\ 
         ICDFG, $|N|= 10$ & 39 & 10.87 & 5 & 2.43 & 6 &28.87 & 11 & 33.14 & \multicolumn{2}{c}{not applicable}  &  \multicolumn{2}{c}{not applicable} & 41 & 2.74 \\      \midrule
      ICDFG   & 204  & 2.35 & 14 &15.87 & 138 & 18.99 & 141 & 19.45  & \multicolumn{2}{c}{not applicable}  &  \multicolumn{2}{c}{not applicable}  & 209& 1.74  \\  \midrule
      Markets, $\mathrm{PB} = 20$   &6 & 0.10 & 68 & 0.12 & 4 & 18.53 & 4 & 22.36 & 4& 0.15&4&4.63 &   0&- \\   
      Markets, $\mathrm{PB} = 35$  & 39 & 0.12 & 35 & 0.13 & 16 & 29.73 & 17 & 30.57&27&4.35  & 27&4.30 &  0&-   \\
      Markets, $\mathrm{PB} = 50$ & 65 & 0.12 & 9 & 0.12 & 26 & 37.03 & 22 & 40.56& 40&6.30&43&8.04  &  0&-   \\
   \bottomrule
\end{tabular}}
\caption{The performances of the various methods applied to the different examples. 
The ``GNE'' column of a method displays how often an equilibrium was found while the ``Time'' column shows how long it took (in seconds) to compute the equilibrium on average. The ``Non-Existence'' column shows how often \textsc{BARON} was able to give a lower bound bigger zero on the objective, i.e.~giving a certificate for non existence of equilibria and the corresponding ``Time'' column shows how long it took (in seconds) to compute the lower bound.}
 \label{table: overview}
\end{table}

\begin{figure}[h!]
    \centering
    \scalebox{1}{
\subfigure[QL]{
\begin{tikzpicture}
\begin{axis}[boxplot/draw direction=y, x = 01cm,  xmajorticks=false,  height = \heightbp]
\addplot+ [boxplot, mark size = \marksize]
table [row sep =\\, y index=0] {
0.37   \\ 
0.76   \\ 
0.51   \\ 
0.39   \\ 
0.31   \\ 
0.47   \\ 
0.40   \\ 
0.35   \\ 
0.47   \\ 
0.38   \\ 
0.25   \\ 
0.40   \\ 
0.37   \\ 
0.50   \\ 
0.36   \\ 
0.24   \\ 
0.26   \\ 
0.72   \\ 
0.33   \\ 
0.33   \\ 
0.31   \\ 
0.68   \\ 
0.46   \\ 
0.47   \\ 
0.26   \\ 
0.28   \\ 
0.58   \\ 
0.49   \\ 
0.51   \\ 
0.35   \\ 
0.24   \\ 
0.41   \\ 
0.42   \\ 
0.23   \\ 
0.31   \\ 
0.45   \\ 
0.64   \\ 
0.42   \\ 
0.35   \\ 
0.49   \\ 
0.50   \\ 
0.43   \\ 
0.59   \\ 
0.40   \\ 
0.45   \\ 
0.29   \\ 
0.46   \\ 
0.72   \\ 
0.45   \\ 
0.51   \\ 
0.37   \\ 
0.33   \\ 
0.46   \\ 
0.63   \\ 
0.45   \\ 
0.51   \\ 
0.52   \\ 
0.56   \\ 
0.26   \\ 
0.97   \\ 
0.21   \\ 
0.20   \\ 
0.49   \\ 
0.35   \\ 
0.32   \\ 
0.44   \\ 
0.39   \\ 
0.33   \\ 
0.33   \\ 
0.32   \\ 
0.48   \\ 
0.65   \\ 
0.24   \\ 
0.42   \\ 
0.36   \\ 
0.29   \\ 
0.43   \\ 
0.44   \\ 
0.56   \\ 
0.37   \\ 
0.29   \\ 
0.41   \\ 
0.59   \\ 
0.39   \\ 
0.52   \\ 
0.46   \\ 
0.55   \\ 
0.42   \\ 
0.35   \\ 
0.41   \\ 
0.41   \\ 
0.31   \\ 
0.51   \\ 
0.26   \\ 
0.34   \\ 
0.28   \\ 
0.30   \\ 
0.51   \\ 
0.33   \\ 
0.62   \\ 
}node[above=-1pt] at
(boxplot box cs: \boxplotvalue{median},0.5)
{\scalebox{\scalebp}{\tiny\pgfmathprintnumber{\boxplotvalue{median}}}};;
\end{axis}
\end{tikzpicture}}\hspace{0.25cm}
\subfigure[$\hat{V}$]{
\begin{tikzpicture}
\begin{axis}[boxplot/draw direction=y, x = 1cm,  xmajorticks=false,  height = \heightbp]
\addplot+ [boxplot, mark size = \marksize]
table [row sep=\\,y index=0] {
9.37   \\ 
41.24   \\ 
40.40   \\ 
18.70   \\ 
33.33   \\ 
48.27   \\ 
   \\ 
18.77   \\ 
   \\ 
10.57   \\ 
9.21   \\ 
47.97   \\ 
13.23   \\ 
37.09   \\ 
17.16   \\ 
15.88   \\ 
16.95   \\ 
48.13   \\ 
25.67   \\ 
16.18   \\ 
23.85   \\ 
37.14   \\ 
57.91   \\ 
   \\ 
23.06   \\ 
49.69   \\ 
   \\ 
   \\ 
57.64   \\ 
20.71   \\ 
25.61   \\ 
42.26   \\ 
49.54   \\ 
15.91   \\ 
39.59   \\ 
   \\ 
36.59   \\ 
49.49   \\ 
45.03   \\ 
   \\ 
   \\ 
51.38   \\ 
23.40   \\ 
31.79   \\ 
   \\ 
7.09   \\ 
28.32   \\ 
   \\ 
34.63   \\ 
49.82   \\ 
16.93   \\ 
   \\ 
41.29   \\ 
   \\ 
49.52   \\ 
   \\ 
34.64   \\ 
   \\ 
42.37   \\ 
29.17   \\ 
40.47   \\ 
3.73   \\ 
31.57   \\ 
33.20   \\ 
53.38   \\ 
18.85   \\ 
36.93   \\ 
44.98   \\ 
57.45   \\ 
   \\ 
19.69   \\ 
47.76   \\ 
22.30   \\ 
36.23   \\ 
21.82   \\ 
11.92   \\ 
   \\ 
29.62   \\ 
   \\ 
37.19   \\ 
37.07   \\ 
45.28   \\ 
39.92   \\ 
50.01   \\ 
31.08   \\ 
32.43   \\ 
   \\ 
41.18   \\ 
   \\ 
44.16   \\ 
45.41   \\ 
47.55   \\ 
   \\ 
41.33   \\ 
18.89   \\ 
10.69   \\ 
27.29   \\ 
28.29   \\ 
34.15   \\ 
31.54   \\ 
}node[above=-1pt] at
(boxplot box cs: \boxplotvalue{median},0.5)
{\scalebox{\scalebp}{\tiny\pgfmathprintnumber{\boxplotvalue{median}}}};;
\end{axis}
\end{tikzpicture}}\hspace{0.25cm}

\subfigure[$\hat{V}^\pen$]{
\begin{tikzpicture}
\begin{axis}[boxplot/draw direction=y, x = 1cm,  xmajorticks=false,  height = \heightbp]
\addplot+ [boxplot, mark size = \marksize]
table [row sep=\\,y index=0] {
9.19   \\ 
44.01   \\ 
39.27   \\ 
18.67   \\ 
20.75   \\ 
34.62   \\ 
58.99   \\ 
15.87   \\ 
45.84   \\ 
10.98   \\ 
8.84   \\ 
   \\ 
13.85   \\ 
41.09   \\ 
10.96   \\ 
14.47   \\ 
17.89   \\ 
54.63   \\ 
27.33   \\ 
13.57   \\ 
41.56   \\ 
44.53   \\ 
   \\ 
   \\ 
31.86   \\ 
45.13   \\ 
   \\ 
   \\ 
   \\ 
14.77   \\ 
17.37   \\ 
   \\ 
   \\ 
21.02   \\ 
32.84   \\ 
   \\ 
55.75   \\ 
57.06   \\ 
49.64   \\ 
   \\ 
31.09   \\ 
53.39   \\ 
25.75   \\ 
37.41   \\ 
   \\ 
6.86   \\ 
23.39   \\ 
   \\ 
41.09   \\ 
47.86   \\ 
14.61   \\ 
59.27   \\ 
30.91   \\ 
   \\ 
46.76   \\ 
32.94   \\ 
40.90   \\ 
   \\ 
55.27   \\ 
39.50   \\ 
20.38   \\ 
2.88   \\ 
33.67   \\ 
27.40   \\ 
54.69   \\ 
20.36   \\ 
41.65   \\ 
45.86   \\ 
58.79   \\ 
52.96   \\ 
25.99   \\ 
33.98   \\ 
20.32   \\ 
25.34   \\ 
17.58   \\ 
12.74   \\ 
44.79   \\ 
43.41   \\ 
   \\ 
30.03   \\ 
43.21   \\ 
43.54   \\ 
37.38   \\ 
32.19   \\ 
27.83   \\ 
30.45   \\ 
   \\ 
43.47   \\ 
   \\ 
   \\ 
54.89   \\ 
50.52   \\ 
   \\ 
   \\ 
26.10   \\ 
9.90   \\ 
27.03   \\ 
34.47   \\ 
28.32   \\ 
27.59   \\ 
}node[above=-1pt] at
(boxplot box cs: \boxplotvalue{median},0.5)
{\scalebox{\scalebp}{\tiny\pgfmathprintnumber{\boxplotvalue{median}}}};;
\end{axis}
\end{tikzpicture}}\hspace{0.25cm}
\subfigure[$\hat{V}_{\alpha}$]{
\begin{tikzpicture}
\begin{axis}[boxplot/draw direction=y, x = 1cm,  xmajorticks=false,  height = \heightbp, ymin =  0.11]
\addplot+ [boxplot, mark size = \marksize]
table [row sep=\\,y index=0] {
0.30   \\ 
0.51   \\ 
1.64   \\ 
0.47   \\ 
0.16   \\ 
0.35   \\ 
0.37   \\ 
0.28   \\ 
0.79   \\ 
0.15   \\ 
0.05   \\ 
0.44   \\ 
0.27   \\ 
4.30   \\ 
0.15   \\ 
0.16   \\ 
0.41   \\ 
0.89   \\ 
0.59   \\ 
0.31   \\ 
0.38   \\ 
0.57   \\ 
0.59   \\ 
0.34   \\ 
0.45   \\ 
0.41   \\ 
0.91   \\ 
0.72   \\ 
0.86   \\ 
0.63   \\ 
0.23   \\ 
0.65   \\ 
0.60   \\ 
0.17   \\ 
0.62   \\ 
0.55   \\ 
2.22   \\ 
0.32   \\ 
0.36   \\ 
0.57   \\ 
0.61   \\ 
0.62   \\ 
0.45   \\ 
0.53   \\ 
0.79   \\ 
0.18   \\ 
0.48   \\ 
0.75   \\ 
0.96   \\ 
0.71   \\ 
0.46   \\ 
0.39   \\ 
0.54   \\ 
0.62   \\ 
1.24   \\ 
1.75   \\ 
0.46   \\ 
1.04   \\ 
0.23   \\ 
0.55   \\ 
0.09   \\ 
0.06   \\ 
0.78   \\ 
0.30   \\ 
0.48   \\ 
0.41   \\ 
0.26   \\ 
0.31   \\ 
0.26   \\ 
0.29   \\ 
0.68   \\ 
2.32   \\ 
0.22   \\ 
0.93   \\ 
0.29   \\ 
0.36   \\ 
0.96   \\ 
0.60   \\ 
1.05   \\ 
0.46   \\ 
0.37   \\ 
0.45   \\ 
0.35   \\ 
0.35   \\ 
0.29   \\ 
0.49   \\ 
0.60   \\ 
0.38   \\ 
0.40   \\ 
0.31   \\ 
0.47   \\ 
0.37   \\ 
0.95   \\ 
0.27   \\ 
0.83   \\ 
0.19   \\ 
0.38   \\ 
1.72   \\ 
0.50   \\ 
2.27   \\ 
}node[above=-1pt] at
(boxplot box cs: \boxplotvalue{median}-0.02,0.5)
{\scalebox{\scalebp}{\tiny\pgfmathprintnumber{\boxplotvalue{median}}}};;
\end{axis}
\end{tikzpicture}}\hspace{0.25cm}
\subfigure[$\hat{V}_{\alpha}^\pen$]{
\begin{tikzpicture}
\begin{axis}[boxplot/draw direction=y, x = 1cm,   xmajorticks=false,  height = \heightbp, ymin = 0.11]
\addplot+ [boxplot, mark size = \marksize]
table [row sep=\\,y index=0] {
0.29   \\ 
0.71   \\ 
1.86   \\ 
0.48   \\ 
0.17   \\ 
0.35   \\ 
0.38   \\ 
0.22   \\ 
0.86   \\ 
0.14   \\ 
0.05   \\ 
0.63   \\ 
0.29   \\ 
6.28   \\ 
0.45   \\ 
0.15   \\ 
0.29   \\ 
0.82   \\ 
0.52   \\ 
0.47   \\ 
0.43   \\ 
0.60   \\ 
0.55   \\ 
0.34   \\ 
0.31   \\ 
0.33   \\ 
0.81   \\ 
0.74   \\ 
1.08   \\ 
0.52   \\ 
0.22   \\ 
0.56   \\ 
0.57   \\ 
0.15   \\ 
0.42   \\ 
0.56   \\ 
4.37   \\ 
0.25   \\ 
0.36   \\ 
0.64   \\ 
0.82   \\ 
0.52   \\ 
0.42   \\ 
0.55   \\ 
0.73   \\ 
0.12   \\ 
0.45   \\ 
0.46   \\ 
1.07   \\ 
0.72   \\ 
0.51   \\ 
0.36   \\ 
0.49   \\ 
0.67   \\ 
1.24   \\ 
1.87   \\ 
0.41   \\ 
1.17   \\ 
0.25   \\ 
0.74   \\ 
0.12   \\ 
0.06   \\ 
0.61   \\ 
0.31   \\ 
0.42   \\ 
0.51   \\ 
0.37   \\ 
0.25   \\ 
0.26   \\ 
0.23   \\ 
0.59   \\ 
2.05   \\ 
0.17   \\ 
0.95   \\ 
0.34   \\ 
0.19   \\ 
0.54   \\ 
1.00   \\ 
0.85   \\ 
0.41   \\ 
0.29   \\ 
0.38   \\ 
0.30   \\ 
0.38   \\ 
0.28   \\ 
0.49   \\ 
0.82   \\ 
0.36   \\ 
0.45   \\ 
0.32   \\ 
0.54   \\ 
0.33   \\ 
1.02   \\ 
0.40   \\ 
0.64   \\ 
0.19   \\ 
0.40   \\ 
1.00   \\ 
0.35   \\ 
1.96   \\ 
}node[above=-1pt] at
(boxplot box cs: \boxplotvalue{median}-0.08,0.5)
{\scalebox{\scalebp}{\tiny\pgfmathprintnumber{\boxplotvalue{median}}}};;
\end{axis}
\end{tikzpicture}}\hspace{0.25cm}
\subfigure[BR]{
\begin{tikzpicture}
\begin{axis}[boxplot/draw direction=y, x = 1cm,  xmajorticks=false,  height = \heightbp]
\addplot+ [boxplot, mark size = \marksize]
table [row sep=\\,y index=0] {
0.18   \\ 
0.29   \\ 
0.29   \\ 
0.27   \\ 
0.18   \\ 
0.36   \\ 
0.27   \\ 
0.27   \\ 
0.29   \\ 
0.27   \\ 
0.19   \\ 
0.47   \\ 
0.23   \\ 
0.31   \\ 
0.23   \\ 
0.22   \\ 
0.24   \\ 
0.37   \\ 
0.22   \\ 
0.31   \\ 
0.21   \\ 
0.42   \\ 
0.22   \\ 
0.22   \\ 
0.34   \\ 
0.28   \\ 
0.22   \\ 
0.21   \\ 
0.55   \\ 
0.21   \\ 
0.21   \\ 
0.24   \\ 
0.29   \\ 
0.21   \\ 
0.22   \\ 
0.25   \\ 
0.36   \\ 
0.22   \\ 
0.24   \\ 
   \\ 
0.26   \\ 
0.26   \\ 
   \\ 
0.25   \\ 
0.24   \\ 
0.25   \\ 
   \\ 
0.59   \\ 
0.62   \\ 
0.26   \\ 
0.26   \\ 
0.35   \\ 
0.26   \\ 
0.24   \\ 
0.33   \\ 
0.26   \\ 
0.25   \\ 
0.24   \\ 
0.40   \\ 
0.20   \\ 
0.24   \\ 
0.25   \\ 
0.25   \\ 
0.25   \\ 
0.29   \\ 
0.34   \\ 
0.54   \\ 
0.23   \\ 
0.24   \\ 
0.34   \\ 
0.23   \\ 
0.30   \\ 
0.21   \\ 
0.23   \\ 
0.38   \\ 
0.27   \\ 
0.31   \\ 
0.23   \\ 
0.31   \\ 
0.22   \\ 
0.23   \\ 
0.23   \\ 
0.22   \\ 
0.23   \\ 
   \\ 
0.32   \\ 
0.24   \\ 
0.53   \\ 
0.30   \\ 
0.23   \\ 
0.26   \\ 
0.22   \\ 
0.24   \\ 
0.23   \\ 
0.23   \\ 
0.22   \\ 
0.23   \\ 
0.31   \\ 
}node[above=-1pt] at
(boxplot box cs: \boxplotvalue{median},0.5)
{\scalebox{\scalebp}{\tiny\pgfmathprintnumber{\boxplotvalue{median}}}};;
\end{axis}
\end{tikzpicture}}}
    \caption{Boxplots of the performance of all methods with respect to the instance type (2,20,m,10) in the JCDFG.  The diagrams show the distribution of the computation time (in seconds) of a original GNE. We did not include the time when no equilibrium was found. In this regard, the methods (a)-(f) found (100,80,80,100,100,96) equilibria respectively.}
    \label{fig: WhiskerJCDFG}

\end{figure}
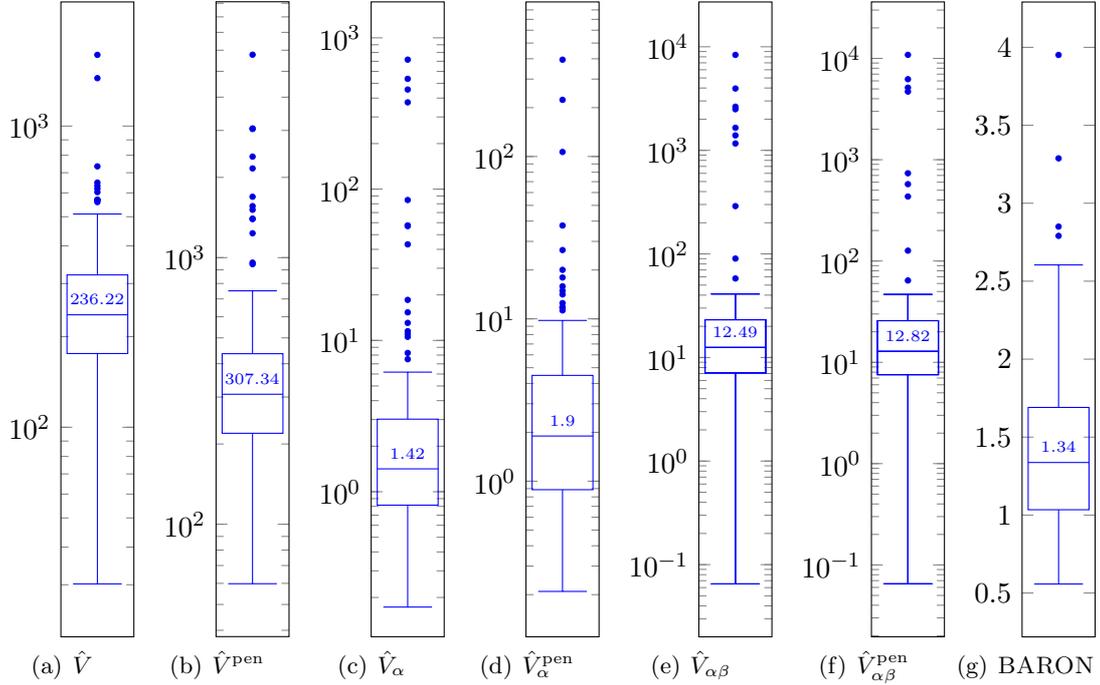
\clearpage

\begin{figure}[h!]
    \centering
    \scalebox{1}{
\subfigure[QL]{
\begin{tikzpicture}
\begin{axis}[boxplot/draw direction=y, x = 01cm,  xmajorticks=false,  height = \heightbp]
\addplot+ [boxplot, mark size = \marksize]
table [row sep =\\, y index=0] {
0.32   \\ 
0.70   \\ 
   \\ 
   \\ 
0.27   \\ 
   \\ 
0.37   \\ 
0.34   \\ 
   \\ 
   \\ 
0.35   \\ 
0.43   \\ 
   \\ 
0.51   \\ 
   \\ 
0.26   \\ 
   \\ 
   \\ 
   \\ 
0.39   \\ 
   \\ 
   \\ 
   \\ 
   \\ 
0.34   \\ 
0.32   \\ 
0.61   \\ 
   \\ 
   \\ 
0.31   \\ 
0.23   \\ 
0.40   \\ 
0.50   \\ 
0.22   \\ 
0.31   \\ 
   \\ 
0.61   \\ 
0.41   \\ 
0.35   \\ 
   \\ 
   \\ 
   \\ 
0.74   \\ 
0.39   \\ 
   \\ 
0.27   \\ 
0.43   \\ 
0.65   \\ 
   \\ 
0.53   \\ 
   \\ 
0.38   \\ 
   \\ 
   \\ 
   \\ 
   \\ 
   \\ 
0.55   \\ 
   \\ 
   \\ 
0.21   \\ 
   \\ 
0.50   \\ 
   \\ 
0.31   \\ 
0.38   \\ 
0.38   \\ 
   \\ 
0.32   \\ 
0.33   \\ 
0.49   \\ 
0.62   \\ 
0.24   \\ 
   \\ 
   \\ 
0.30   \\ 
0.50   \\ 
   \\ 
0.56   \\ 
0.41   \\ 
0.28   \\ 
   \\ 
   \\ 
0.41   \\ 
   \\ 
   \\ 
   \\ 
0.42   \\ 
0.39   \\ 
0.42   \\ 
   \\ 
0.32   \\ 
   \\ 
0.63   \\ 
0.34   \\ 
0.30   \\ 
0.33   \\ 
0.49   \\ 
   \\ 
0.83   \\ 
}node[above=-1pt] at
(boxplot box cs: \boxplotvalue{median},0.5)
{\scalebox{\scalebp}{\tiny\pgfmathprintnumber{\boxplotvalue{median}}}};;
\end{axis}
\end{tikzpicture}}\hspace{0.25cm}
\subfigure[Non-Ex.]{
\begin{tikzpicture}
\begin{axis}[boxplot/draw direction=y, x = 01cm,  xmajorticks=false,  height = \heightbp]
\addplot+ [boxplot, mark size = \marksize]
table [row sep =\\, y index=0] {
   \\ 
   \\ 
   \\ 
0.44   \\ 
   \\ 
   \\ 
   \\ 
   \\ 
   \\ 
   \\ 
   \\ 
   \\ 
0.50   \\ 
   \\ 
   \\ 
   \\ 
0.34   \\ 
0.56   \\ 
0.38   \\ 
   \\ 
0.35   \\ 
   \\ 
   \\ 
0.47   \\ 
   \\ 
   \\ 
   \\ 
   \\ 
   \\ 
   \\ 
   \\ 
   \\ 
   \\ 
   \\ 
   \\ 
   \\ 
   \\ 
   \\ 
   \\ 
   \\ 
   \\ 
   \\ 
   \\ 
   \\ 
0.51   \\ 
   \\ 
   \\ 
   \\ 
   \\ 
   \\ 
   \\ 
   \\ 
   \\ 
0.75   \\ 
   \\ 
   \\ 
   \\ 
   \\ 
0.26   \\ 
   \\ 
   \\ 
0.21   \\ 
   \\ 
0.35   \\ 
   \\ 
   \\ 
   \\ 
   \\ 
   \\ 
   \\ 
   \\ 
   \\ 
   \\ 
0.55   \\ 
   \\ 
   \\ 
   \\ 
   \\ 
   \\ 
   \\ 
   \\ 
0.71   \\ 
0.55   \\ 
   \\ 
0.46   \\ 
0.45   \\ 
0.72   \\ 
   \\ 
   \\ 
   \\ 
0.68   \\ 
   \\ 
   \\ 
   \\ 
   \\ 
   \\ 
   \\ 
   \\ 
0.40   \\ 
   \\ 
}node[above=-1pt] at
(boxplot box cs: \boxplotvalue{median},0.5)
{\scalebox{\scalebp}{\tiny\pgfmathprintnumber{\boxplotvalue{median}}}};;
\end{axis}
\end{tikzpicture}}\hspace{0.25cm}
\subfigure[$\hat{V}$]{
\begin{tikzpicture}
\begin{axis}[boxplot/draw direction=y, x = 1cm,  xmajorticks=false,  height = \heightbp]
\addplot+ [boxplot, mark size = \marksize]
table [row sep=\\,y index=0] {
9.78   \\ 
38.61   \\ 
   \\ 
   \\ 
12.21   \\ 
   \\ 
10.61   \\ 
16.03   \\ 
   \\ 
   \\ 
15.24   \\ 
47.54   \\ 
   \\ 
51.05   \\ 
   \\ 
8.72   \\ 
   \\ 
   \\ 
   \\ 
30.75   \\ 
   \\ 
   \\ 
   \\ 
   \\ 
23.75   \\ 
49.72   \\ 
   \\ 
   \\ 
   \\ 
16.04   \\ 
10.26   \\ 
   \\ 
   \\ 
4.54   \\ 
37.32   \\ 
   \\ 
   \\ 
21.08   \\ 
45.07   \\ 
   \\ 
   \\ 
   \\ 
39.57   \\ 
29.00   \\ 
   \\ 
   \\ 
36.81   \\ 
45.07   \\ 
   \\ 
   \\ 
   \\ 
   \\ 
   \\ 
   \\ 
   \\ 
   \\ 
   \\ 
   \\ 
   \\ 
   \\ 
12.20   \\ 
   \\ 
36.33   \\ 
   \\ 
44.39   \\ 
28.10   \\ 
25.80   \\ 
   \\ 
21.14   \\ 
9.33   \\ 
   \\ 
53.34   \\ 
12.01   \\ 
   \\ 
   \\ 
15.16   \\ 
48.12   \\ 
   \\ 
38.16   \\ 
   \\ 
14.44   \\ 
   \\ 
   \\ 
51.51   \\ 
   \\ 
   \\ 
   \\ 
48.48   \\ 
18.08   \\ 
46.98   \\ 
   \\ 
53.64   \\ 
   \\ 
10.42   \\ 
   \\ 
18.62   \\ 
29.54   \\ 
25.00   \\ 
   \\ 
43.19   \\ 
}node[above=-1pt] at
(boxplot box cs: \boxplotvalue{median},0.5)
{\scalebox{\scalebp}{\tiny\pgfmathprintnumber{\boxplotvalue{median}}}};;
\end{axis}
\end{tikzpicture}}\hspace{0.25cm}

\subfigure[$\hat{V}^\pen$]{
\begin{tikzpicture}
\begin{axis}[boxplot/draw direction=y, x = 1cm,  xmajorticks=false,  height = \heightbp]
\addplot+ [boxplot, mark size = \marksize]
table [row sep=\\,y index=0] {
9.73   \\ 
42.83   \\ 
   \\ 
   \\ 
13.48   \\ 
   \\ 
11.40   \\ 
14.53   \\ 
   \\ 
   \\ 
17.75   \\ 
47.99   \\ 
   \\ 
40.16   \\ 
   \\ 
9.25   \\ 
   \\ 
   \\ 
   \\ 
28.08   \\ 
   \\ 
   \\ 
   \\ 
   \\ 
21.87   \\ 
50.32   \\ 
   \\ 
   \\ 
   \\ 
18.68   \\ 
10.80   \\ 
50.48   \\ 
32.82   \\ 
4.71   \\ 
38.09   \\ 
   \\ 
52.85   \\ 
17.01   \\ 
53.59   \\ 
   \\ 
   \\ 
   \\ 
33.45   \\ 
37.17   \\ 
   \\ 
10.86   \\ 
29.88   \\ 
27.09   \\ 
   \\ 
   \\ 
   \\ 
   \\ 
   \\ 
   \\ 
   \\ 
   \\ 
   \\ 
   \\ 
   \\ 
   \\ 
15.60   \\ 
   \\ 
35.94   \\ 
   \\ 
59.49   \\ 
31.82   \\ 
21.63   \\ 
   \\ 
17.32   \\ 
10.42   \\ 
   \\ 
   \\ 
12.78   \\ 
   \\ 
   \\ 
17.66   \\ 
44.21   \\ 
   \\ 
33.95   \\ 
50.99   \\ 
15.82   \\ 
   \\ 
   \\ 
45.82   \\ 
   \\ 
   \\ 
   \\ 
54.69   \\ 
27.47   \\ 
51.13   \\ 
   \\ 
48.52   \\ 
   \\ 
15.15   \\ 
23.93   \\ 
20.56   \\ 
29.78   \\ 
23.24   \\ 
   \\ 
56.86   \\ 
}node[above=-1pt] at
(boxplot box cs: \boxplotvalue{median},0.5)
{\scalebox{\scalebp}{\tiny\pgfmathprintnumber{\boxplotvalue{median}}}};;
\end{axis}
\end{tikzpicture}}\hspace{0.25cm}
\subfigure[BR]{
\begin{tikzpicture}
\begin{axis}[boxplot/draw direction=y, x = 1cm,  xmajorticks=false,  height = \heightbp]
\addplot+ [boxplot, mark size = \marksize]
table [row sep=\\,y index=0] {
0.25   \\ 
0.26   \\ 
   \\ 
   \\ 
0.26   \\ 
   \\ 
0.25   \\ 
0.35   \\ 
   \\ 
   \\ 
0.25   \\ 
0.32   \\ 
   \\ 
0.26   \\ 
0.41   \\ 
0.26   \\ 
   \\ 
   \\ 
   \\ 
0.35   \\ 
   \\ 
   \\ 
   \\ 
   \\ 
0.27   \\ 
0.31   \\ 
0.24   \\ 
   \\ 
   \\ 
0.24   \\ 
0.24   \\ 
0.24   \\ 
0.31   \\ 
0.22   \\ 
0.22   \\ 
   \\ 
0.24   \\ 
0.23   \\ 
0.23   \\ 
   \\ 
   \\ 
   \\ 
0.40   \\ 
0.23   \\ 
   \\ 
0.23   \\ 
0.23   \\ 
0.31   \\ 
   \\ 
0.32   \\ 
   \\ 
   \\ 
   \\ 
   \\ 
   \\ 
   \\ 
0.24   \\ 
   \\ 
   \\ 
0.15   \\ 
   \\ 
0.23   \\ 
   \\ 
0.24   \\ 
0.24   \\ 
0.24   \\ 
   \\ 
0.24   \\ 
0.22   \\ 
0.24   \\ 
0.33   \\ 
0.23   \\ 
   \\ 
   \\ 
0.23   \\ 
0.25   \\ 
   \\ 
0.25   \\ 
0.32   \\ 
0.24   \\ 
   \\ 
   \\ 
0.24   \\ 
   \\ 
   \\ 
   \\ 
0.23   \\ 
0.23   \\ 
0.31   \\ 
   \\ 
0.24   \\ 
   \\ 
0.24   \\ 
0.24   \\ 
0.24   \\ 
0.23   \\ 
0.32   \\ 
   \\ 
0.25   \\ 
}node[above=-1pt] at
(boxplot box cs: \boxplotvalue{median},0.5)
{\scalebox{\scalebp}{\tiny\pgfmathprintnumber{\boxplotvalue{median}}}};;
\end{axis}
\end{tikzpicture}}}
    \caption{Boxplots of the performance of all methods with respect to the instance type (2,20,m,10) in the ICDFG.  The diagrams show the distribution of the computation time (in seconds) of a original GNE. We did not include the time when no equilibrium was found. In this regard, the quasi-linear approach found 56 equilibria and provided a non-existence certificate in 20 cases. The methods (c)-(e) found (45,50,57) equilibria respectively.}
    \label{fig: WhiskerICDFG}

\end{figure}

\section{Conclusions}
We derived a new characterization of generalized Nash equilibria by convexifying the original instance $I$, leading to a set of more structured convexified instances $\mathcal{I}^\conv$ of the GNEP. 
This convexification approach is very general and thus its relevance is relying on the identification of classes of original instances and corresponding well-behaved convexified instances. We illustrated this by deriving  
for the three problem classes 
of quasi-linear, $k$-restrictive-closed and restrictive-closed GNEPs, respectively, 
new characterizations of the existence and computability of generalized Nash equilibria. We demonstrated the applicability of the latter by presenting various methods and corresponding numerical results for the computation of equilibria in the CDFG {and transportation markets}.
In this regard, our convexification offers an approach to systematically tackle the poorly understood class of non-convex and discrete GNEPs via identifying original and corresponding well-behaved convexified instances in order to then draw conclusions for the original instance from the convexified one via our main Theorem~\ref{thm:main}. Therefore we 
believe that there is still  untapped potential in our convexification method in order to obtain structural insights into the problem as well as pave the way for a more tractable computational approach.

\bibliographystyle{plain}
\bibliography{master-bib}

\clearpage
{
\appendix
\section{Appendix}
\subsection{Detailed Numerical Results}
In Table~\ref{tab: JCDFG} and~\ref{tab: ICDFG} we present the numerical results in a more detailed fashion.
That is, we subdivided the numerical results w.r.t.~each instance type. As described above, for each instance type there are 10 different instances.   
}

\begin{table}[h!]
    \centering
    \resizebox{\columnwidth}{!}{
\begin{tabular}{lcccccccccccccc}
\toprule
& \multicolumn{4}{c}{Quasi-Linear}&\multicolumn{2}{c}{$\hat V$} & \multicolumn{2}{c}{$\hat{V}^\pen$} & \multicolumn{2}{c}{$\hat V_\alpha$} & \multicolumn{2}{c}{$\hat{V}_\alpha^\pen$} &   \multicolumn{2}{c}{BR}  
\\\cmidrule(lr){2-5}  \cmidrule(lr){6-7} \cmidrule(lr){8-9}\cmidrule(lr){10-11}\cmidrule(lr){12-13} \cmidrule(lr){14-15} 
     Instance-type       & GNE & Time & Non-Existence & Time& GNE & Time & GNE & Time  & GNE & Time & GNE & Time  & GNE & Time   \\\midrule
(2,10,s,1) & 10 & 0.14 & 0 & - & 9 & 4.66 & 10 & 7.47 & 10 & 0.33 & 10 & 0.33 & 10 & 0.93 \\ 
(2,10,s,10) & 10 & 0.18 & 0 & - & 10 & 6.70 & 10 & 5.83 & 10 & 0.08 & 10 & 0.08 & 9 & 0.26 \\ 
(2,10,m,1) & 10 & 0.14 & 0 & - & 10 & 10.12 & 10 & 9.13 & 10 & 0.15 & 10 & 0.14 & 10 & 0.28 \\ 
(2,10,m,10) & 10 & 0.24 & 0 & - & 10 & 15.32 & 10 & 13.79 & 10 & 0.12 & 10 & 0.10 & 9 & 0.42 \\ 
(2,15,s,1) & 10 & 0.18 & 0 & - & 10 & 18.91 & 10 & 15.53 & 10 & 0.71 & 10 & 1.13 & 9 & 0.31 \\ 
(2,15,s,10) & 10 & 0.33 & 0 & - & 9 & 17.06 & 9 & 15.11 & 10 & 0.53 & 10 & 0.24 & 8 & 0.33 \\ 
(2,15,m,1) & 10 & 0.17 & 0 & - & 10 & 19.69 & 10 & 18.72 & 10 & 0.24 & 10 & 0.20 & 10 & 0.90 \\ 
(2,15,m,10) & 10 & 0.37 & 0 & - & 10 & 33.74 & 10 & 30.78 & 10 & 0.53 & 10 & 0.32 & 10 & 0.32 \\ 
(2,20,s,1) & 10 & 0.20 & 0 & - & 9 & 15.98 & 9 & 13.02 & 10 & 2.02 & 10 & 1.68 & 9 & 0.90 \\ 
(2,20,s,10) & 10 & 0.44 & 0 & - & 9 & 26.75 & 7 & 23.28 & 10 & 0.63 & 10 & 1.44 & 9 & 0.32 \\ 
(2,20,m,1) & 10 & 0.20 & 0 & - & 9 & 28.36 & 9 & 26.75 & 10 & 1.93 & 10 & 0.32 & 10 & 2.26 \\ 
(2,20,m,10) & 10 & 0.44 & 0 & - & 8 & 27.58 & 10 & 29.82 & 10 & 0.50 & 10 & 0.55 & 10 & 0.27 \\ 
(4,10,s,1) & 10 & 0.22 & 0 & - & 9 & 13.74 & 10 & 14.66 & 10 & 0.32 & 10 & 0.20 & 10 & 2.00 \\ 
(4,10,s,10) & 10 & 0.68 & 0 & - & 9 & 10.77 & 10 & 14.96 & 10 & 0.40 & 10 & 0.34 & 10 & 1.67 \\ 
(4,10,m,1) & 10 & 0.21 & 0 & - & 9 & 14.08 & 9 & 12.76 & 10 & 1.44 & 10 & 2.55 & 9 & 0.75 \\ 
(4,10,m,10) & 10 & 0.62 & 0 & - & 8 & 24.38 & 8 & 21.75 & 10 & 0.80 & 10 & 0.64 & 8 & 4.93 \\ 
(4,15,s,1) & 10 & 0.37 & 0 & - & 5 & 36.69 & 4 & 39.25 & 10 & 6.24 & 10 & 4.67 & 8 & 1.61 \\ 
(4,15,s,10) & 10 & 2.08 & 0 & - & 0 & - & 0 & - & 10 & 5.71 & 10 & 5.75 & 9 & 3.41 \\ 
(4,15,m,1) & 9 & 0.37 & 0 & - & 4 & 32.73 & 6 & 33.66 & 9 & 7.03 & 9 & 2.95 & 10 & 1.44 \\ 
(4,15,m,10) & 10 & 1.64 & 0 & - & 2 & 32.53 & 2 & 29.52 & 10 & 2.34 & 10 & 4.21 & 9 & 5.00 \\ 
(4,20,s,1) & 10 & 0.47 & 0 & - & 4 & 34.85 & 4 & 33.88 & 10 & 4.96 & 10 & 1.84 & 6 & 0.91 \\ 
(4,20,s,10) & 10 & 2.35 & 0 & - & 2 & 17.67 & 2 & 16.03 & 10 & 13.06 & 10 & 11.34 & 9 & 0.74 \\ 
(4,20,m,1) & 10 & 0.42 & 0 & - & 5 & 42.39 & 6 & 39.76 & 10 & 4.02 & 10 & 2.74 & 8 & 0.72 \\ 
(4,20,m,10) & 9 & 2.31 & 0 & - & 1 & 25.39 & 1 & 26.66 & 10 & 5.64 & 10 & 6.30 & 7 & 6.24 \\ 
(10,10,s,1) & 8 & 0.73 & 0 & - & 4 & 29.08 & 5 & 37.21 & 9 & 4.71 & 9 & 4.51 & 9 & 2.04 \\ 
(10,10,s,10) & 10 & 7.34 & 0 & - & 2 & 10.92 & 2 & 11.14 & 8 & 17.21 & 8 & 11.96 & 8 & 3.12 \\ 
(10,10,m,1) & 10 & 0.70 & 0 & - & 7 & 36.48 & 6 & 25.91 & 10 & 5.40 & 10 & 7.92 & 9 & 2.13 \\ 
(10,10,m,10) & 10 & 7.07 & 0 & - & 1 & 5.41 & 1 & 4.74 & 8 & 14.77 & 7 & 6.55 & 7 & 6.49 \\ 
(10,15,s,1) & 10 & 1.69 & 0 & - & 0 & - & 0 & - & 7 & 9.17 & 8 & 8.50 & 5 & 9.35 \\ 
(10,15,s,10) & 10 & 15.08 & 0 & - & 0 & - & 0 & - & 9 & 17.27 & 8 & 13.07 & 8 & 1.89 \\ 
(10,15,m,1) & 7 & 1.60 & 0 & - & 0 & - & 0 & - & 6 & 14.20 & 6 & 9.33 & 8 & 2.25 \\ 
(10,15,m,10) & 10 & 33.70 & 0 & - & 0 & - & 0 & - & 2 & 6.79 & 2 & 20.22 & 7 & 3.86 \\ 
(10,20,s,1) & 6 & 1.98 & 0 & - & 1 & 39.26 & 1 & 39.97 & 6 & 13.15 & 6 & 14.75 & 9 & 2.19 \\ 
(10,20,s,10) & 8 & 29.80 & 0 & - & 0 & - & 0 & - & 4 & 25.34 & 4 & 28.61 & 9 & 7.49 \\ 
(10,20,m,1) & 4 & 1.81 & 0 & - & 0 & - & 0 & - & 6 & 12.52 & 7 & 9.05 & 9 & 2.49 \\ 
(10,20,m,10) & 10 & 26.98 & 0 & - & 0 & - & 0 & - & 5 & 20.31 & 5 & 20.32 & 8 & 2.35 \\ 
\hline 
\end{tabular}} \caption{The performances of the various methods applied to the JCDFG. 
The ``GNE'' column of a method displays how often an equilibrium was found while the ``Time'' column shows how long it took (in seconds) to compute the equilibrium on average. The ``Non-Existence'' column shows how often \textsc{BARON} was able to give a lower bound bigger zero on the objective, i.e.~giving a certificate for non existence of equilibria and the corresponding ``Time'' column shows how long it took (in seconds) to compute the lower bound.}
\label{tab: JCDFG}
\end{table}

 \begin{table}[h!]
     \centering
     \scalebox{0.92}{
\begin{tabular}{lccccccccccc}\toprule
 & \multicolumn{4}{c}{Quasi-Linear} &  \multicolumn{2}{c}{$\hat V$} & \multicolumn{2}{c}{$\hat{V}^\pen$} & \multicolumn{2}{c}{BR} 
\\\cmidrule(lr){2-5}  \cmidrule(lr){6-7}\cmidrule(lr){8-9} \cmidrule(lr){10-11}
     Instance-type & GNE & Time & Non-Existence & Time & GNE & Time & GNE & Time  & GNE & Time  \\\midrule
(2,10,s,1) & 10 & 0.19 & 0 & - & 10 & 6.91 & 10 & 7.91 & 10 & 4.08 \\ 
(2,10,s,10) & 7 & 0.18 & 0 & - & 7 & 9.34 & 7 & 5.61 & 7 & 0.24 \\ 
(2,10,m,1) & 9 & 0.14 & 0 & - & 9 & 11.39 & 9 & 10.09 & 9 & 0.35 \\ 
(2,10,m,10) & 4 & 0.18 & 3 & 0.23 & 4 & 13.07 & 4 & 13.28 & 4 & 0.25 \\ 
(2,15,s,1) & 10 & 0.18 & 0 & - & 9 & 18.74 & 9 & 18.70 & 9 & 1.77 \\ 
(2,15,s,10) & 9 & 0.30 & 0 & - & 7 & 23.70 & 7 & 24.66 & 8 & 2.73 \\ 
(2,15,m,1) & 10 & 0.16 & 0 & - & 10 & 20.20 & 10 & 18.68 & 10 & 3.02 \\ 
(2,15,m,10) & 7 & 0.31 & 1 & 0.38 & 6 & 16.95 & 6 & 16.32 & 7 & 0.31 \\ 
(2,20,s,1) & 10 & 0.20 & 0 & - & 9 & 16.40 & 9 & 14.82 & 10 & 5.09 \\ 
(2,20,s,10) & 6 & 0.32 & 1 & 0.48 & 7 & 22.20 & 7 & 19.62 & 8 & 0.27 \\ 
(2,20,m,1) & 10 & 0.20 & 0 & - & 10 & 23.56 & 10 & 25.71 & 10 & 1.10 \\ 
(2,20,m,10) & 5 & 0.40 & 1 & 0.44 & 5 & 17.45 & 5 & 18.39 & 5 & 0.28 \\ 
(4,10,s,1) & 9 & 0.23 & 0 & - & 9 & 21.20 & 9 & 26.56 & 8 & 0.66 \\ 
(4,10,s,10) & 8 & 0.35 & 0 & - & 7 & 9.08 & 7 & 8.63 & 9 & 0.55 \\ 
(4,10,m,1) & 7 & 0.20 & 2 & 0.54 & 7 & 16.16 & 5 & 15.15 & 7 & 0.54 \\ 
(4,10,m,10) & 5 & 0.49 & 0 & - & 4 & 12.02 & 5 & 21.45 & 5 & 0.52 \\ 
(4,15,s,1) & 6 & 0.35 & 0 & - & 2 & 42.98 & 2 & 43.80 & 7 & 2.09 \\ 
(4,15,s,10) & 2 & 1.35 & 0 & - & 0 & - & 0 & - & 5 & 0.85 \\ 
(4,15,m,1) & 8 & 0.42 & 0 & - & 3 & 47.19 & 2 & 27.11 & 7 & 2.09 \\ 
(4,15,m,10) & 1 & 1.02 & 1 & 0.67 & 1 & 5.40 & 1 & 5.71 & 1 & 0.67 \\ 
(4,20,s,1) & 7 & 0.46 & 0 & - & 3 & 36.39 & 3 & 37.67 & 8 & 1.11 \\ 
(4,20,s,10) & 4 & 1.76 & 0 & - & 0 & - & 0 & - & 4 & 0.61 \\ 
(4,20,m,1) & 9 & 0.42 & 0 & - & 3 & 45.70 & 2 & 35.60 & 8 & 0.70 \\ 
(4,20,m,10) & 2 & 1.10 & 0 & - & 0 & - & 1 & 55.62 & 2 & 0.59 \\ 
(10,10,s,1) & 3 & 0.83 & 1 & 5.21 & 2 & 19.11 & 4 & 32.08 & 4 & 1.93 \\ 
(10,10,s,10) & 4 & 8.24 & 0 & - & 1 & 21.72 & 2 & 35.24 & 4 & 1.73 \\ 
(10,10,m,1) & 3 & 0.73 & 3 & 0.61 & 1 & 12.12 & 2 & 23.98 & 3 & 1.93 \\ 
(10,10,m,10) & 2 & 3.89 & 0 & - & 0 & - & 1 & 43.77 & 2 & 2.04 \\ 
(10,15,s,1) & 4 & 1.82 & 0 & - & 0 & - & 0 & - & 4 & 4.17 \\ 
(10,15,s,10) & 5 & 11.89 & 0 & - & 0 & - & 0 & - & 5 & 1.80 \\ 
(10,15,m,1) & 2 & 2.15 & 0 & - & 0 & - & 0 & - & 2 & 14.36 \\ 
(10,15,m,10) & 1 & 39.69 & 1 & 5.11 & 0 & - & 0 & - & 1 & 1.80 \\ 
(10,20,s,1) & 5 & 2.08 & 0 & - & 1 & 54.64 & 0 & - & 5 & 2.01 \\ 
(10,20,s,10) & 5 & 39.86 & 0 & - & 0 & - & 1 & 40.44 & 6 & 2.01 \\ 
(10,20,m,1) & 1 & 1.73 & 0 & - & 1 & 45.32 & 1 & 33.58 & 1 & 2.26 \\ 
(10,20,m,10) & 4 & 14.06 & 0 & - & 0 & - & 0 & - & 4 & 1.78 \\ 
\hline
\end{tabular}}
     \caption{The performances of the various methods applied to the ICDFG. 
The ``GNE'' column of a method displays how often an equilibrium was found while the ``Time'' column shows how long it took (in seconds) to compute the equilibrium on average. The ``Non-Existence'' column shows how often \textsc{BARON} was able to give a lower bound bigger zero on the objective, i.e.~giving a certificate for non existence of equilibria and the corresponding ``Time'' column shows how long it took (in seconds) to compute the lower bound.}
     \label{tab: ICDFG}
 \end{table}

\end{document}